\documentclass[11pt]{article}

\usepackage[margin=1in]{geometry}
\usepackage{float}
\usepackage{adjustbox}
\usepackage{booktabs}
\usepackage{xcolor}
\usepackage{amsmath}
\usepackage{amssymb}
\usepackage{caption}
\usepackage{amsthm}
\usepackage{subcaption}
\usepackage{mathtools}
\usepackage{comment}
\usepackage{soul}
\usepackage{hyperref}
\usepackage{authblk}

\theoremstyle{definition}
\newtheorem{assumption}{Assumption}[section]

\newtheorem{remark}{Remark}[section]

\newtheorem{corollary}{Corollary}[section]

\newtheorem{proposition}{Proposition}[section]

\newcommand{\ct}[2]{#1\,{\scriptsize(#2)}}
\newcommand{\na}{\multicolumn{1}{c}{--}}
\newcommand{\grp}[1]{\multicolumn{3}{@{}l}{\textit{#1}}}

\usepackage[
    backend=biber,
    style=apa,
    natbib=true,
    uniquename=false,
    url=false,
    isbn=false,
    doi=true,
    eprint=false,
]{biblatex}
\title{The fine structure of electricity price volatility}
\author[1,2,3]{Thomas K. Kloster}
\author[2]{Fred Espen Benth}

\affil[1]{Department of Economics and Business Economics, Aarhus University}
\affil[2]{Department of Data Science and Analytics, BI Norwegian Business School}
\affil[3]{CoRE, Center for Research in Energy: Economics and Markets}

\begin{document}
\maketitle

\begin{abstract}
    We conduct the first rigorous study of electricity price volatility for the full panel of electricity prices across three European generation zones. By interpreting the observed day-ahead prices as local averages of a latent price process governed by a stochastic partial differential equation, we develop estimators of the weekly integrated variance. The inherently infinite dimensional setting introduce several complications that are not relevant in the conventional finite dimensional semimartingale setting, and we spend considerable effort in dealing with these. In particular, we must account for both mean-reversion in prices and semigroup-smoothing in the estimated variance. We provide a detailed decomposition and interpretation of the empirical estimates across three vastly different European generation zones, namely Germany, Norway, and Spain. Our findings indicate that each zone has very different drivers of volatility, and that the impact of generation variables differs considerably. We document that leverage effects appear to be present at first sight, but disappear once we condition on suitable state variables, thereby showing that electricity price volatility does not generally exhibit asymmetric responses to price shocks. 
\end{abstract}


\section{Introduction}
Volatility is a ubiquitous term in finance and typically refers to the second order structure of the conditional return distribution. It is generally regarded as being the most dominant time-varying feature of the distribution and plays a critical role in, e.g., portfolio management, derivatives pricing, and trading impact \citep{AndersenBollerslev1998,AndersenBollerslevDieboldEbens2001,AndersenBollerslevDieboldLabys2003,BNS2002_RV,ChristensenKinnebrockPodolskij2010,ChristensenOomenPodolskij2010,ChristensenPodolskij2012,ChristoffersenJacobsMimouni2010,FrenchSchwertStambaugh1987,RealizedRange2008}. Depending on the horizon and data availability, different proxies for the second order structure can be considered, but most modern measures rely on quadratic variation theory.

In this work, the goal is to investigate the conditional second order structure in European day-ahead (spot) electricity markets in a systematic and rigorous way, down to the individual delivery period. While there is a large literature on stylized facts and forecasting of spot price \emph{levels}, there seems to be less focus on the \emph{volatility} of the panel of prices. The market for electricity is very different from classical equity markets and it therefore requires different interpretations and modeling tools \citep[see][for comprehensive introductions]{BenthMonograph,Weron2014}. The most notable feature of electricity spot prices is that they evolve as a high-dimensional multivariate time series $P_t\in\mathbb{R}^d$, where each component $P_t^{(i)}$ represents the price of consuming 1MWh of electricity during the $i$'th \emph{delivery period}. Typically, it is the case that $d=24$ or $d=96$, which, respectively, corresponds to hourly and quarter-hourly delivery periods throughout the day. Unlike stock prices, electricity spot prices can also be negative, which means that log-transformations are not generally feasible, as is otherwise common in the literature on volatility. Other unique features of electricity spot prices that distinguish them from traditional asset prices are the presence of mean-reversion, short-lived spikes, and seasonality effects. Furthermore, the cross-sectional dependence structure of electricity prices exhibits a very particular diurnal pattern as shown in \citet{Kloster2026}; in particular, electricity spot prices exhibit ``cyclicality'' in the sense that the earliest and the latest delivery period of the day are more highly correlated than, e.g., the earliest and the mid-day delivery periods. This is generally ascribed to the strong correlation of the underlying supply and demand at these periods, which is largely agnostic to when the price is determined. 

These characteristics can be modeled by imposing suitable assumptions within a multivariate model, but it is illustrated in \citet{Kloster2026} that one obtains a flexible and parsimonious framework by letting $P_t$ evolve as a function-valued process, such that the daily cross section of prices corresponds to observations along a latent price curve. More precisely, it is shown that when $P_t$ takes values in the separable Hilbert space of square-integrable functions on the unit circle, denoted by $L^2(\mathbb{S}^1)$ where $\mathbb{S}^1$ is the unit circle, then cyclicality of prices is essentially a model-free implication. The functional time series perspective is not new in the context of electricity spot markets \citep[see, e.g.,][]{EPF_functional1,EPF_functional2,EPF_functional3}, but the choice of $\mathbb{S}^1$ as an index for the cross sectional dimension introduced in \citet{Kloster2026} yields a particularly tractable framework in which we shall operate. The infinite dimensional perspective have several convenient implications that cannot be modelled by conventional methods, in addition to the diurnal periodic structure induced in the prices from the embedding into $L^2(\mathbb{S}^1)$. It allows for modeling arbitrary and changing partitions of the day into delivery periods, which is relevant as these are not fixed throughout time or generation zone in practice. It also allows for the interpretation of spot prices as functionals of the latent daily price curve, which is a perspective we adopt throughout. Furthermore, the interpretation of both time and cross section as a continuum is very tractable for the purpose of derivatives pricing, which is an area where infinite dimensional models have generally been used to great effect \citep[][]{BenthParaschiv2018,AmbitFutures2014,CuchieroEnergy2024}, and an area where volatility modeling is particularly important.

To illustrate our framework, suppose that $\lbrace [h_{i-1},h_{i})\rbrace_{i=1}^{d}$ is the partition of delivery periods where $h_{0}=0<h_{1}<\cdots <h_{d}=2\pi$. Each $h_i$ represents a unique time of the day via the canonical mapping onto the circle $h_{i}\mapsto(\cos(h_i),\sin(h_i))=(\cos (i/(2\pi d)),\sin(i/(2\pi d)))$. Let $X_t(h)$ be the instantaneous latent price of electricity on time $h$ of day $t$ and suppose that the function-valued process $X_t(\cdot)$  
takes values in $L^2(\mathbb{S}^1)$. Since electricity can be consumed at any point during the delivery period, it is economically reasonable that the \emph{observed} price, $P_t^{(i)}$, corresponds to the average of the true, but unobserved, price over that period. More precisely, we find it natural to impose that
\begin{equation}\label{eq:price_local_average}
P_{t}^{(i)}=\frac{1}{(h_i-h_{i-1})}\int_{h_{i-1}}^{h_i}X_t(h)dh.
\end{equation}
In this view, each observed price $P_{t}^{(i)}$ is a bounded linear functional of $X_t$, which has several convenient statistical implications. In particular, we do not a priori require the latent process $X_t$ to take values in a ``nice'' subspace of $L^2(\mathbb{S}^1)$, which is often required to make the evaluation functional $\delta_x f=f(x)$ continuous. By formulating the evolution of $X_t(\cdot)$ as a stochastic partial differential equation (SPDE) in $L^2(\mathbb{S}^1)$, we may construct an estimate of the integrated variance -- which is an operator on $L^2(\mathbb{S}^1)$ -- via methods related to quadratic variation theory. Our methods and results are similar to those of \citet{BenthSchroersVeraart2022,BenthSchroersVeraart2024}, who study realized covariation for a class of semilinear SPDE's in a high frequency setting. They introduce the so-called semigroup-adjusted realized covariation, which generalizes the multivariate realized covariation of \citet{BNS_2004_RCV}. Their semigroup adjustment is a natural consequence of the functional setting, and must be accounted for to ensure that certain propagation effects disappear in the infill limit. A similar adjustment turns out to play a role in our long-span setting, and we discuss the relation to the high frequency case as well as classical realized covariation in detail.

To the best of our knowledge, the literature on volatility estimation in electricity markets is rather sparse, likely due to the complexity of the task. An approach for volatility estimation at the level of individual delivery periods is presented in \citet{Erdogu2016}, where the price in each delivery period is treated as a separate and independent market. Based on our integrated variance estimator, we construct estimates of the weekly integrated variance in several European markets that, in contrast, explicitly and non-parametrically account for the strong dependence between individual prices and we study how the daily ``term structures'' of volatility behave via principal components analysis. We document that so-called propagation effects such as price mean-reversion make up more than $40\%$ of the total price variation and it is thus important to handle these effects in a robust manner. We show that our realized covariation estimator correlates highly and positively with the price level throughout several generation zones, but that the correlation with generation variables (such as wind or solar power generation) varies substantially with the generation zone. We also consider the leverage effect, which has occasionally been described as an \emph{inverse} leverage effect in electricity markets, since high spot prices are associated to high volatility. We show that there is evidence of an unconditional \emph{weak form} of inverse leverage, but that the asymmetric effect between increases and decreases in the price are largely explained by the price level itself and mean-reversion in volatility. 

The rest of the paper is structured as follows. In Section~\ref{sec:model}, we set up the model and construct the realized covariation estimator. Section~\ref{sec:rcv_of_prices} considers the resulting volatility estimates from data on three generation zones and how they relate to several structural variables. Section~\ref{sec:decomposing_rcv} contains more detailed analyses of some aspects of the realized covariation, Section~\ref{sec:inverse_leverage} studies the inverse leverage effect in detail, and Section~\ref{sec:conclusion} concludes.

\section{Model formulation}\label{sec:model}
We shall model the daily cross section of latent prices $X_t$ as the mild solution to an SPDE of the form \citep[see, e.g., ][for a comprehensive treatment of such SPDE's]{DaPratoZabczyk2014}
\begin{equation}\label{eq:model}
dX_t = \mu_t dt + \mathcal{A}X_tdt + \sigma_t dW_t,
\end{equation}
where $W_t$ is a cylindrical Wiener process on $L^{2}(\mathbb{S}^1)$ generating a filtration $\mathcal{F}_t=\sigma(W_{t}:t\geq 0)$ and $\mathcal{A}$ is an unbounded operator on $\mathbb{S}^1$ generating a $C_0$-semigroup $\mathcal{S}$. The process $\mu_t$ represents a stochastic drift with state-space $L^2(\mathbb{S}^1)$ and we impose that this is $\mathcal{F}_t$-predictable and Bochner integrable. The process $\sigma_t$ represents the instantaneous volatility process and we impose that this is $\mathcal{F}_t$-predictable and takes values in the space of Hilbert-Schmidt operators on $L^2(\mathbb{S}^1)$. As we are interested in characterizing the second order structure of prices, we impose in all of the following that $\sup_{t\in [0,T]}\mathbb{E}\left[\lVert X_t \rVert^2_{L^2(\mathbb{S}^1)}\right]<\infty$, where $T>0$ is the maximum time horizon of interest. Ignoring the drift for simplicity, a sufficient condition is that 
\begin{equation}\label{eq:sigma_second_moment}
\mathbb{E}\left[ \int_{0}^{T}\lVert \mathcal{S}(t-s)\sigma_s\rVert_{HS}^2\;ds\right]<\infty,
\end{equation}
where $\lVert \cdot \rVert_{HS}$ denotes the Hilbert-Schmidt norm. The mild solution to the SPDE \eqref{eq:model} can be written in terms of the semigroup $\mathcal{S}$ as 
\begin{equation}\label{eq:mild_formulation}
    X_t = \mathcal{S}(t)X_{0} + \int_{0}^{t}\mathcal{S}(t-s)\mu_sds + \int_{0}^{t}\mathcal{S}(t-s)\sigma_s dW_s.
\end{equation}
If $\mathcal{S}$ admits an integral kernel $p(t-s,x,y)$ such that $(\mathcal{S}(t)f)(x)=\int p(t,x,y)f(y)dy$, then the price process $X_t$ has the martingale measure representation
\begin{equation}\label{eq:martingale_formulation}
    X_{t}(h) = \int_{\mathbb{S}^1}p(t,x,y)X_{0}(y)dy + \int_{0}^{t}\int_{\mathbb{S}^1}p(t-s,x,y)\mu_s(y)dyds + \int_{0}^{t}\int_{\mathbb{S}^1}p(t-s,x,y)\sigma_s (y)W(ds,dy),
\end{equation}
where $W$ is a standard white noise on $\mathbb{R}_+\times \mathbb{S}^1$ of \citep[in the sense of][]{Walsh1986}. Note that $W$ corresponds to a Gaussian Lévy basis on $\mathbb{R}_+\times \mathbb{S}^1$ with zero mean, unit variance and control measure equal to the canonical volume measure on $\mathbb{R}_+\times \mathbb{S}^1$. The martingale measure formulation \eqref{eq:martingale_formulation} therefore shows that the model is a special case of the time-stationary ambit field proposed in \citet{Kloster2026}, with a particular choice of kernel function. In contrast to the ambit field framework, we have much less freedom to choose $p(t,x,y)$, but we note that $p$ can still exhibit temporal singularities, which have been argued for in electricity spot prices in \citet{Bennedsen2017}.

The semigroup $\mathcal{S}$ is a model ingredient which should be specified. It may be known in certain cases, such as in the modeling of forward rates or electricity futures, where the shift semigroup arises naturally. In our case of electricity spot prices, it is not obvious what $\mathcal{S}$ should be and even a classical and tractable choice such as the heat semigroup is only known up to a scaling constant. In the following, we shall see that our measure of realized covariation is, in the best case, only identifiable up to smoothing by $\mathcal{S}$. This does not necessarily constitute a problem, but it changes the interpretation of the measured volatility proxy. In future applications, it will be interesting to study appropriate choices of $\mathcal{S}$ in more detail. 

\begin{remark}\label{rem:OU_1dim}
    The model \eqref{eq:model} generalizes Ornstein-Uhlenbeck type models with stochastic volatility. Indeed, if the spatial dimension is collapsed to a single point, then we are in a standard one-dimensional setting and the only possible choice of $C_0$-semigroup is of the form $\mathcal{S}(t)x=e^{\lambda t}x$ with generator $\mathcal{A}x=\lambda x$ for some $\lambda\in\mathbb{R}$. Hence the model takes the familiar form
    \[
    dX_t = \mu_t dt + \lambda X_tdt + \sigma_t dW_t,
    \]
    where $W_t$ is a standard Brownian motion. 
\end{remark}

\subsection{Volatility estimation}\label{sec:vol_estimation}
Since electricity prices are determined and revealed at a daily frequency, this imposes a fixed upper bound on the sampling frequency. We can therefore not rely on high-frequency infill type results as developed in \citet{BenthSchroersVeraart2022,BenthSchroersVeraart2024}, but are instead constrained to a long-span setting where we cannot expect the discretization error of the SPDE \eqref{eq:model} to vanish. In the following, whenever $A$ is a finite dimensional vector or matrix, we denote by $\lVert A \rVert$ the Euclidean norm on $\mathbb{R}^d$ or the Frobenius norm of $A$, and we denote by $\Sigma_t=\sigma_t\sigma_t^\ast$ the instantaneous variance operator. 

Under the local average interpretation of prices in \eqref{eq:price_local_average}, the individual daily observations are of the form $\widetilde{X}_t^{(i)}$, where
\[
\widetilde{X}^{(i)}_t = \langle X_t,g_i\rangle_{L^2(\mathbb{S}^1)}, \quad g_{i}=\frac{1}{(h_{i}-h_{i-1})}\mathbf{1}_{[h_{i-1},h_i)}.
\]
Suppose that we have $0=t_{0}<t_{1}<\cdots t_{N}$ daily cross sectional observations and collect these in vectors $\widetilde{X}_{t_n}=(\widetilde{X}_{t_n}^{(1)},\ldots ,\widetilde{X}_{t_n}^{(d)})^\top$ and define
\[
\Delta_{n}\widetilde{X} := \widetilde{X}_{t_{n}}-\widetilde{X}_{t_{n-1}}\in \mathbb{R}^d.
\]
\begin{proposition}\label{prop:increment_decomp}
    Define the observation operator
    $A:L^2(\mathbb{S}^1)\to \mathbb{R}^d$ by
    \[
    (A x)_i := \langle x,g_i\rangle_{L^2(\mathbb{S}^1)},
    \]
    such that the observed vector of local averages is $\widetilde{X}_t=AX_t$ with $X_t$ as in \eqref{eq:model}. Then it holds that
    \begin{equation}\label{eq:final-matrix}
    \begin{aligned}
    \mathbb E\left[(\Delta_n\widetilde X)(\Delta_n\widetilde X)^\top\mid \mathcal{F}_{t_{n-1}}\right]
    &=
    \mathbb E\big[A(B_n+D_n)(B_n+D_n)^\ast A^\ast\mid \mathcal{F}_{t_{n-1}}\big] \\
    &\quad \quad +
    A\mathbb E\left[\int_{t_{n-1}}^{t_{n}} \mathcal{S}(t_n-s)\Sigma_s \mathcal{S}(t_n-s)^\ast ds\mid\mathcal{F}_{t_{n-1}}\right]A^\ast, 
    \end{aligned}
    \end{equation}
    where $B_n,D_n$ are, respectively, the propagation and the drift terms
    \[
        B_n = (S(\delta)-I)X_{t_{n-1}}, \quad D_n = \int_{t_{n-1}}^{t_{n}}\mathcal{S}(t_{n} -s)\mu_sds.
    \]
    Equivalently, for $i,j\in\{1,\dots,d\}$, we have
    \begin{align*}
    \mathbb E\big[\Delta_n\widetilde X^{(i)}\Delta_n\widetilde X^{(j)}\big]
    &=
    \mathbb E\big[\langle B_n+D_n,g_i\rangle_{L^2(\mathbb{S}^1)}\,\langle B_n+D_n,g_j\rangle_{L^2(\mathbb{S}^1)}\big]
    \\
    &\quad+
    \mathbb E\left[\int_{t_{n-1}}^{t_n}
    \Big\langle \Sigma_s\,\mathcal{S}(t_n-s)^\ast g_j,\; \mathcal{S}(t_n-s)^\ast g_i\Big\rangle_{L^2(\mathbb{S}^1)}\,ds\right].
    \end{align*}
\end{proposition}

Since observations are fixed at a daily frequency, we can set $\delta=(t_n-t_{n-1})$ and readily define the annualized realized covariation over some window of $w$ days as
\begin{equation}\label{eq:RCV_estimator}
\widehat{\Sigma}_j^w = \frac{1}{\delta w}\sum_{\ell = 1}^{w}(\Delta_{(j-1)w+\ell}\widetilde{X})(\Delta_{(j-1)w+\ell}\widetilde{X})^\top.
\end{equation}
When the period length $w$ is implicit, we shall simply refer to the estimator $\widehat{\Sigma}_j^w$ in \eqref{eq:RCV_estimator} as the RCV as shorthand. An application of Proposition~\ref{prop:increment_decomp}, the tower property, and a change of variables yields that
\begin{equation}\label{eq:RCV_interpretation}
\begin{aligned}
\mathbb{E}\left[ \widehat{\Sigma}_j^w\mid \mathcal{F}_{t_{w(j-1)}} \right] &= \frac{1}{\delta w}\sum_{\ell = 1}^{w}\mathbb{E}\left[ (\Delta_{w(j-1)+\ell}\widetilde{X})(\Delta_{w(j-1)+\ell} \widetilde{X})^\top\mid \mathcal{F}_{t_{w(j-1)}}\right] \\
&= \frac{1}{\delta w}\sum_{\ell=1}^{w}A\mathbb{E}\big[ ((B_n+D_n)(B_n+D_n)^\ast)\mid\mathcal{F}_{t_{w(j-1)}}\big]A^\ast \\
&\quad\quad + \frac{1}{\delta w}A \mathbb{E}\left[ \int_{0}^{\delta}\mathcal{S}(u)\left(\sum_{\ell=1}^{w}\Sigma_{w(t_j-1)+\ell-u}\right)\mathcal{S}(u)^\ast du \mid  \mathcal{F}_{t_{w(j-1)}}\right]A^\ast.
\end{aligned}
\end{equation}
This has the intuitive interpretation that the annualized realized covariation over the next $w$ days is the average (conditionally) expected propagation and drift bias, plus the expected semigroup-weighted average variance over the period. As such, the realized covariation \eqref{eq:RCV_estimator} is a noisy proxy of the conditional average variance, in complete analogy to the finite dimensional setting, but with added bias coming from the semigroup. Supposing that all drift and propagation bias can be removed (i.e., $B_n=D_n=0$), we end up with an estimate of the form
\[
\mathbb{E}\left[ \widehat{\Sigma}_j^w\mid \mathcal{F}_{t_{w(j-1)}} \right] = \frac{1}{\delta w}A \mathbb{E}\left[ \int_{0}^{\delta}\mathcal{S}(u)\left(\sum_{\ell=1}^{w}\Sigma_{w(t_j-1)+\ell-u}\right)\mathcal{S}(u)^\ast du \mid  \mathcal{F}_{t_{w(j-1)}}\right]A^\ast.
\]
This estimate does not generally allow us to extract the raw expected average of $\Sigma_t$ due to the semigroup-weighting. We may, however, still think of the semigroup-weighted average variance as an ``innovation term'', which is only related to the stochastic forcing $\sigma_tdW_t$ in \eqref{eq:model} and thus uniquely contains information on the shocks to prices. We note that in the high-frequency infill asymptotic setting of \citet{BenthSchroersVeraart2022,BenthSchroersVeraart2024}, the so-called semigroup adjusted realized covariation is used to adjust the naive realized covariation estimator, such that the propagation bias $B_n$ disappears in the limit when the mesh between observations tends to zero.

In order to ensure convergence of the average sample realized covariation to a fixed limit, we require stationarity and ergodicity of the sampling sequence. For the series of local averages $(\Delta_n\widetilde{X})$, this is inherited from stationarity and ergodicity of the underlying price differences $(\Delta_n X)$ or the latent process $(X_t,\sigma_t)$ itself. Standard criteria for stationarity of semilinear SPDE's can be found in \citet{DaPratoGatarekZabczyk1992}, which typically boil down to integrability conditions of the stochastic convolution appearing in \eqref{eq:mild_formulation} and stability of the semigroup $\mathcal{S}$. Ergodicity typically requires stronger stability conditions on the semigroup $\mathcal{S}$, and some rather technical criteria on the SPDE dynamics are given in, e.g., \citet{HairerMattingly2011}. A corresponding central limit theorem can be derived by imposing suitable mixing conditions on the SPDE, but we will not pursue this here. 

\begin{assumption}\label{ass:main_stat_assumption}
    The sequence of differenced local averages $(\Delta_n \widetilde{X})_{n\in\mathbb{N}}$ is strictly stationary and ergodic.
\end{assumption}

\begin{proposition}\label{prop:long_span_limit}
    Let $\delta = t_{n}-t_{n-1}$ and let $w\in\mathbb{N}$ denote a window length and suppose that we can partition the observations $(\widetilde{X}_{t_n})_{n=1}^{N}$ into $N_w$ disjoint windows of length $w$. Suppose that the model \eqref{eq:model} satisfies Assumption~\ref{ass:main_stat_assumption} and that $t_0=0$. Then    \begin{equation}\label{eq:long_span_limit}
        \frac{1}{N_w}\sum_{j=1}^{N_w}\widehat{\Sigma}_{j}^{w} \xrightarrow[N_w\to\infty]{a.s.} \mathbb{E}\left[ A(B_1+D_1)(B_1+D_1)^\ast A^\ast  \right] + A\mathbb{E}\left[ \int_{0}^{\delta}\mathcal{S}(\delta -s)\Sigma_s\mathcal{S}(\delta -s)^\ast ds \right]A^\ast,
    \end{equation}
    where $\widehat{\Sigma}_j^w$ is defined by \eqref{eq:RCV_estimator} and $B_1,D_1$ in Proposition~\ref{prop:increment_decomp}.
\end{proposition}
The estimator \eqref{eq:RCV_estimator} and the long-span limit of Proposition~\ref{prop:long_span_limit} are standard constructions, but obscured by extra terms, which are not present in the finite dimensional setting. In particular, the propagation terms $B_n$ and $B_1$ induced by the semigroup are normally not present. In Corollary~\ref{cor:analytic_semigroup} we show that this contributes with more bias than we expect from the drift alone and in Proposition~\ref{prop:drift_bound} we show how this bias can be corrected when the semigroup is known. Proposition~\ref{prop:plugin_estimator} gives a feasible plug-in estimator of the one-step semigroup on the observation space, and hence a feasible correction of the propagation bias when the semigroup is unknown. In particular, we obtain a plugin estimator such that the adjusted RCV estimator $\widehat{\Sigma}_{j,\mathrm{pl}}^{w}$ estimates the conditional second moment of the semigroup-adjusted increment of $\widetilde{X}_{t_j}$, which in general does not eliminate the propagation bias entirely, since we must first estimate the semigroup. 
\begin{proposition}\label{prop:drift_bound}
    Let the setting be as in Proposition~\ref{prop:long_span_limit} and suppose that the drift has bounded second moment, i.e., $\sup_{s\geq 0}\mathbb{E}\left[ \lVert \mu_s \rVert^2_{L^2(\mathbb{S}^1)}\right]<C_\mu<\infty$. Then
    \[
        \frac{1}{N_w}\sum_{j=1}^{N_w}\widehat{\Sigma}_{j}^{w} - \frac{1}{\delta}\mathbb{E}\left[ A B_1 B_1^\ast A^\ast \right] \xrightarrow[N_w\to\infty]{a.s.} A\mathbb{E}\left[ \int_{0}^{\delta}\mathcal{S}(\delta-s)\Sigma_s\mathcal{S}(\delta-s)^\ast ds \right]A^\ast + O(\delta^2),
    \]
\end{proposition}

\begin{proposition}\label{prop:plugin_estimator}
For notational simplicity, write $\widetilde X_n:=\widetilde X_{t_n}$. Suppose that the sequence $(\widetilde X_n)_{n\in\mathbb{N}}$ is strictly stationary, ergodic, and satisfies
\[
\mathbb E[\widetilde X_0]=0,
\qquad
\mathbb E[\lVert\widetilde X_0\rVert^2]<\infty.
\]
Set $\Gamma:=\mathbb E[\widetilde X_0\widetilde X_0^\top]$, and suppose that $\Gamma$ is positive definite. Define the effective one-step semigroup on the observation space by
\[
\mathcal S_\delta
:=
\mathbb E[\widetilde X_1\widetilde X_0^\top]\Gamma^{-1}\in\mathbb R^{d\times d},
\]
and its sample estimator by
\begin{equation}\label{eq:semigroup_sample_estimator}
\widehat{\mathcal{S}}_{\delta,N}
:=\Big(\sum_{n=1}^N \widetilde{X}_n\widetilde{X}_{n-1}^\top\Big)
\Big(\sum_{n=1}^N \widetilde{X}_{n-1}\widetilde{X}_{n-1}^\top\Big)^{-1}.
\end{equation}
Then it holds that $\widehat{\mathcal{S}}_{\delta,N}\xrightarrow[N\to\infty]{a.s.}\mathcal{S}_\delta$. Define the plug-in residuals by
\[
\widehat{\varepsilon}_n^{\mathrm{pl}}:=
\widetilde{X}_n-\widehat{\mathcal{S}}_{\delta,N}\widetilde{X}_{n-1}
=
\Delta_n\widetilde{X}-(\widehat{\mathcal{S}}_{\delta,N}-I_d)\widetilde{X}_{n-1},
\]
and the residual-based feasible realized covariation by
\begin{equation}\label{eq:RCV_pl}
\widehat{\Sigma}_{j,\mathrm{pl}}^{w}
:=\frac{1}{\delta w}\sum_{\ell=1}^w\widehat{\varepsilon}_{(j-1)w+\ell}^{\mathrm{pl}}
\big(\widehat{\varepsilon}_{(j-1)w+\ell}^{\mathrm{pl}}\big)^\top.
\end{equation}
Then each $\widehat\Sigma_{j,\mathrm{pl}}^{\,w}$ is positive semidefinite. Moreover, setting $N=wN_w$, it holds that
\[
\frac1{N_w}\sum_{j=1}^{N_w}\widehat\Sigma_{j,\mathrm{pl}}^{w}
\xrightarrow[N_w\to\infty]{a.s.}
\frac{1}{\delta}\,
\mathbb{E}\Big[(\widetilde X_1-\mathcal S_\delta \widetilde X_0)
(\widetilde{X}_1-\mathcal{S}_\delta \widetilde{X}_0)^\top
\Big],
\]
Equivalently,
\begin{equation}\label{eq:plugin_equiv_representation}
\frac1{N_w}\sum_{j=1}^{N_w}\widehat{\Sigma}_{j,\mathrm{pl}}^{w}
\xrightarrow[N_w\to\infty]{a.s.}
\frac{1}{\delta}\mathbb{E}\left[(\Delta_1\widetilde X)(\Delta_1\widetilde X)^\top\right]
-
\frac{1}{\delta}\mathbb E\Big[
(\mathcal S_\delta-I_d)\widetilde X_0\widetilde X_0^\top(\mathcal S_\delta-I_d)^\top
\Big].
\end{equation}
\end{proposition}

\begin{corollary}
    Let the setting be as in Proposition~\ref{prop:plugin_estimator} and assume in addition that there exists a matrix $S_{\delta,A}\in\mathbb R^{d\times d}$ such that
    \begin{equation}\label{eq:semigroup_closure_condition}
    A\mathcal{S}_\delta=S_{\delta,A}A,
    \end{equation}
    and that either $\mu\equiv 0$ or the data have been de-drifted so that
    \[
    \mathbb E[\widetilde{X}_n\mid \mathcal F_{t_{n-1}}]=S_{\delta,A}\widetilde{X}_{n-1}.
    \]
    Then $\mathcal{S}_\delta=S_{\delta,A}$, and therefore
    \[
    \frac{1}{N_w}\sum_{j=1}^{N_w}\widehat\Sigma_{j,\mathrm{pl}}^{\,w}
    \xrightarrow[N_w\to\infty]{a.s.}
    A\mathbb{E}\Big[\int_0^\delta
    \mathcal{S}(\delta-s)\Sigma_s\mathcal S(\delta-s)^\ast\,ds\Big]A^\ast.
    \]
\end{corollary}
\begin{remark}\label{rem:linear_semigroup_predictor}
    Without the closure condition \eqref{eq:semigroup_closure_condition}, the sample estimator \eqref{eq:semigroup_sample_estimator} still converges to a matrix representation on the observation space. Indeed, the population limit
    \[
    \mathcal{S}_{\delta}:=
    \mathbb{E}[\widetilde{X}_{t_1}\widetilde{X}_0^\top]\Gamma^{-1},
    \qquad
    \Gamma:=\mathbb{E}[\widetilde{X}_0\widetilde{X}_0^\top],
    \]
    is the unique solution to the least-squares problem
    \[
    \mathcal{S}_{\delta}=
    \arg\min_{M\in\mathbb R^{d\times d}}
    \mathbb E\left[\|\widetilde{X}_{t_1}-M\widetilde{X}_0\|^2\right].
    \]
    In other words, $\mathcal{S}_{\delta}\widetilde{X}_t$ is the best linear one-step predictor of the next vector of local averages, $\widetilde{X}_{t+\delta}$, from the current vector of local averages, $\widetilde X_t$. The sample estimator $\widehat{\mathcal{S}}_{\delta,N}$ is a feasible estimator of this population predictor.
\end{remark}
\begin{corollary}\label{cor:analytic_semigroup}
Assume that the drift has bounded second moment and that the semigroup $\mathcal S$ is analytic with $X_{t_0}\in\mathcal D((-\mathcal A)^\alpha)$ for some $\alpha\in(0,1]$. Then
\[
\frac{\delta}{N_w}\sum_{k=1}^{N_w}\widehat{\Sigma}^{w}_k\xrightarrow[N_w\to\infty]{a.s.}
A\,\mathbb{E}\left[\int_{0}^{\delta}
\mathcal S(\delta-s)\,\Sigma_s\,\mathcal S(\delta-s)^\ast\,ds\right]A^\ast
\;+\;O(\delta^{2\alpha}).
\]
\end{corollary}

\begin{remark}\label{rem:semigroup_regularity}
If, in addition to the conditions of Corollary~\ref{cor:analytic_semigroup}, $\mathcal{S}(u)$ is close to the identity on $[0,\delta]$ in the relevant sense (e.g. $\lVert\mathcal{S}(u)-I\rVert_{\mathrm{op}}\lesssim u$ and $\Sigma_s$ is sufficiently regular),
then the leading term in Corollary~\ref{cor:analytic_semigroup} above can be approximated by
$A\,\mathbb{E}[\int_{0}^{\delta}\Sigma_s\,ds]\,A^\ast$ up to an additional $O(\delta)$ error.
\end{remark}

\begin{remark}\label{rem:spatial_regularity}
    The condition $X_t\in\mathcal{D}((-\mathcal{A})^{\alpha})$ required in Corollary~\ref{cor:analytic_semigroup} is a spatial regularity condition. Applying $(-\mathcal{A})^\alpha$ to the mild solution \eqref{eq:mild_formulation} we see that $X_t\in\mathcal{D}((-\mathcal{A})^\alpha)$ if it holds that
    \begin{equation}\label{eq:HS_regularity}
    \mathbb{E}\left[ \int_{0}^{T}\lVert (-\mathcal{A})^\alpha \mathcal{S}(t-s)\sigma_s \rVert^2_{HS} ds \right] < \infty.
    \end{equation}
\end{remark} 

\subsubsection{Relation to classical realized covariation}
In the one-dimensional setting of Remark~\ref{rem:OU_1dim}, the corresponding result to that of Proposition~\ref{prop:long_span_limit} is that
\begin{equation}\label{eq:RCV_univariate}
\frac{1}{N_w}\sum_{j=1}^{N_w}\widehat{\Sigma}_j^w - \frac{1}{\delta}\left( e^{\lambda\delta}-1\right)^2\mathbb{E}[\widetilde{X}_0^2]\xrightarrow[N_w\to\infty]{a.s.}\frac{1}{\delta}\int_{0}^{\delta}e^{2\lambda (\delta-t)}\mathbb{E}\left[\sigma_t^2\right]dt + O(\delta),
\end{equation}
where the right-hand-side is exactly a noisy proxy of the annualized integrated variance \citep[and similarly for the realized covariation of][]{BNS_2004_RCV}. In this case, the semigroup weighting is explicitly known, and a small-time expansion of the integrand on the right hand side of \eqref{eq:RCV_univariate} can, for small $\delta$, recover $\frac{1}{\delta}\int_{0}^{\delta}\mathbb{E}\left[\sigma_t^2\right]dt$ up to an additional $O(\delta)$ term. A severe complication of the genuinely infinite dimensional setting is that the semigroup is much less explicitly known. A simplification is possible if the mild solution \eqref{eq:mild_formulation} is a strong Itô semimartingale in $L^2(\mathbb{S}^1)$, admitting the strong solution
\begin{equation}\label{eq:strong_form_solution}
X_t = X_0+\int_{0}^{t}\left( \mathcal{A}X_s+\mu_s \right)ds+\int_{0}^{t}\sigma_sdW_s,
\end{equation}
which does not directly involve the semigroup $\mathcal{S}$. In this case, we have $X_t\in\mathcal{D}(\mathcal{A})$ and we obtain that the limit in Proposition~\ref{prop:long_span_limit} is of the form 
\[
\frac{1}{N_w}\sum_{j=1}^{N_w}\widehat{\Sigma}_j^w - \frac{1}{\delta}A\mathbb{E}\left[ \left(\int_{0}^{\delta}\mathcal{A}X_sds\right)\left(\int_{0}^{\delta}\mathcal{A}X_sds\right)^\ast\right]A^\ast\xrightarrow[N_w\to\infty]{a.s.}\frac{1}{\delta}A\mathbb{E}\left[\int_{0}^{\delta}\sigma_t\sigma^\ast_tdt\right]A^\ast + O(\delta),
\]
in analogue to the classical case of realized covariation and in accordance with Corollary~\ref{cor:analytic_semigroup}. This situation is similar to that of \citet{BenthSchroersVeraart2022,BenthSchroersVeraart2024}, who derive quite explicit conditions to characterize when their semigroup adjustment is needed and when the ``naive'' realized covariation is appropriate. If the Hilbert space in which $X_t$ takes values is finite dimensional, then $X_t$ is always a semimartingale and our estimator collapses to the usual realized covariation of the observed process $\widetilde{X}_t=AX_t$. To this end, note that without the semigroup-weighting, the innovation term in \eqref{eq:long_span_limit} can be written entrywise as
\[
\left( A\mathbb{E}\left[ \int_{0}^{\delta}\Sigma_sds\right]A^\ast\right)_{i,j} = \int_{0}^{\delta}\mathbb{E}\left[ \langle \Sigma_s g_j,g_i\rangle_{L^2(\mathbb{S}^1)} \right]ds.
\]
If the operator $\Sigma_s$ itself admits a kernel, $c$, then we have
\[
\left( A\mathbb{E}\left[ \int_{0}^{\delta}\Sigma_sds\right]A^\ast\right)_{i,j} = \int_{0}^{\delta}\mathbb{E}\left[ \frac{1}{\lvert I_i\rvert \lvert I_j\rvert }\int_{I_i}\int_{I_j}c_s(h,u) du dh \right] ds,
\]
where $I_i,I_j$ represent the bin intervals. Hence we have, in the semimartingale case \eqref{eq:strong_form_solution}, the concrete interpretation as the time integral of the average covariance kernel over spatial bins $i$ and $j$.

\subsubsection*{Realized variance of the average daily price}
The object of study is often the average daily spot price, corresponding to modeling $\overline{P}_t=\frac{1}{2\pi}\int_{0}^{2\pi}X_t(h)dh$. This object arises as a natural univariate representation of ``the price of electricity'' and it forms the basis of the underlying in conventional futures contracts. If $X_t$ is a strong semimartingale solution to \eqref{eq:model}, then the process $\overline{P}_t$ has quadratic variation (ignoring propagation and drift terms for simplicity)
\begin{equation}\label{eq:QV_avg_price}
\langle \overline{P}\rangle_t = \int_{0}^{t}\lVert \sigma^{\ast}_s\overline{g}\rVert^2_{L^2(\mathbb{S}^1)} ds = \int_{0}^{t}\langle \Sigma_s\overline{g},\overline{g}\rangle_{L^2(\mathbb{S}^1)} ds,
\end{equation}
where $\overline{g}(h)=\frac{1}{2\pi}\mathbf{1}_{[0,2\pi)}(h)$. In particular, it holds that the $w$-period realized variance $\widehat{RV}_j^w$ of $\overline{P}$ is given by (even with propagation and drift terms included)
\begin{equation}\label{eq:RV_vs_RCV}
\widehat{RV}_j^w = \frac{1}{d^2}\mathbf{1}^\top\widehat{\Sigma}_j^w\mathbf{1} = \frac{1}{d^2}\sum_{i=1}^{d}\sum_{k=1}^{d}\left(\widehat{\Sigma}_j^w\right)_{ik}, 
\end{equation}
where $\mathbf{1}=(1,\ldots ,1)^\top \in \mathbb{R}^d$. In practice, this means we may easily obtain the realized variance of the average daily price from our RCV estimator. This highlights one of the strengths of the continuous local average interpretation \eqref{eq:price_local_average}, since the realized variance of the average price and the realized covariation of the panel of prices are connected without imposing any restrictions on the sampling scheme in time or space. For example, the relation \eqref{eq:RV_vs_RCV} is easily adapted to the case of spatial averaging over intervals of different lengths, by incorporating non-equal weights. 

\begin{remark}
    Notice that if $X$ is only a mild solution to \eqref{eq:model} of the form \eqref{eq:mild_formulation}, then the semigroup smoothing disappears from the quadratic variation \eqref{eq:QV_avg_price} whenever it holds that $\mathcal{S}(t)^\ast \overline{g} = \overline{g}$ for all $t\geq 0$. This is often the case, such as for diffusion type semigroups on $\mathbb{S}^1$.
\end{remark}

\section{Realized covariation of electricity spot prices}\label{sec:rcv_of_prices}
In this Section, we apply the RCV estimator \eqref{eq:RCV_estimator} and the propagation adjusted version \eqref{eq:RCV_pl} to construct proxies of the weekly integrated variance in three European electricity markets. We consider three different generation zones that are characterized by different underlying characteristics: Germany, Norway (NO2), and Spain. Germany is the largest market for electricity and related products in Europe, and owing to their central geographical placement in Europe, their electricity production is affected by many adjacent generation zones via, e.g., interconnectors. Norway is characterized by having a large share of hydropower generation, which is expected to lead to more stable prices, since they can regulate their production according to price signals. Spain is characterized by a large amount of solar energy, which may cause volatile price behaviour due to the nature of renewable generation. To save space throughout the main body of the paper, we shall mainly depict the results for Germany, with the corresponding results for Norway and Spain provided in Appendix~\ref{app:plots}.

In the notation of Section~\ref{sec:vol_estimation}, we choose windows of $w=7$ days, such that the RCV is a noisy estimate of the annualized (i.e. $\delta=\tfrac{1}{365}$) weekly (semigroup weighted) integrated variance. The choice of one week is made to eliminate subtle weekly periodicities such as weekend effects. We consider data from October 1st 2018 to September 30th 2025. Prior to this period, the German zone was coupled with Austria and hence represents a structurally different market, and after this period most European zones converted from 24 daily delivery periods (hourly) to 96 daily delivery periods (15 minutes). Due to the local average interpretation \eqref{eq:price_local_average}, it is in principle possible to still compute the hourly local averages from the quarterly prices, but we refrain from doing this due to potential market reactions to the transition. This leaves us with $N=2551$ daily observations per zone and a total of 61368 data points per generation zone. Aggregation into disjoint weekly bins leaves us with 365 weeks of data. To obtain more data and, by extension, more stable results in the following, we therefore consider a 7-day rolling RCV estimate according to \eqref{eq:RCV_estimator} or \eqref{eq:RCV_pl}, such that we have 2557 RCV estimates per zone. This induces a degree of temporal smoothing, but the interpretation as a noisy proxy for the 1-week integrated variance remains, and the result of Proposition~\ref{prop:long_span_limit} is unchanged for the rolling sequence of price differences. Additionally, since each 7-day interval contains exactly one of each weekday, the issue of deterministic weekly periodicities is still avoided. On daylight-savings transition days, we impute the data as follows. If the affected day has 23 hours, we let the missing hour have a price given by the average of the adjacent hours. If the affected day has 25 hours, we have two ``overlapping hours''; the prices in these two hours are merged into one price as their average.

Before considering the estimates obtained from data, we first deal with the terms $D_n,B_n$ of \eqref{eq:RCV_estimator}. In Section~\ref{sec:stationarity}, we argue that the sequence of differences $(\Delta_n\widetilde{X})_{n\in\mathbb{N}}$ obtained from the data are consistent with the assumption of stationarity and do not exhibit any unit root type effects. Hence we can reasonably assume that $D_n\approx 0$ and that any non-stationarities are slowly moving trends. In Section~\ref{sec:mean_reversion}, we apply the estimator $\widehat{\Sigma}_{j,\mathrm{pl}}^w$ to eliminate propagation effects and show that this differs considerably from the raw estimator $\widehat{\Sigma}_{j}^w$. It follows that the propagation terms $B_n$ contribute with substantial variation in prices and should not be ignored. Furthermore, we consider the estimated matrix representation of the semigroup $\mathcal{S}(\delta)$, which reveals an intricate structure.

\begin{figure}[htb]
    \centering
    \begin{subfigure}[b]{0.32\textwidth}
        \centering
        \textbf{Germany}
        \includegraphics[width=\textwidth]{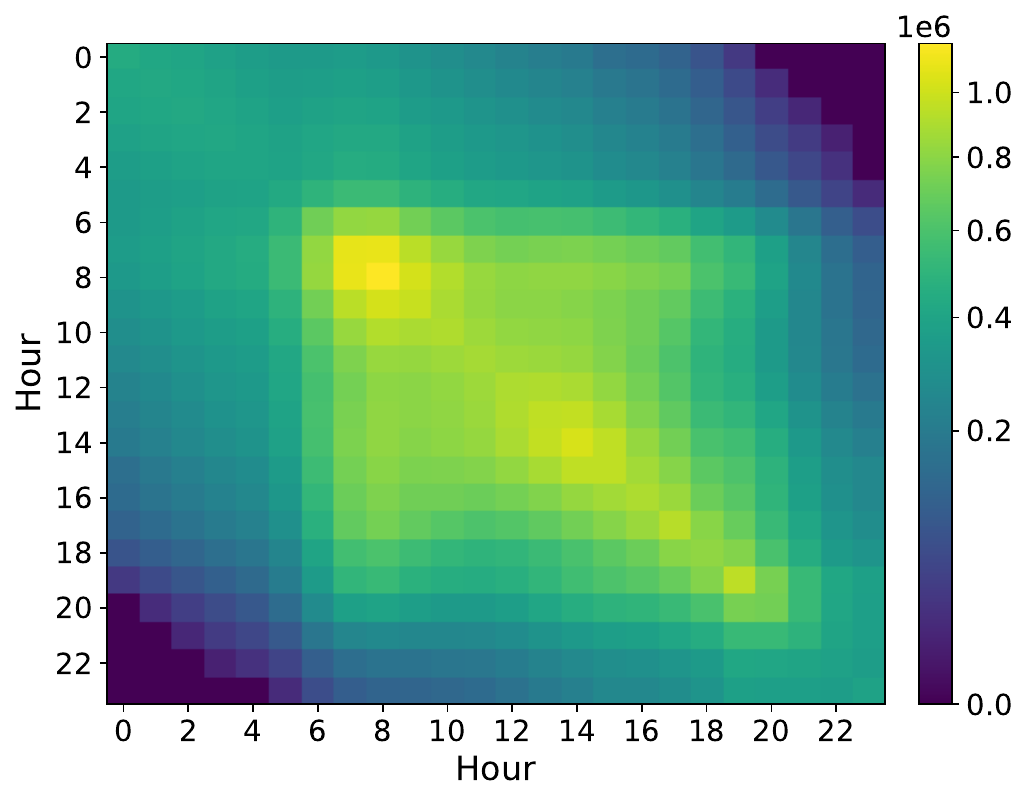}
    \end{subfigure}
    \hfill
    \begin{subfigure}[b]{0.32\textwidth}
        \centering
        \textbf{Norway (NO2)}
        \includegraphics[width=\textwidth]{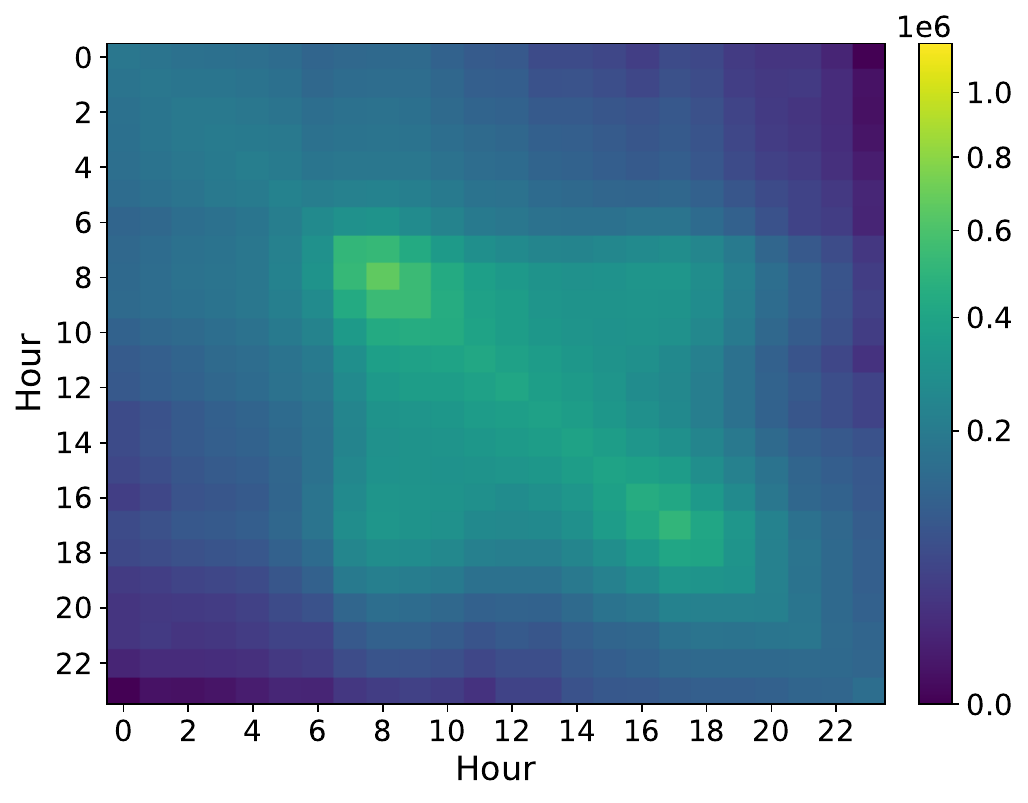}
    \end{subfigure}
    \hfill
    \begin{subfigure}[b]{0.32\textwidth}
        \centering
        \textbf{Spain}
        \includegraphics[width=\textwidth]{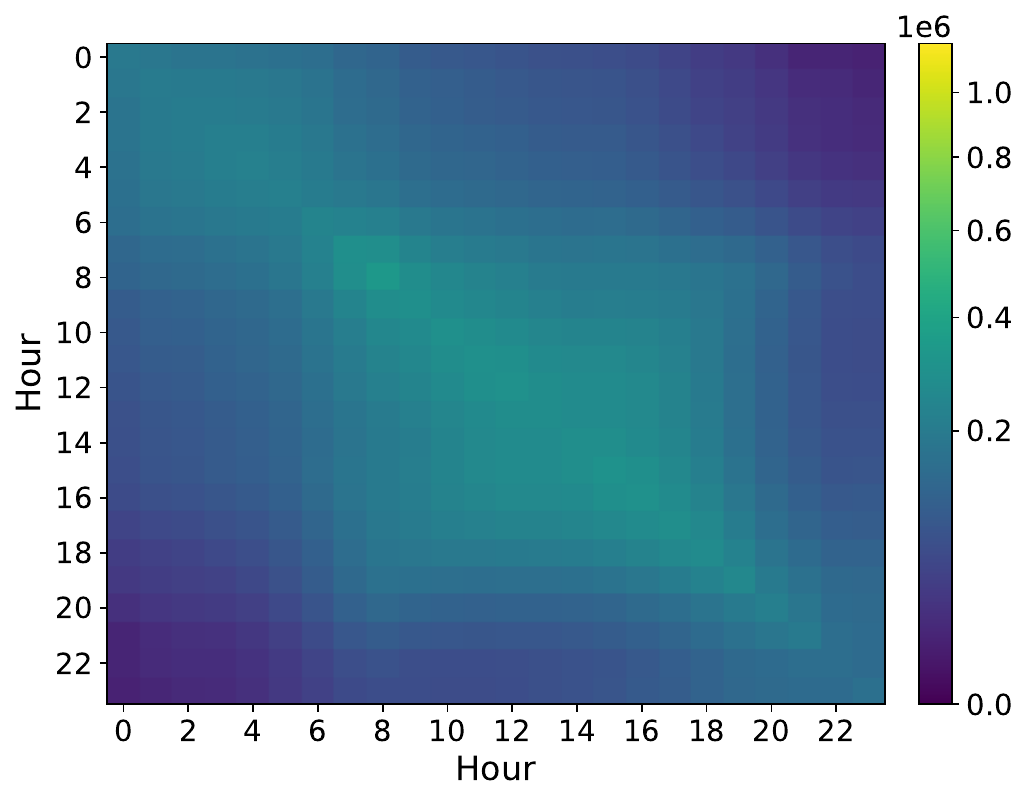}
    \end{subfigure} 
    \vspace{1cm}
    \begin{subfigure}[b]{0.32\textwidth}
        \centering
        \includegraphics[width=\textwidth]{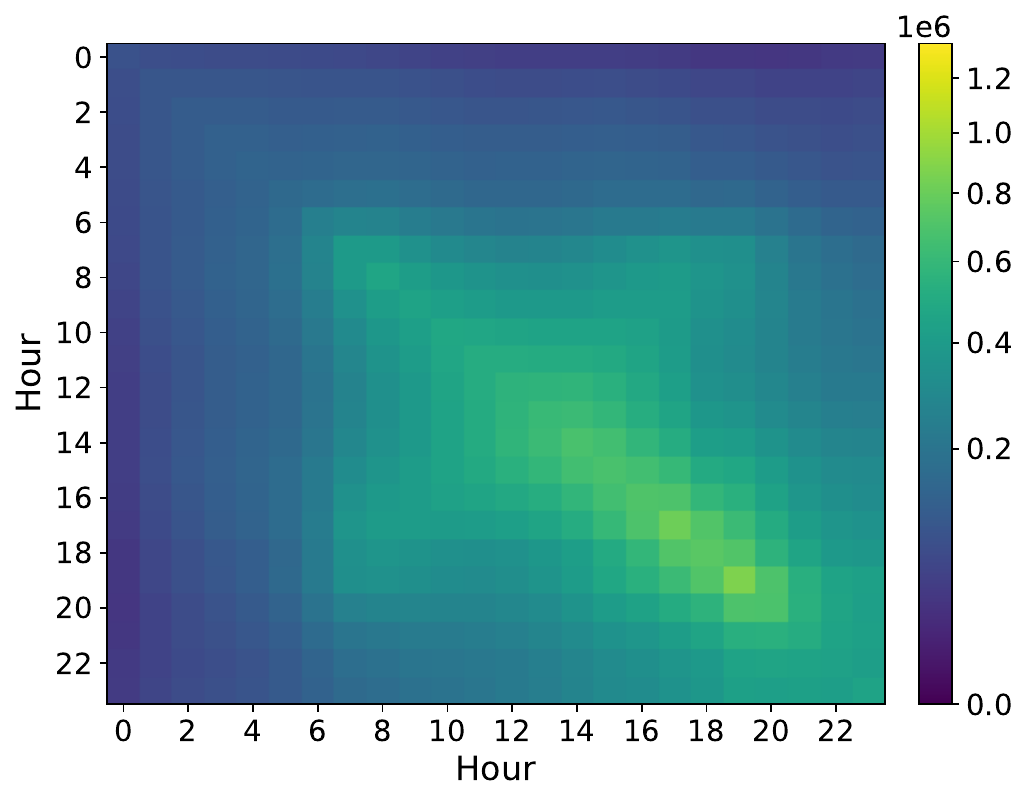}
    \end{subfigure}
    \hfill
    \begin{subfigure}[b]{0.32\textwidth}
        \centering
        \includegraphics[width=\textwidth]{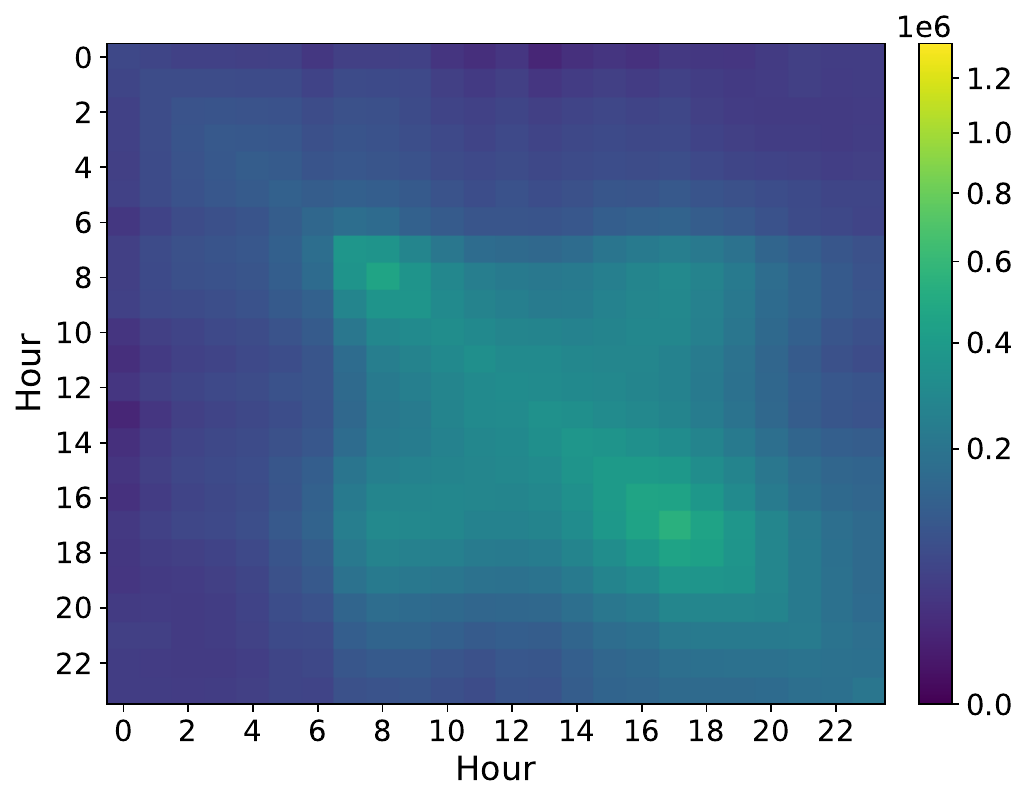}
    \end{subfigure}
    \hfill
    \begin{subfigure}[b]{0.32\textwidth}
        \centering
        \includegraphics[width=\textwidth]{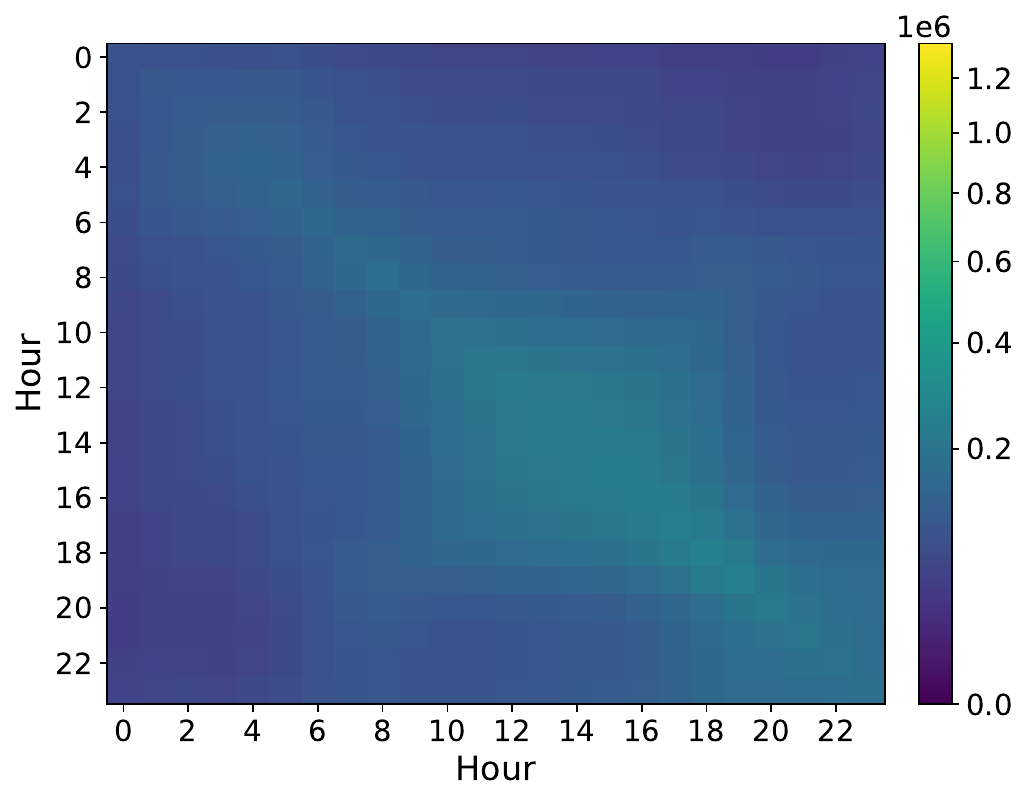}
    \end{subfigure} 
    \caption{Matrices depicting the average of rolling weekly realized covariation $\frac{1}{N_7}\sum_{j=1}^{N_7}\widehat{\Sigma}^7_j$ (top) and the average of rolling weekly adjusted realized covariation $\frac{1}{N_7}\sum_{j=1}^{N_7}\widehat{\Sigma}_{j,\mathrm{pl}}^w$ (bottom). Note that all of the plots are on the same scale.}
    \label{fig:realized_covariation_matrices}
\end{figure}

\subsection{Handling non-stationarities}\label{sec:stationarity}
It is well-known that electricity price levels tend to exhibit non-stationary behaviour in the form of predictable trends and seasonal periodicities. Apart from potentially contradicting Assumption~\ref{ass:main_stat_assumption}, this also leads to larger bias and contamination in the estimates of weekly realized covariation. The procedure of removing non-stationarities in electricity prices is often associated to quite intricate filtering procedures, with an overview of the one-dimensional case given in, e.g., \citet{JanczuraJoanna2013Isas}. We are, however, primarily interested in the sequence of differences, $(\Delta\widetilde{X})_{n=1}^{N}$ and so, if the non-stationarities in price levels can be interpreted as slow trends or unit roots, then the sequence of price differences should be suitably stationary.

We carry out several stationarity assessments for the 24-dimensional series of observed local averages $(P_n)_{n=1}^{N}$ in each country. As the primary test, we consider the multivariate KPSS type test for stationarity proposed in \citet{NyblomHarvey2000}. We supply this with the standard univariate KPSS test for stationarity on the individual price series, as well as the augmented Dickey-Fuller (ADF) test for unit roots on the individual price series. The test statistics and their critical values are tabulated in Table~\ref{tab:stationarity_levels} in Appendix~\ref{app:stationarity}. It is apparent from the KPSS tests that we reject stationarity in price levels across all countries at any reasonable level of significance, both jointly and in the individual price series. Interestingly, however, the ADF test rejects the presence of a unit root in most of the individual price series, which fits well with the notion that electricity spot prices are mean-reverting and hence shocks are not persistent. In Table~\ref{tab:stationarity_differences} in Appendix~\ref{app:stationarity}, we state the corresponding test statistics for the differenced price series, $(\Delta\widetilde{X}_n)_{n=1}^{N}$, which unanimously favours the assumption of stationarity at any reasonable level of significance. This suggests that the non-stationarities in price levels are slowly evolving trends that can largely be differenced out. Interestingly, any periodic effects seem to also be eliminated by differencing, at least to the extent at which the applied tests can detect them. To this end, we note that the KPSS alternative can be quite weak to some forms of non-stationarity in the covariance structure. Overall, we treat the price differences $(\Delta\widetilde{X}_n)_{n=1}^{N}$ as being largely trend-stationary, but noting that there may still be some non-stationary effects which are hard to pick up. This would potentially invalidate the long run estimator of Proposition~\ref{prop:long_span_limit}, but does not affect the conditional estimates of the RCV estimator \eqref{eq:RCV_estimator}. In particular, we note that the differencing means that the drift contribution is small, such that $D_n\approx 0$. Hence we expect that most of the bias arises from the propagation term, $B_n$, which can then be accounted for via Proposition~\ref{prop:plugin_estimator}.  

\subsection{Filtering of the propagation term}\label{sec:mean_reversion}
In addition to non-stationary predictable effects, it is well-documented that electricity prices exhibit mean-reversion, in the sense that price shocks do not persist indefinitely and tend to gradually diminish over time. In \citet{HuismanHuurmanMahieu2007}, it is found that prices in separate hours tend to mean-revert at different speeds. In the model \eqref{eq:model}, this is captured by the term $\mathcal{A}X_tdt$ and leads to the additional propagation terms $\mathbb{E}[A B_nB_n^\ast A^\ast]$ in the estimator \eqref{eq:RCV_estimator}. As we are mainly interested in the second order structure of the innovations/shocks to prices, we seek to filter out the deterministic propagation part by means of the plugin estimator $\widehat{\Sigma}^w_{j,\mathrm{pl}}$ of Proposition~\ref{prop:plugin_estimator}. Given the results of Section~\ref{sec:stationarity}, we assume that the drift admits the decomposition
\[
\mu_t = b_t+\widetilde{\mu}_t,
\]
where $b_t$ is a predictable process such that 
\[
\frac{1}{N}\sum_{n=1}^{N}\lVert m_n-m_{n-1}\rVert^2 \xrightarrow[N\to\infty]{}0, \quad m_n = A\int_{0}^{t_n}\mathcal{S}(t_n-s)b_sds,
\]
and the joint process $(\widetilde{\mu},\sigma)$ is strictly stationary. We can then estimate $\widehat{m}_n=\widehat{a}_n$, where
\begin{equation}\label{eq:mean_regression}
(\widehat{a}_n,\widehat{c}_n) = \arg\min_{a,c}\sum_{k=1}^{n-1}K\left( \frac{n-k}{\tau}\right)\lVert \widetilde{X}_k-a-c(k-n)\rVert^2,
\end{equation}
where $K$ is a kernel and $\tau$ the bandwidth. For our applications, we choose an Epanechnikov kernel with a bandwidth of 90 days, corresponding to $\tau = 0.246$; in order to eliminate weekly seasonality, we also include day-of-week dummies in the regression \eqref{eq:mean_regression}. The process $(\widetilde{X}_{t_n}-\widehat{m}_n)_{n\in\mathbb{N}}$ is then strictly stationary and with mean zero, such that Proposition~\ref{prop:plugin_estimator} applies and the resulting estimator $\widehat{\Sigma}^w_{j,\mathrm{pl}}$ removes some of the predictable propagation effect.

The semigroup estimator $\widehat{\mathcal{S}}_\delta$ in \eqref{eq:semigroup_sample_estimator} is, according to Remark~\ref{rem:linear_semigroup_predictor}, such that $\widehat{\mathcal{S}}_{\delta}\widetilde{X}_t$ is the best linear one-step predictor of the de-meaned vector of prices. The $(i,j)$'th entry of $\widehat{\mathcal S}_\delta$ can thus be interpreted as the contribution of today's hour $j$ in predicting tomorrow's hour $i$. Hence if the hourly price series $P_t^{(h)}$ are independent across $h$, the off-diagonal elements of $\widehat{\mathcal S}_\delta$ vanish, and the $i$'th diagonal entry equals exactly the AR(1) coefficient of the time series $P_t^{(i)}$ (this intuition breaks down whenever the off-diagonal elements are non-zero, but it is worth keeping in mind).

On Figure~\ref{fig:estimated_semigroup_DE_LU}, we depict the estimated $\widehat{\mathcal{S}}_\delta$ from the full data from German generation zone along with the corresponding eigenvalue spectrum. The estimated off-diagonal matrix is clearly non-zero and quite heterogeneous, supporting the findings of cross-sectionally varying mean-reversion rates in \citet{HuismanHuurmanMahieu2007}. In particular, we note that column 24 is positive and of larger magnitude than the other hours, with a slight decay from the evening to morning hours. This is evidence of price cyclicality as also documented in \citet{Kloster2026}, where the last price of today contributes with significant predictive power for the first few hours of tomorrow. The largest eigenvalue is $\lambda_1\approx0.8$, corresponding to common shocks across all hours to have a half-life of roughly $-\log(2)/\log(0.8)\approx 3$ days. The corresponding plots for Norway and Spain are depicted in Figures~\ref{fig:estimated_semigroup_NO_2} and \ref{fig:estimated_semigroup_ES} in Appendix~\ref{app:plots} and we note that the overall structure is strikingly similar although the common shock half-life is ca. 6 days in Norway. In practice, we cannot rely on the semigroup matrix $\widehat{\mathcal{S}}_\delta$ estimated from the full dataset, as this creates look-ahead bias. Whenever it is relevant to maintain causality in the analyses throughout the paper, we re-estimate $\widehat{\mathcal{S}}_\delta$ based on the available backward-looking data every four weeks (28 days) starting with a 52 week burn-in.

\begin{figure}[htbp]
    \centering
    \begin{subfigure}[b]{0.45\textwidth}
        \centering
        \includegraphics[width=\textwidth]{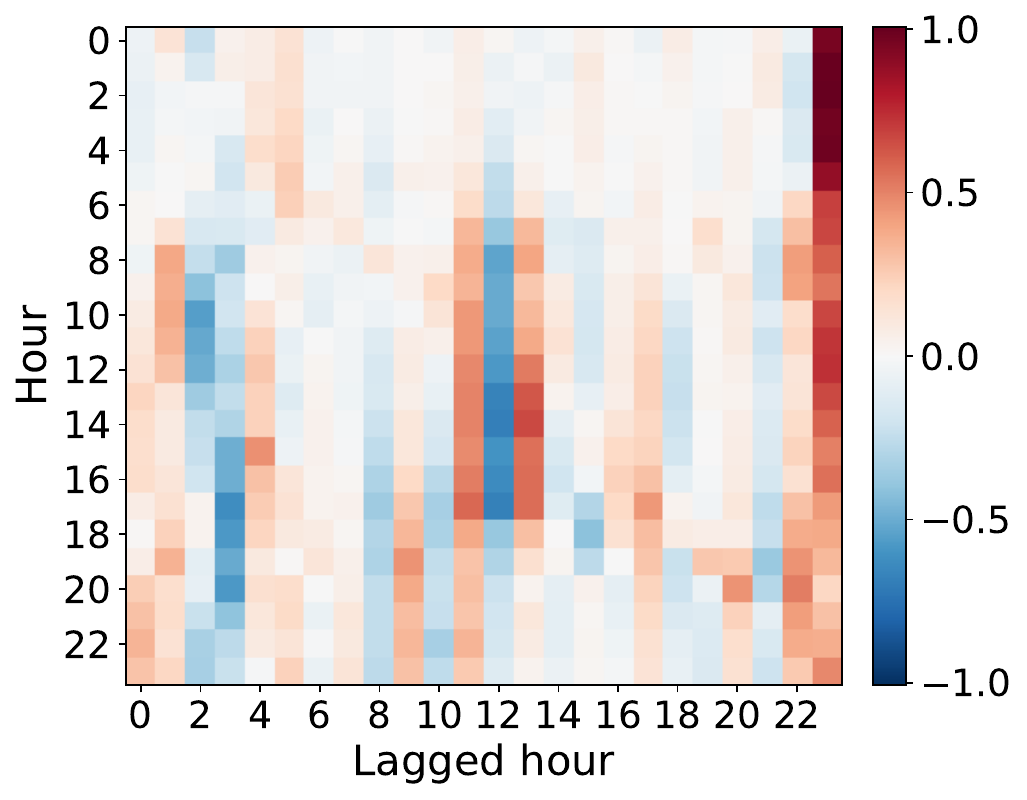}
        \subcaption{Semigroup matrix $\widehat{\mathcal{S}}_\delta$}
    \end{subfigure}
    \hfill
    \begin{subfigure}[b]{0.45\textwidth}
        \centering
        \includegraphics[width=\textwidth]{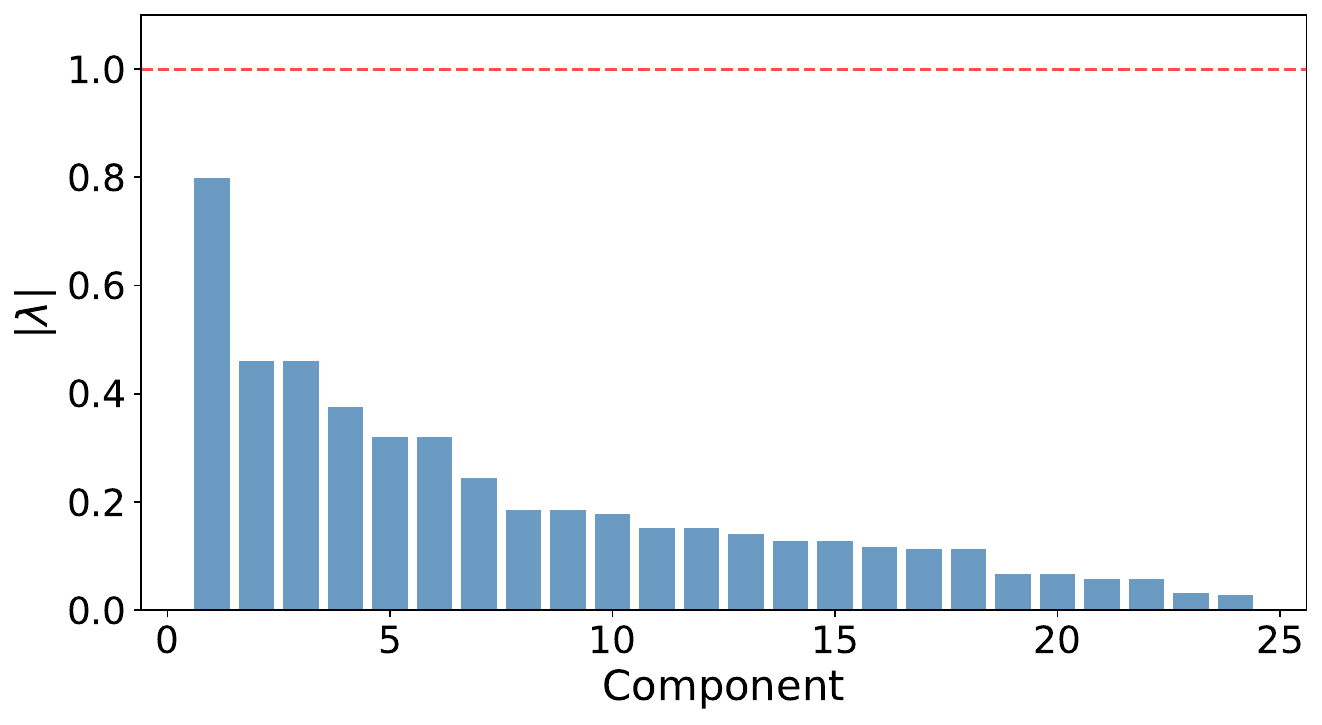}
        \subcaption{Eigenvalue spectrum}
    \end{subfigure}
    \caption{Estimated semigroup matrix $\widehat{\mathcal{S}}_\delta$ in the German generation zone (left) and the corresponding eigenvalue spectrum (right).}
    \label{fig:estimated_semigroup_DE_LU}
\end{figure}

\subsection{Preliminary findings}
In order to assess the magnitude of the propagation terms $B_j$ in the estimator $\widehat{\Sigma}^w_j$, we construct the \emph{propagation share}, denoted $PS$, designed to measure how much variation in the naive RCV estimator $\widehat{\Sigma}_j^w$ can be explained by the deterministic propagation via $\mathcal{A}$. Supposing that the sequence $(\widetilde{X}_{t_n})_{n}$ has been de-meaned via \eqref{eq:mean_regression}, it follows from Proposition~\ref{prop:increment_decomp} that $\Delta_n\widetilde{X}$ decomposes as $\Delta_n\widetilde{X} = B_n+M_n$, where $B_n$ and $M_n$ are orthogonal. Hence 
\[
\sum_n \lVert \Delta_n\widetilde{X}\rVert^2 = \sum_{n}\lVert B_n\rVert^2 + \sum_{n}\lVert M_n\rVert^2,
\]
and we can define $PS$ as the following coefficient of determination
\[
PS = \frac{\sum_{n}\lVert \widehat{B}_n\rVert^2}{\sum_n \lVert \Delta_n\widetilde{X}\rVert^2},
\]
where $\widehat{B}_n=(\widehat{\mathcal{S}}_\delta - I)\widetilde{X}_{t_{n-1}}$. Across the three zones, the propagation share is $PS\approx 0.4$, indicating that approximately $40\%$ of the observed variation across all hours can be attributed to mean-reversion. Interestingly, when inspecting the propagation shares of individual hours $PS_h$, corresponding to the ratio of column-sums of squared components
\[
PS_h = \frac{\sum_n (\widehat{B}_n^{(h)})^2}{\sum_n (\Delta_n\widetilde X^{(h)})^2}, \qquad h = 1, \ldots, d,
\]
we find that there is a clear decreasing proportion as $h$ increases. This means that, in all three zones under consideration, the late hours tend to exhibit a larger share of variation coming from raw noise than the early hours. In other words, the early hours tend to be relatively more predictable from the current price level itself. We depict the propagation share across hours on Figure~\ref{fig:propagation_DE_LU} for Germany, and on for Norway and Spain on Figures~\ref{fig:propagation_NO_2} and \ref{fig:propagation_ES} in Appendix~\ref{app:plots}. An interesting question arising from the observed monotone decay in $h\mapsto PS_h$ is whether this can be attributed to the temporal distance between the time at which prices are determined and the time of delivery. Indeed, when prices are determined, the price $P^{(d)}_t$ is valid many hours later than $P_t^{(1)}$ and there may thus be more uncertainty about the underlying supply and demand further in the future, leading to less predictable prices. We further investigate this effect in Section~\ref{sec:propagation}

On Figure~\ref{fig:realized_covariation_matrices}, we depict the unconditional estimates of the rolling weekly RCV in each market, constructed according to Proposition~\ref{prop:long_span_limit}. We include both the estimates obtained from $\widehat{\Sigma}_j^w$ and $\widehat{\Sigma}^w_{j,\mathrm{pl}}$. It is apparent that it is very important to account for the deterministic mean-reversion, as this removes a substantial amount of variation in prices. On Figure~\ref{fig:realized_correlation_matrices}, we depict the corresponding standardized ``realized correlations'', computed as
\begin{equation}\label{eq:realized_corr}
\rho = \widehat{Q}^{-1}\widehat{\Sigma}^w\widehat{Q}^{-1}, 
\end{equation}
where $\widehat{\Sigma}^w=\frac{1}{N_w}\sum_{j=1}^{N_w}\widehat{\Sigma}_j^w$ and $\widehat{Q} = \mathrm{diag}(\widehat{\Sigma}^w)^{1/2}$, and similarly for the propagation adjusted version. This realized correlation serves as a normalized measure of the estimated RCV, to better visually gauge patterns. Also here, it is apparent that the propagation adjustment leads to substantially different results. Most strikingly, we observe that prior to the adjustment there is a cluster of negative values, whereas all entries across all zones are positive after the adjustment. The correlations that appear negative at first sight are between early and late hours; in Section~\ref{sec:propagation}, we show that this apparent negative correlation arises from the larger propagation share of early hours where most variation arise from the deterministic propagation which tends to drag early hour prices in the same direction, regardless of the shocks to other hours.

\begin{figure}
    \centering
    \includegraphics[width=\linewidth]{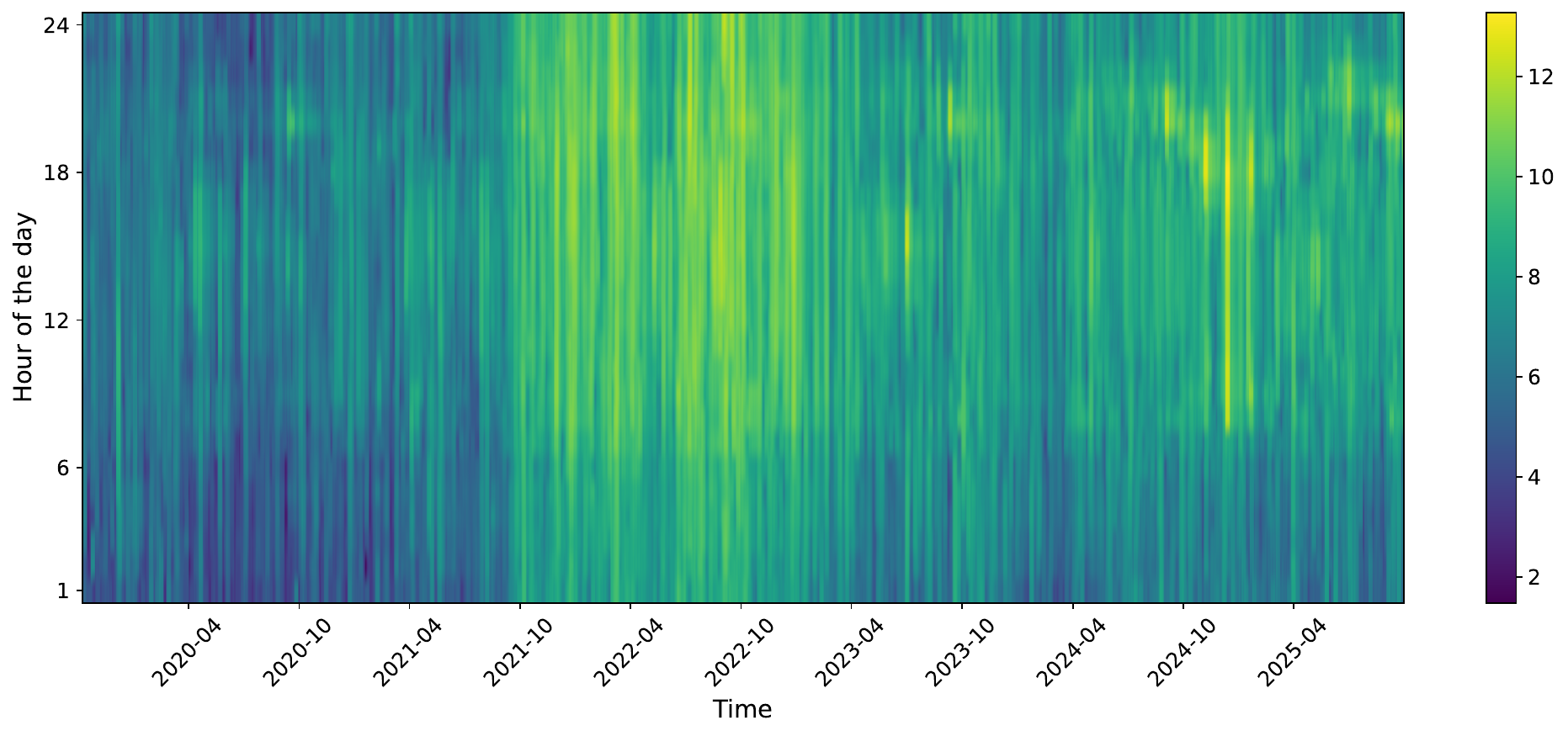}
    \caption{Heatmap of logarithm of estimated weekly RCV, $\mathrm{diag}(\widehat{\Sigma}^7_{t,\mathrm{pl}})$, for the German generation zone.}\label{fig:log_vol_heatmap_DE_LU}
\end{figure}

As an illustration of the conditional integrated variance, we depict the diagonal $\mathrm{diag}(\widehat{\Sigma}_{j,\mathrm{pl}}^w)$ over time, corresponding to the vector of integrated variances across all hours. Figure~\ref{fig:log_vol_heatmap_DE_LU} shows the results for the German zone, while Figures~\ref{fig:log_vol_heatmap_NO_2} and \ref{fig:log_vol_heatmap_ES} in Appendix~\ref{app:plots} show the results of the Norwegian and Spanish zone. To maintain a reasonable scale, we plot the logarithm of the estimated conditional integrated variances, which in particular keeps the extreme spikes during 2022 at reasonable levels. It is clear that the overall volatility varies wildly over time with apparent clustering effects. Furthermore, we note that periods of high volatility seems to often coincide with periods of high price levels. We shall return to this issue in Section~\ref{sec:inverse_leverage}.

\begin{figure}[t]
    \centering
    \begin{subfigure}[b]{0.32\textwidth}
        \centering
        \textbf{Germany}
        \includegraphics[width=\textwidth]{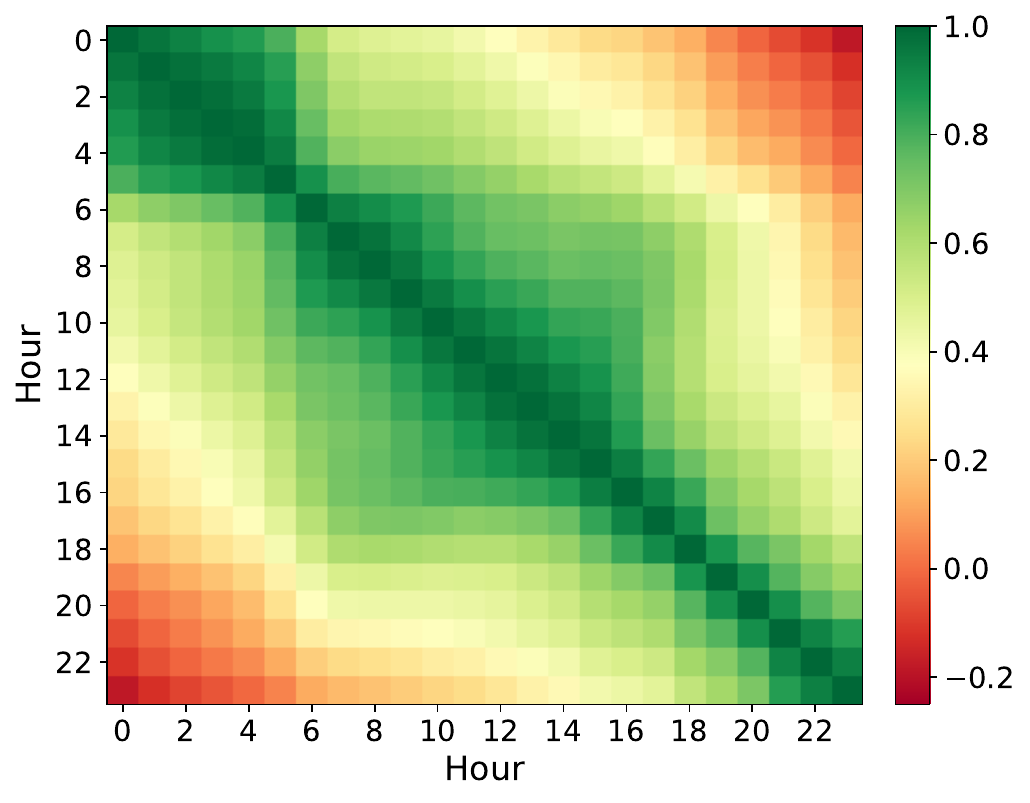}
    \end{subfigure}
    \hfill
    \begin{subfigure}[b]{0.32\textwidth}
        \centering
        \textbf{Norway (NO2)}
        \includegraphics[width=\textwidth]{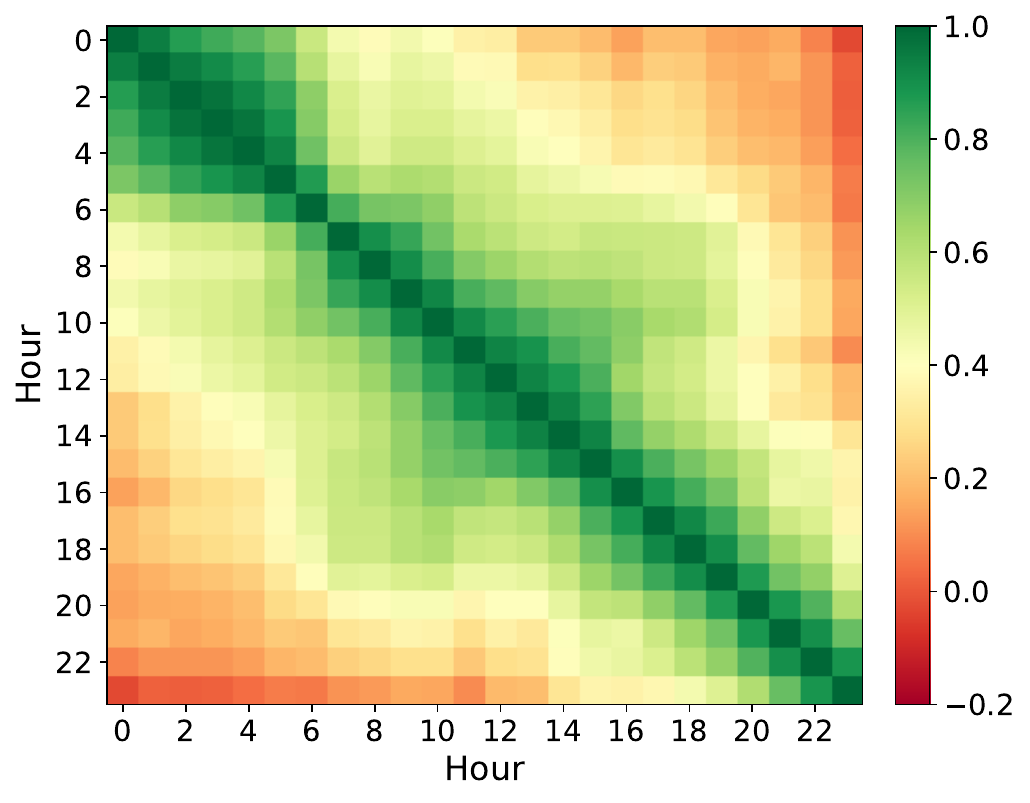}
    \end{subfigure}
    \hfill
    \begin{subfigure}[b]{0.32\textwidth}
        \centering
        \textbf{Spain}
        \includegraphics[width=\textwidth]{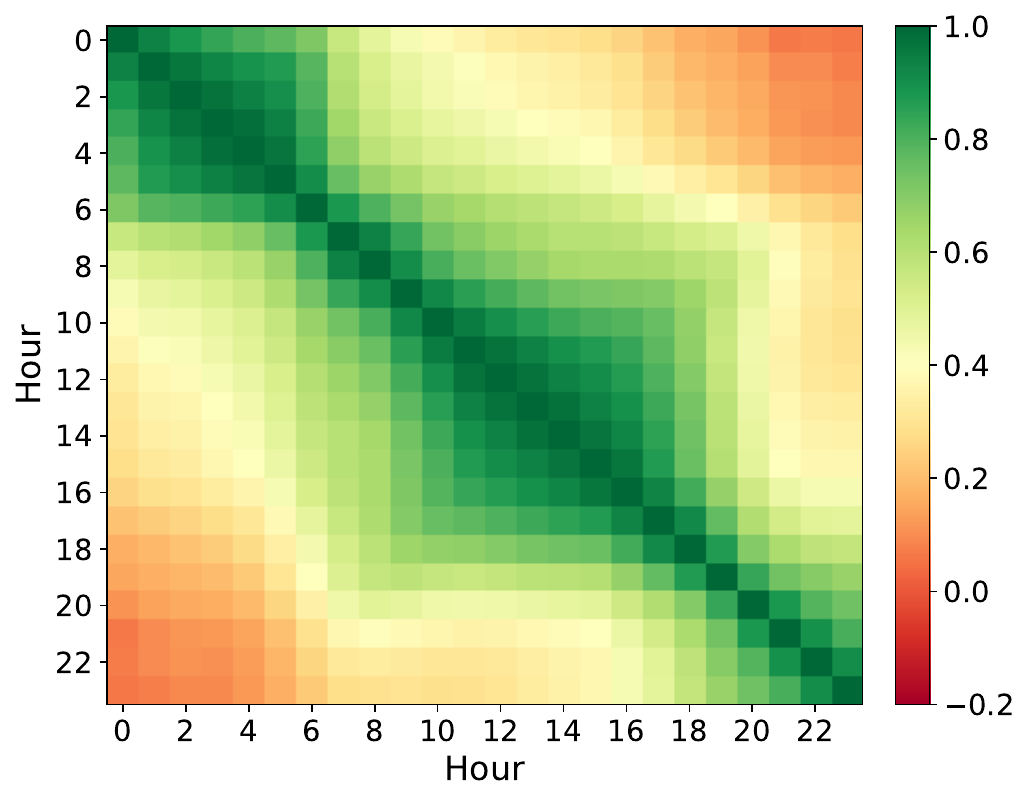}
    \end{subfigure} 
    \vspace{1cm}
    \begin{subfigure}[b]{0.32\textwidth}
        \centering
        \includegraphics[width=\textwidth]{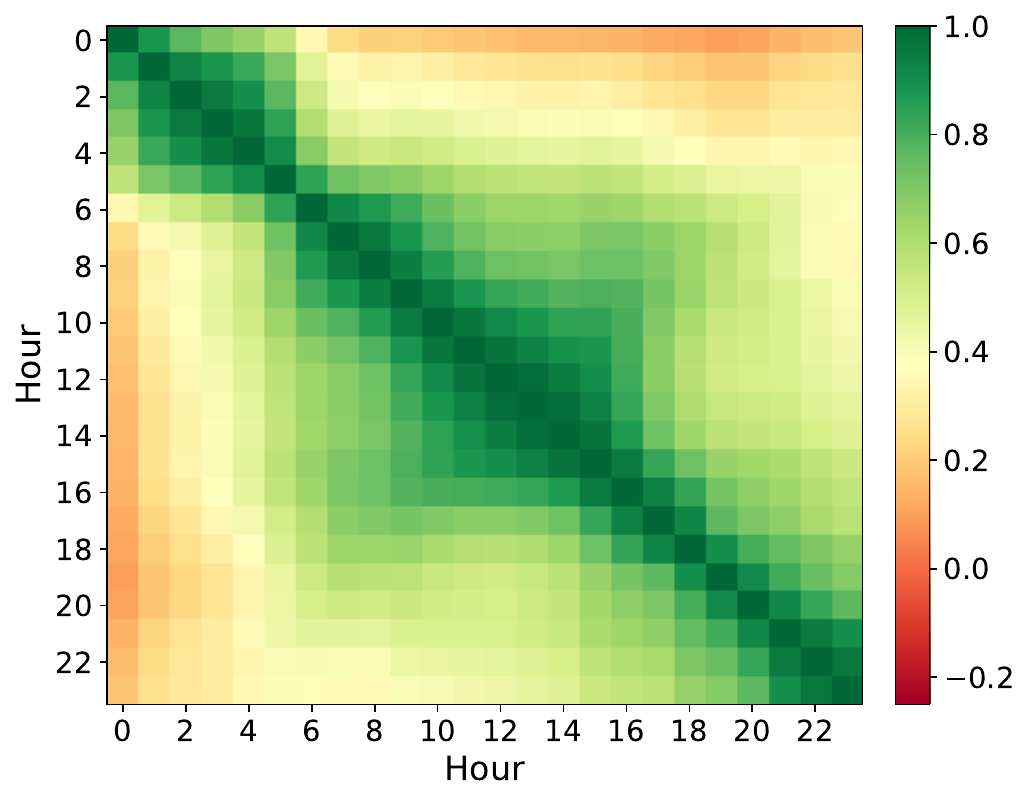}
    \end{subfigure}
    \hfill
    \begin{subfigure}[b]{0.32\textwidth}
        \centering
        \includegraphics[width=\textwidth]{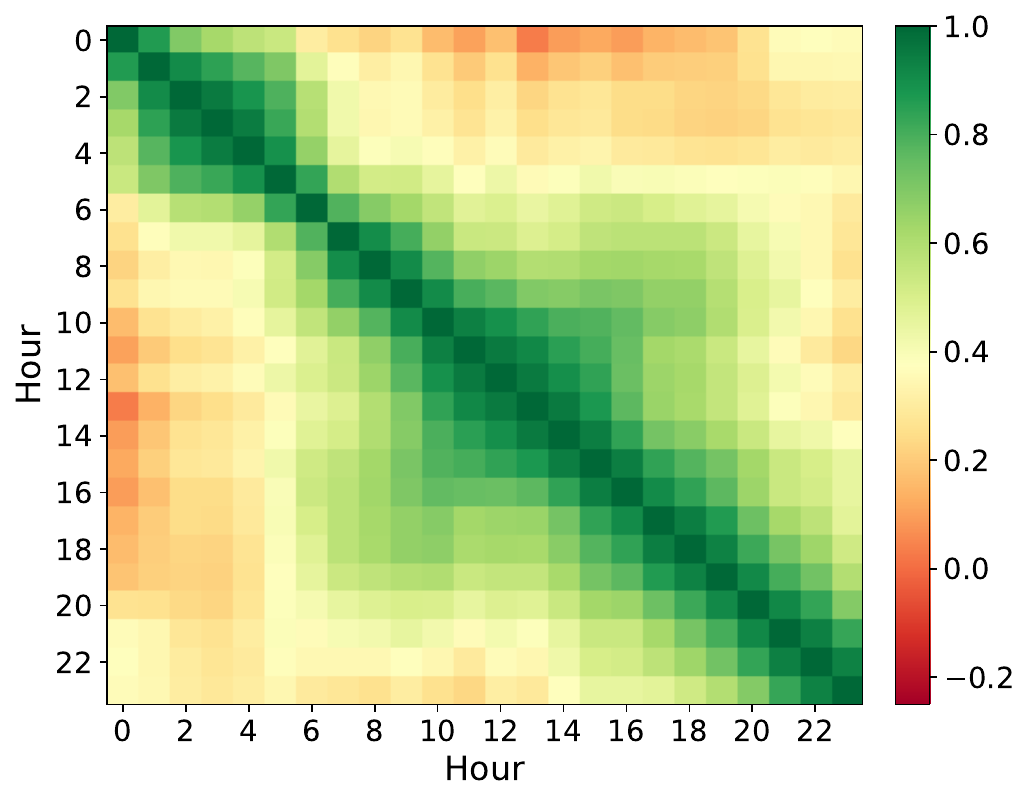}
    \end{subfigure}
    \hfill
    \begin{subfigure}[b]{0.32\textwidth}
        \centering
        \includegraphics[width=\textwidth]{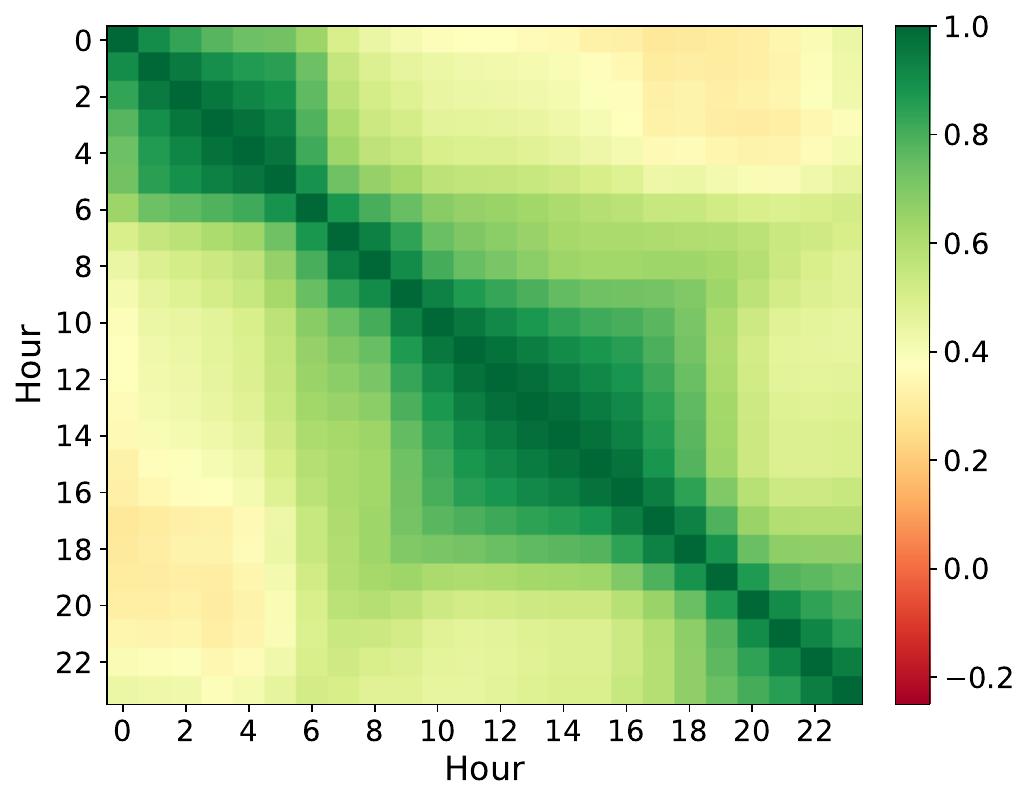}
    \end{subfigure} 
    \caption{Matrices depicting the realized correlations, computed via \eqref{eq:realized_corr}. Top panel shows the correlations computed via the raw RCV and the lower panel shows the correlations computed via the propagation adjusted RCV.}
    \label{fig:realized_correlation_matrices}
\end{figure}

\section{Decomposing the realized covariation}\label{sec:decomposing_rcv}
The raw estimate of the realized covariation matrix $\widehat{\Sigma}^w_{\mathrm{pl}}=\frac{1}{N_w}\sum_{j=1}^{N_w}\widehat{\Sigma}_{j,\mathrm{pl}}^w\in\mathbb{R}^{d\times d}$ reveals an interesting pattern. In order to assess the structure in more detail, we carry out a principal component decomposition of the form
\[
\widehat{\Sigma}_{\mathrm{pl}}^w=Q\Lambda Q^\top,
\]
where $Q^\top Q=Q Q^{\top}=I_d$, $\Lambda = \mathrm{diag}(\lambda_1,\ldots ,\lambda_d)$ with eigenvalues $\lambda_1\geq \cdots\geq \lambda_d\geq 0$, and $\widehat{\Sigma}_{\mathrm{pl}}^wq_k=\lambda_k q_k$. Equivalently, 
\[
\widehat{\Sigma}^w_{\mathrm{pl}} = \sum_{k=1}^{d}\lambda_kq_kq_k^\top,
\]
where $q_k\in\mathbb{R}^d$ are the principal directions and $\lambda_k$ the variance explained along that direction. The \emph{factor loadings} are then defined as $\sqrt{\lambda_k}q_k$ for $k=1,\ldots ,d$, ranked by the proportion of variance that they explain. The \emph{factor scores} are then the coordinates of each observation, $x_t$, in the factor directions such that the $k$'th factor score is $S_{k}(t)=q_k^\top x_t$. For all three generation zones, we find that between six and eight principal components are able to reproduce $95\%$ of the variance, indicating that a few common factors drive most of the volatility. We note that this amount of factors is relatively high compared to more traditional areas of finance, such as yield curve modeling, where three factors often suffice. This is not too surprising, since yield curves evolve much less erratically than electricity prices. 

The factor loadings of the estimated RCV in the German generation zone are depicted on Figure~\ref{fig:factor_loadings_DE_LU} for the first four factor loadings. These reveal an interesting structure which is remarkably similar across the three zones, where the corresponding plots for Norway and Spain are provided on Figures~\ref{fig:factor_loadings_NO_2} and \ref{fig:factor_loadings_ES} in Appendix~\ref{app:plots}. The first factor explains more than $60\%$ of the total variation in $\widehat{\Sigma}_{\mathrm{pl}}^w$ in all three zones, and we note that this can be interpreted as a ``level factor'', that describes the overall shape of the weekly RCV. The interpretation of a level factor is slightly different from many conventional settings, since $q_1$ is not flat across hours. Indeed, $q_1$ shows a distinct hump shape that rises from the early hours, peaks around 18-19 o'clock, and slightly decreases towards the late hours. This indicates that the mid-day and late hours generally exhibit higher levels of volatility. All of the entries of $q_1$ do, however, share the same sign, meaning that the first factor determines co-movements of the volatility across all hours. 

\begin{table}[t]
\footnotesize
\centering
\setlength{\tabcolsep}{3pt}
\renewcommand{\arraystretch}{0.9}
\caption{Regressions of $\log (S_1(t))$ on same-week fundamentals, using the propagation-adjusted RCV. All regressors are contemporaneous (same week). t-statistics are reported in parentheses and p-values computed from HAC standard errors (Newey--West, 7 lags). Significance: $^{***}$\,$p<0.01$, $^{**}$\,$p<0.05$, $^{*}$\,$p<0.10$. RES means ``renewable energy sources''.}
\label{tab:contemp_logIV_DE}
\begin{tabular}{@{}lcc@{}}
\toprule
 & RES + Price & All sources \\
\midrule
\grp{\textbf{Price level}} \\
$\overline{P}$     & \ct{0.0289$^{***}$}{18.55} & \ct{0.0257$^{***}$}{11.76} \\
$\overline{P}^2$   & \ct{$-$3.8e-5$^{***}$}{$-$9.18} & \ct{$-$3.4e-5$^{***}$}{$-$7.02} \\

\grp{\textbf{RES forecasts, PC1 scores}} \\
Load PC1           & \na & \ct{$-$6.4e-6}{$-$1.61} \\
Wind PC1           & \ct{2.1e-5$^{***}$}{14.46} & \ct{1.8e-4}{1.46} \\
Solar PC1          & \ct{1.8e-5$^{***}$}{11.83} & \ct{1.0e-5}{0.53} \\

\grp{\textbf{RES actuals and forecast errors}} \\
Wind actual        & \na & \ct{$-$8.3e-4}{$-$1.39} \\
Wind error         & \ct{1.5e-4$^{**}$}{2.56} & \ct{$-$7.2e-4}{$-$1.18} \\
Solar actual       & \na & \ct{$-$8.8e-5}{$-$0.65} \\
Solar error        & \ct{3.1e-4$^{**}$}{2.44} & \ct{2.2e-4}{1.18} \\

\grp{\textbf{Thermal generation}} \\
Lignite            & \na & \ct{$-$1.2e-4$^{***}$}{$-$4.44} \\
Hard coal          & \na & \ct{6.5e-6}{0.19} \\

\grp{\textbf{Hydro and other}} \\
Hydro pump.\ stor. & \na & \ct{8.3e-4$^{***}$}{5.51} \\
Hydro run-of-river & \na & \ct{$-$9.0e-4$^{***}$}{$-$5.20} \\
Biomass            & \na & \ct{$-$3.1e-4}{$-$1.36} \\
Geothermal         & \na & \ct{0.014}{1.49} \\
Waste              & \na & \ct{4.6e-4}{1.47} \\
Other              & \na & \ct{$-$5.8e-5}{$-$0.06} \\
Other renewable    & \na & \ct{$-$7.4e-3$^{**}$}{$-$2.51} \\
\midrule
$R^2$              & 0.730 & 0.813 \\
$N$                & 2{,}129 & 2{,}119 \\
\bottomrule
\end{tabular}
\end{table}

The remaining factors all have sign changes and therefore induce more sophisticated co-movement patterns. For example, we see that the third principal component in the German zone loads primarily onto the early hours, thus describing a substantial part of their variation. In contrast to the level factor, the third principal component induces negative correlation between the early and later hours. Interestingly, the second principal component seems to exhibit a shape close to the so-called \emph{duck curve}, which refers to the conventional shape of electricity demand, less the demand supplied by renewable resources. These patterns are similar in the Norwegian zone, whereas in the Spanish zone the role of the second and third principal components are seemingly swapped. These slight discrepancies could likely arise from structural and market specific features such as price coupling and generation variables, and a more detailed analysis would be interesting to pursue in the future. 

\begin{figure}
    \centering
    \includegraphics[width=0.9\linewidth]{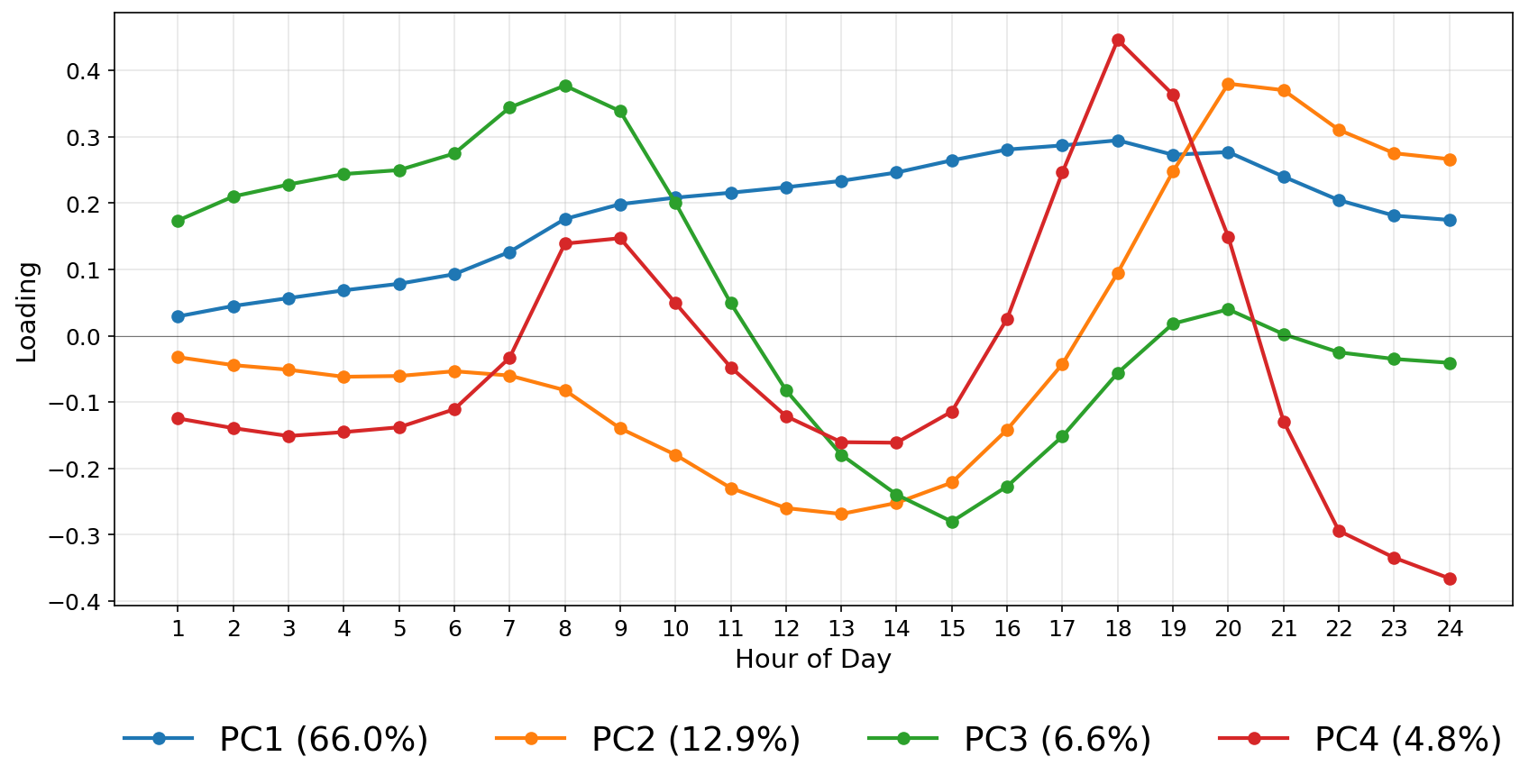}
    \caption{First four factor loadings in the German generation zone. Note that due to rotation invariance, the loadings are only identified up to a sign. The legends also denote the percentage of variance explained by each component.}
    \label{fig:factor_loadings_DE_LU}
\end{figure}

\subsection{Factor loading stability}
The principal components with loadings depicted on Figure~\ref{fig:factor_loadings_DE_LU} are extracted based on the full sample. It is therefore not guaranteed that any interpretations based on these loadings are valid at all points in time. On Figure~\ref{fig:loading_stability_DE}, we plot the weekly estimated factor loadings over time from the German generation zone and corresponding to the first four principal components (aligned to have same sign throughout). The same figure is reproduced for Norway and Spain on Figures~\ref{fig:loading_stability_NO_2} and \ref{fig:loading_stability_ES} in Appendix~\ref{app:plots}. Although there is some -- seemingly seasonal -- variation in the factor loadings, their overall \emph{shapes} are very consistent over time, especially for the first three factors. This suggests that the interpretations of the loadings are stable over time.

\subsection{The level factor}\label{sec:level_factor}
Since the level factor explains an excess of 60\% of the variation in the weekly RCV across all three zones under consideration, we consider this as the main driver of volatility. In order to assess what drives the level factor, and by extension volatility, we consider the factor score, $S_1(t)$, and regress $\log (S_1(t))$ on various variables relating to expected electricity demand and production, as well as the overall price level. The log-transform is chosen to suppress outliers, as the score has several large spikes; the log-transform is valid as the score is always positive throughout our datasets. The data is described in Appendix~\ref{app:data}. Each zone has day-ahead forecasted wind and solar generation (except NO2 which has no solar generation) and the actual realized production mix. We include rolling weekly averages of all these, as well as the PC1 factor scores of wind and solar forecasts in the regression, where we note that the first principal component of wind and solar forecasts explain 86\% and 95\% of their respective variance. We then control for price-scaling effects by including the rolling weekly average price $\overline{P}$ as well as $\overline{P}^2$. We also include the average forecast errors on wind and solar generation. We note that the regressions do not contain any information on what predicts volatility, but merely illustrates contemporaneous relationships between volatility and various generation and price variables. The results are depicted for the German generation zone in Table~\ref{tab:contemp_logIV_DE}, while similar tables are depicted for Norway and Spain in Appendix~\ref{app:plots}. We do not add seasonal dummies, as these are largely absorbed by the price level controls. We also note that the coefficients are not comparable across zones, as the amount of electricity generated by the same type of fuel varies greatly. 

\begin{figure}[t]
    \centering
    \begin{subfigure}[b]{0.4\textwidth}
        \includegraphics[width=\linewidth]{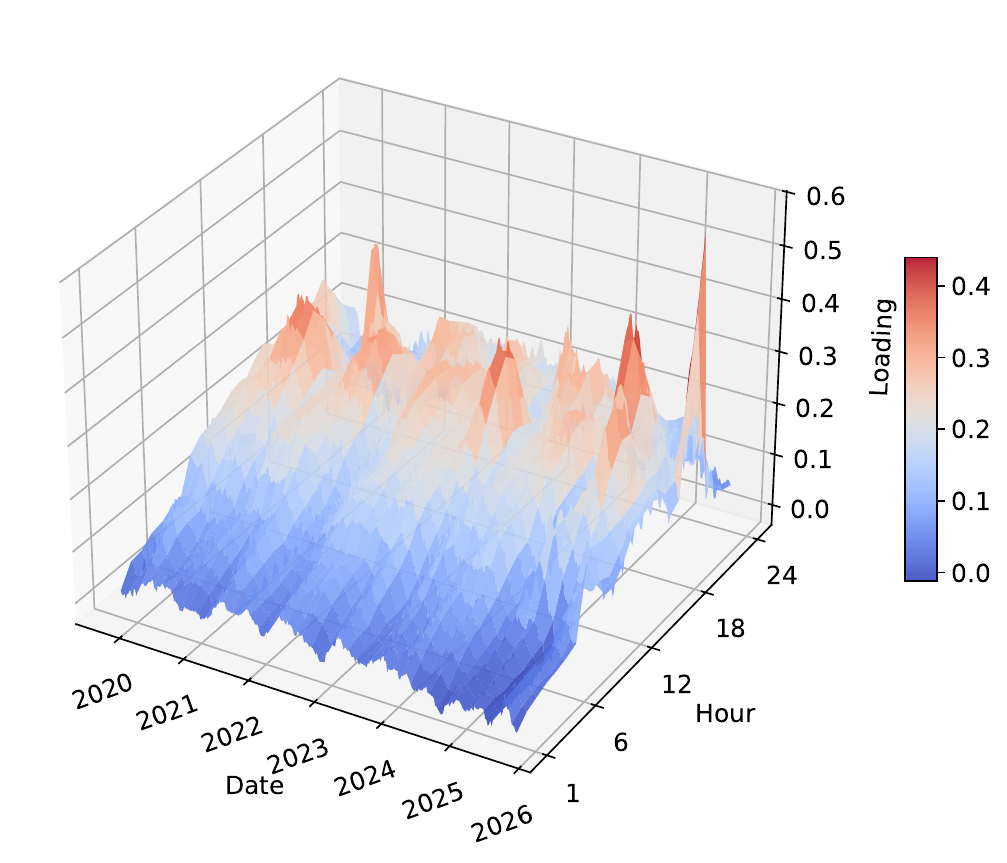}
        \caption{PC1}
    \end{subfigure}
    \hspace{0.04\textwidth}%
    \begin{subfigure}[b]{0.4\textwidth}
        \includegraphics[width=\linewidth]{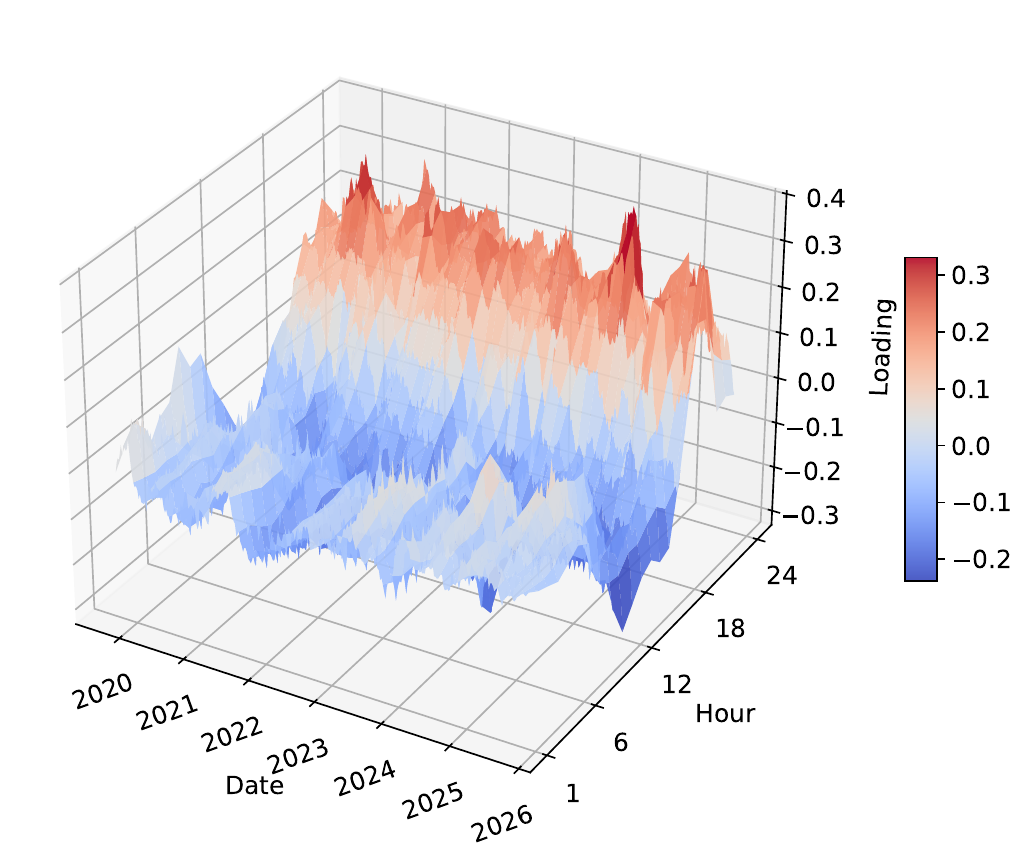}
        \caption{PC2}
    \end{subfigure}

    \vspace{0.3cm}

    \begin{subfigure}[b]{0.4\textwidth}
        \includegraphics[width=\linewidth]{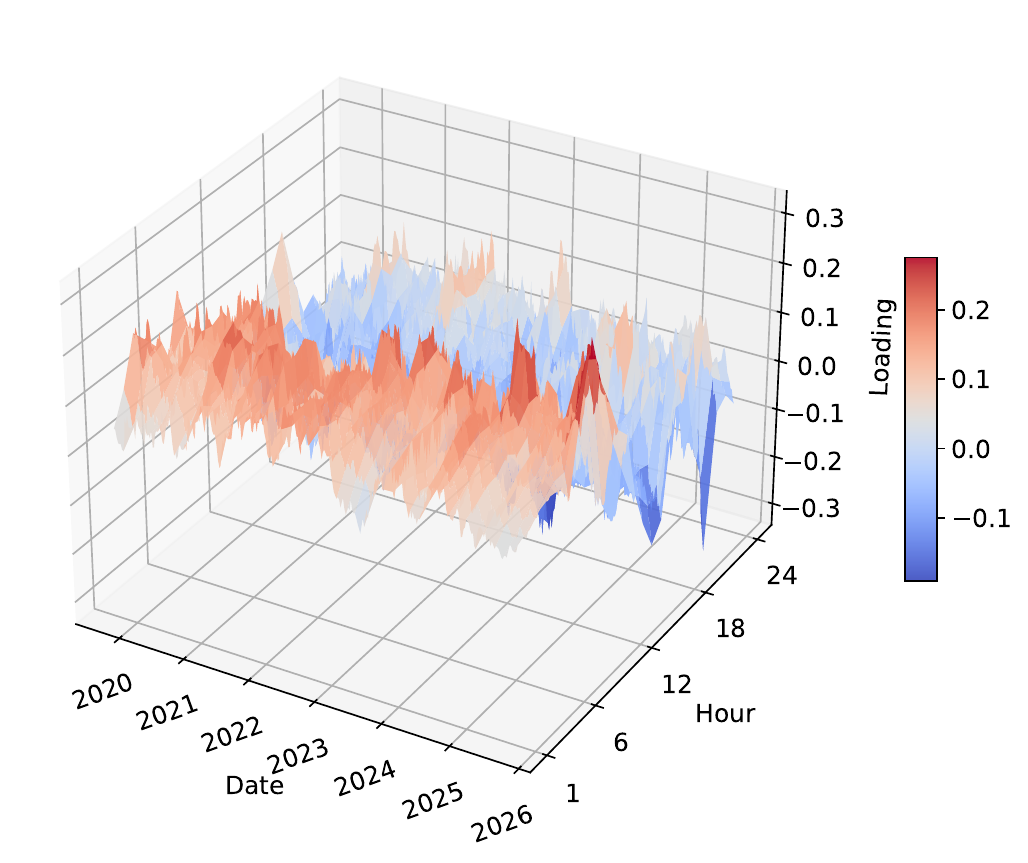}
        \caption{PC3}
    \end{subfigure}
    \hspace{0.04\textwidth}%
    \begin{subfigure}[b]{0.4\textwidth}
        \includegraphics[width=\linewidth]{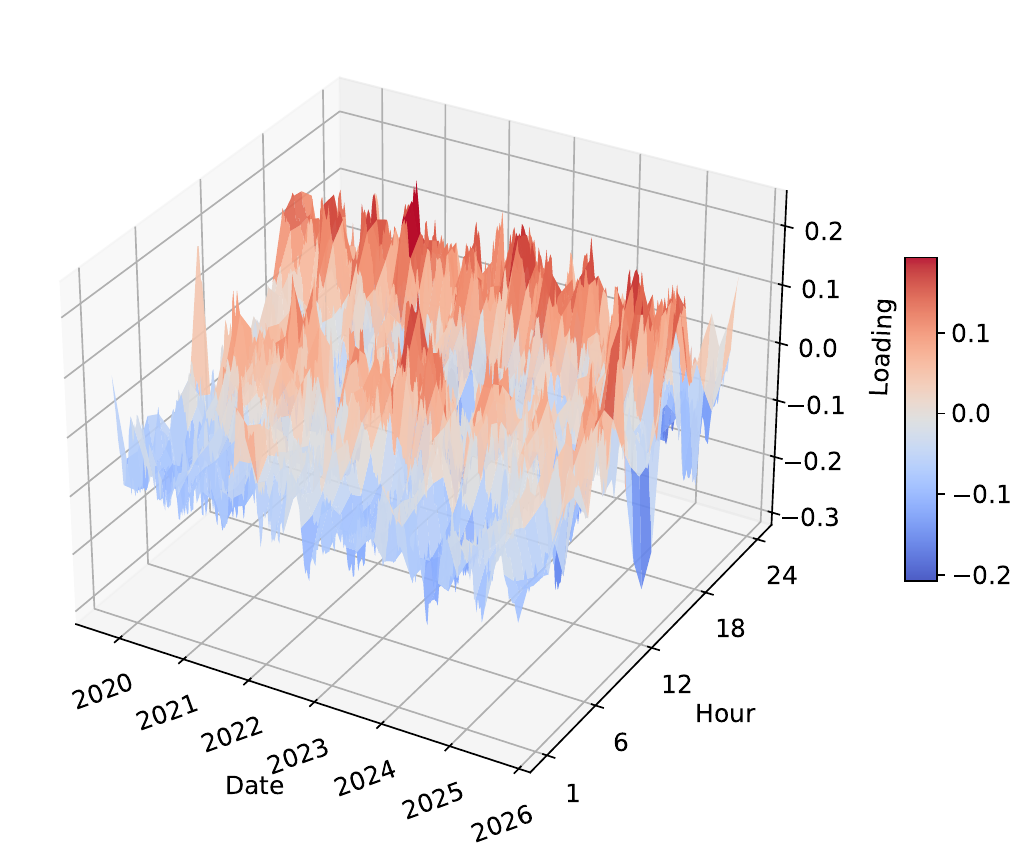}
        \caption{PC4}
    \end{subfigure}
    \caption{Factor loadings over time for the first four principal components in the German generation zone.}\label{fig:loading_stability_DE}
\end{figure}

The regression results reveal that the level of weekly RCV is very strongly tied to the price level itself in all zones. We refer to this as price/level-scaling and return to this issue when investigating leverage effects in Section~\ref{sec:inverse_leverage}. Both $\overline{P}$ and $\overline{P}^2$ are highly significant in all zones, and combined explain 54\%, 52\%, and 33\% of the variation in $\log(S_1(t))$ in Germany, Norway, and Spain, respectively. We note that in the full regression, almost all variables relating to renewable production are statistically insignificant, even the forecast errors. A notable exception is that of hydro generation, which is highly significant and adds a large amount of explainability in the Norwegian zone in particular. Overall, we achieve very high $R^2$ statistics in all zones, particularly Germany and Spain which both have $R^2>80\%$. By investigating the sign of the regression coefficients, we see that high forecasts on renewable energy sources are associated with higher price volatility, but that high realized renewable production contributes with different effects to the price volatility. For example, high realized wind production generally has a negative effect on the level of volatility. The remaining renewable sources such as hydro, biomass, and geothermal energy have varying effects and significance across the zones, which reflect their relative importance in the specific markets. Conventional fuel sources such as oil and gas contribute with opposite signs; oil generation is generally associated with lower volatility  while gas generation is generally associated with higher volatility.

We conclude that, overall, it seems electricity price volatility is highly and positively correlated with the spot price level itself and that the realized production mix is important in explaining volatility dynamics, but the relative importance and contributions of the production mix components are highly individual to the market under consideration.

\subsection{Investigating the propagation share}\label{sec:propagation}
In this Section, we briefly turn to the propagation share across hours, which exhibits a decreasing pattern. This does not inherently mean that early hours are less volatile in absolute terms and we therefore depict the sum $(\widehat{\Sigma}^w_{\mathrm{pl}})^{(h,h)}+\frac{1}{N}\sum_{n}\widehat{B}_n^{(h)}$ on Figure~\ref{fig:propagation_DE_LU}. This shows both the relative and absolute share of volatility versus propagation/mean-reversion and clearly indicates that the higher propagation share of early hours coincides with lower volatility levels in the German generation zone. Similar plots for Norway and Spain are provided on Figures~\ref{fig:propagation_NO_2} and \ref{fig:propagation_ES} in Appendix~\ref{app:plots}. As hinted earlier, a possible explanation for the decreasing propagation share and predictability of volatility over the hours of the day could be the temporal distance from the time of price settlement to the time of delivery, which inherently induces more uncertainty.
\begin{figure}
    \centering
    \begin{subfigure}[b]{0.4\textwidth}
        \centering
        \includegraphics[width=\textwidth]{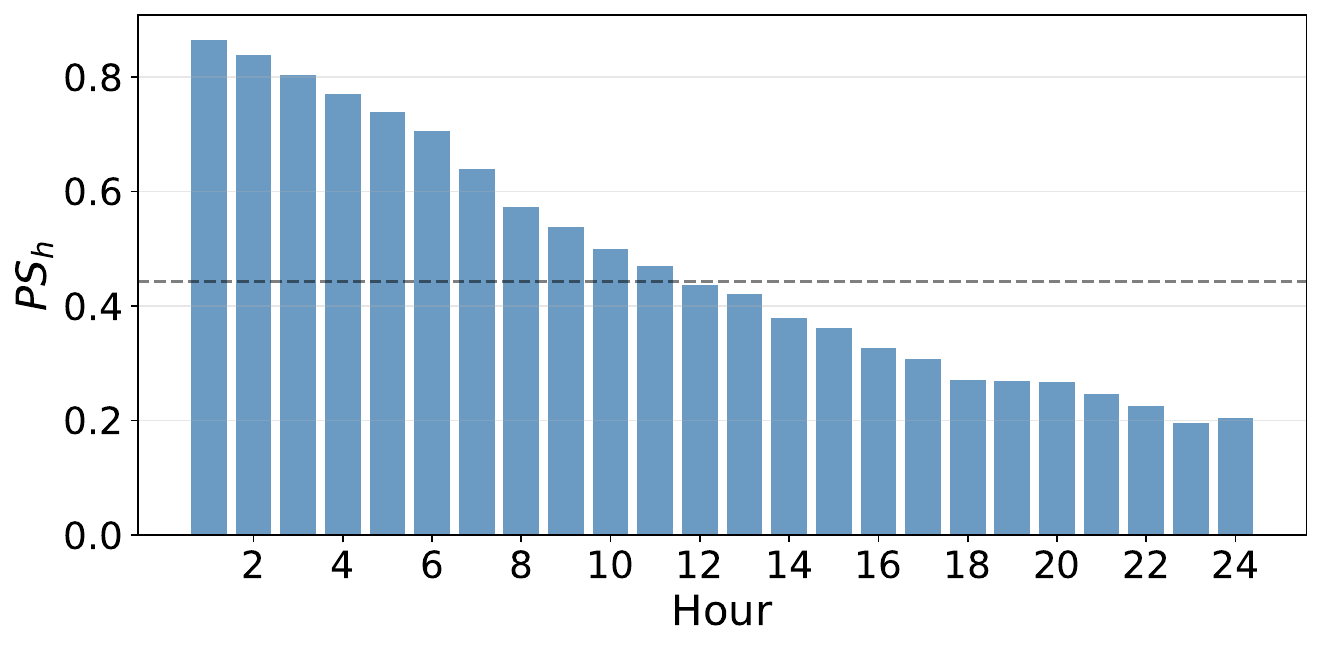}
        \subcaption{Propagation share per hour.}
    \end{subfigure}
    \hspace{0.04\linewidth}
    \begin{subfigure}[b]{0.4\textwidth}
        \centering
        \includegraphics[width=\textwidth]{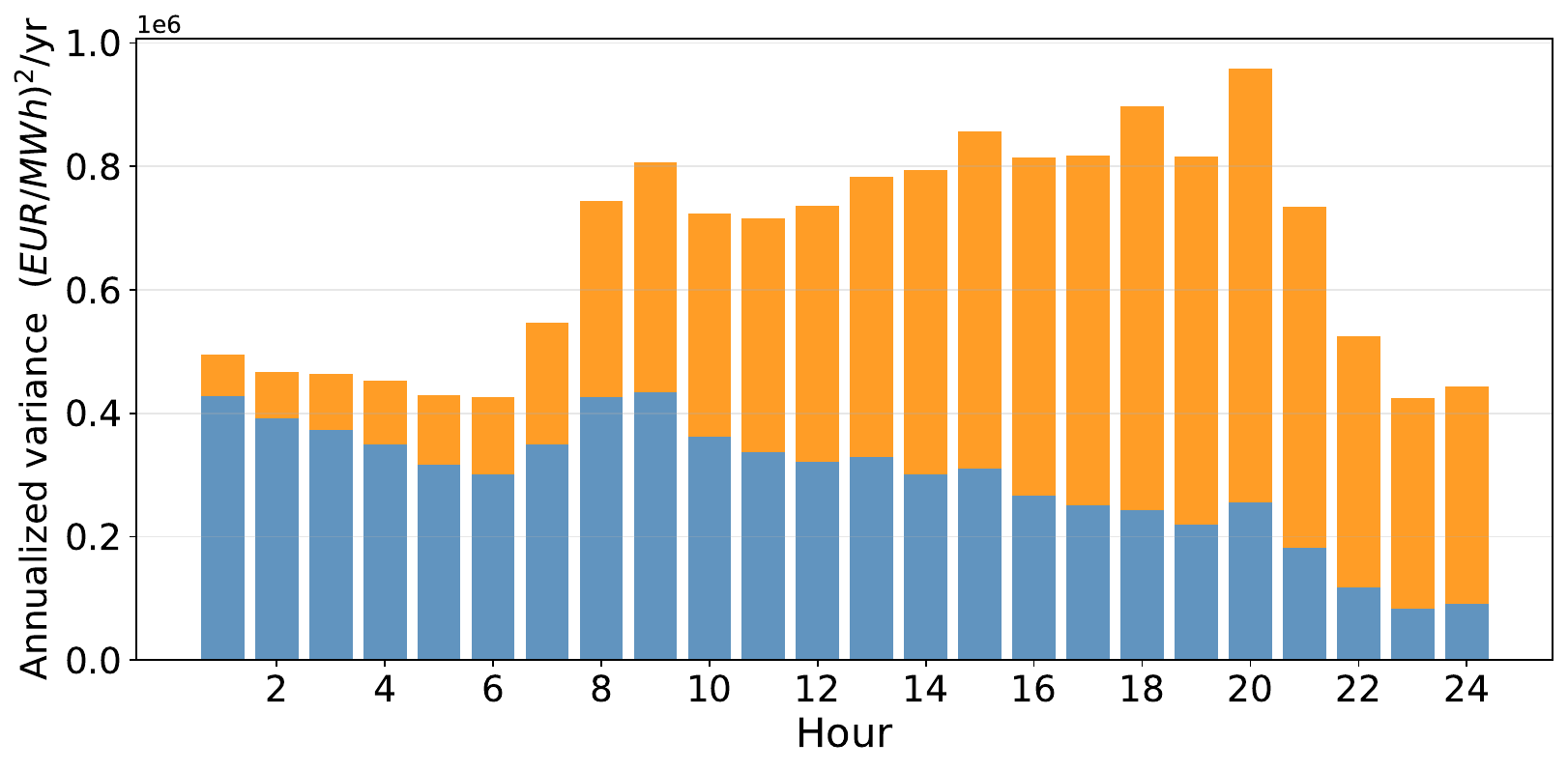}
        \subcaption{Propagation versus innovation.}
    \end{subfigure}
    \caption{Propagation share per hour in the German generation zone, with the total average overlayed as a black line (left). Annualized total RCV in the German generation zone with the blue bottom bar indicating the propagation effect and the orange top bar indicating the innovation effect (right).}
    \label{fig:propagation_DE_LU}
\end{figure}
To assess whether this is the case, we run the following regression
\[
PS_h = \alpha +\beta_0h + \beta_1e_h^{\mathrm{wind}}+\beta_2e_h^{\mathrm{solar}},
\]
where $e_h^{\mathrm{wind}}$ is the mean squared error of the standardized (zero mean, unit variance) wind production forecast versus actual wind production in hour $h$ across the whole dataset, and similar for solar. It turns out that the wind forecast error is almost perfectly collinear with the hour, and as such this adds little value in the regression. The solar forecast errors are concentrated around the peak solar generation hours (mid-day) and as such adds value beyond just the hourly ordinals. Since the Norwegian zone has no solar generation, we only consider Germany and Spain with the regression results tabulated in Table~\ref{tab:PS_uncertainty_regression}. The hour ordinal alone explains 96\% of the total variation while the marginal contribution of solar production after controlling for hour is robust and identifies a separable role for forecast-error uncertainty. In both the German and Spanish zones, a one standard deviation increase in the per-hour solar MSE is associated to an decrease of approximately 4 percentage points in the propagation share. While the simple regression is not conclusive evidence, we conjecture that forecast uncertainty and time-to-delivery effects are important in explaining the increasing volatility and unpredictability of prices over the day.

\begin{table}[t]
\footnotesize
\centering
\caption{OLS regression of the per-hour propagation share $PS_h$ on the hour ordinal and per-hour mean squared forecast errors of wind and solar
generation. Predictors are standardized to zero mean and unit
variance, so the coefficients are directly comparable across regressors and
zones. Significance: $^{***}$\,$p<0.01$, $^{**}$\,$p<0.05$, $^{*}$\,$p<0.10$.}
\label{tab:PS_uncertainty_regression}
\setlength{\tabcolsep}{6pt}
\begin{tabular}{l cc cc}
\toprule
 & \multicolumn{2}{c}{Germany (DE-LU)} & \multicolumn{2}{c}{Spain (ES)} \\
\cmidrule(lr){2-3} \cmidrule(lr){4-5}
 & Coef. & $t$-stat & Coef. & $t$-stat \\
\midrule
$\alpha$  (constant)             &  $0.473^{***}$  &  $87.70$  &  $0.389^{***}$  & $105.40$ \\
$\beta_0$ (hour)                  & $-0.247^{***}$  & $-13.56$  & $-0.153^{***}$  & $-35.48$ \\
$\beta_1$ (wind MSE, $e_h^{\mathrm{wind}}$)
                                  & $+0.042^{**}$   &  $+2.23$  & $+0.010^{*}$    &  $+1.98$ \\
$\beta_2$ (solar MSE, $e_h^{\mathrm{solar}}$)
                                  & $-0.039^{***}$  &  $-6.08$  & $-0.036^{***}$  &  $-7.39$ \\
\midrule
$R^2$         & \multicolumn{2}{c}{$0.987$} & \multicolumn{2}{c}{$0.990$} \\
adj.\ $R^2$   & \multicolumn{2}{c}{$0.986$} & \multicolumn{2}{c}{$0.989$} \\
$N$           & \multicolumn{2}{c}{$24$}    & \multicolumn{2}{c}{$24$}    \\
\bottomrule
\end{tabular}
\normalsize
\end{table}

\section{The inverse leverage effect}\label{sec:inverse_leverage}
The \emph{leverage effect} typically refers to the negative correlation between an asset return and its volatility. In particular, this means that there is an asymmetric association between increases and decreases in volatility. It has been suggested in the literature that electricity prices in some markets seem to exhibit an inverse leverage effect, which is to say that price shocks are positively correlated with shocks to volatility \citep[see, e.g.,][]{InverseLeverage,JanczuraWeron2010}. To the best of our knowledge, most findings of inverse leverage seem to be from estimation or calibration of a specific model and has not been studied via rigorous and non-parametric volatility estimates, and we therefore consider this effect in some detail. We distinguish between two types of leverage effect, which we refer to as weak and strong form, respectively:
\begin{itemize}
    \item \emph{Weak form}: Changes in integrated variance responds asymmetrically to price changes.
    \item \emph{Strong form}: The asymmetry is not explained by price-level scaling or mean-reversion in volatility.
\end{itemize}
The strong form of leverage is essentially a conditional version of the weak form. Indeed, given the findings in Section~\ref{sec:rcv_of_prices}, we expect high absolute price levels to lead to large deviations around these levels, which we refer to as price-level scaling of the volatility. Similarly, when the price level is large, we expect that it is more likely to shrink fast towards the long-run mean. The strong form leverage imposes that we observe a correlation between volatility and price changes even after accounting for these effects. 

\subsection{Leverage effects in the level factor}\label{sec:leverage_level_factor}
As in Subsection~\ref{sec:level_factor}, we consider the score of the level factor as a proxy for the overall level of volatility. Denoting by $\overline{P}_t=\frac{1}{7d}\sum_{i=1}^d\sum_{j=1}^{7}P_{t+1-j}^{(i)}$ the rolling weekly average and by $\Delta\overline{P}_t=\overline{P}_t-\overline{P}_{t-7}$ the non-overlapping weekly difference, we then construct estimates of $\Delta\log(\widehat{S}_1(t))=\mathbb{E}\left[ \log(S_1(t))-\log(S_1(t-7))\mid \Delta \overline{P}_{t-7} \right]$ by partitioning the range of lagged price changes into 20 equal-frequency bins $B_b$ for $b=1,\ldots ,20$ and estimating the conditional mean within each bin, $\mu_b$ via the regression
\[
\Delta\log (S_1(t)) = \sum_{b=1}^{20} \mu_b\mathbf{1}_{\lbrace \Delta\overline{P}_{(t-7)} \in \mathcal{B}_b\rbrace} + \varepsilon(t).
\]
Hence, $\Delta\log(\widehat{S}_1(t))$ estimates the expected change in the log-level of RCV one week ahead, conditional on the price change observed the previous week. Standard errors are computed using the Newey-West (HAC) estimator with 14 lags to account for serial dependence induced by the overlapping rolling window. This is similar to the ``news impact curve'' construction in, e.g., \citet{EngleNg1993}. We construct four progressive specifications of news impact curves to isolate the leverage effect from confounding dynamics. In each case, both $\Delta \log(\widehat{S}_1(t))$ and $\Delta \overline{P}_{t-7}$ are residualized on a common subset of controls before estimation as follows (residualizing the price level ensures that the bins carry no information about the controls)
\begin{enumerate}
    \item No controls. 
    \item Price/level-scaling control via $\sinh^{-1}(\overline{P}_{t-7})$.
    \item Mean-reversion control via $\log (S_1(t-7))$.
    \item Joint price/level-scaling and mean-reversion control.
\end{enumerate}
The results of this construction are depicted for the German market in Figure~\ref{fig:leverage_nic_DE}, where we also include a regression line of the form
\begin{equation}\label{eq:nic_regression_beta}
\Delta\log (\widehat{S}_1(t)) = \beta_0 + \beta^+ \max\lbrace \Delta\overline{P}_{t-7}, 0\rbrace + \beta^- \min\lbrace \Delta\overline{P}_{t-7}, 0\rbrace + \varepsilon(t).
\end{equation}
The figure reveals that, unconditionally, there seems to be some asymmetry, but that many of the estimated $\widehat{\mu}_b$ for $\overline{P}_{t-7}>0$ are not statistically different from 0. A standard Wald test rejects $\beta_+=\beta_-$ at the 5\% level, suggesting some asymmetry. The price/level-scaling control dampens some of the effect and seemingly induces asymmetry that favours co-movements between price and volatility, while the mean-reversion control induces a very clear asymmetry that favours an inverse leverage effect. The joint control, however, essentially flattens the news impact curve to the point where there is no asymmetry, thus indicating that, upon controlling for the current state, most of the price/volatility co-movement disappears. The finding of no asymmetry is supported by the Wald test statistic, which is now unable to reject $\beta_+=\beta_-$. The corresponding news impact curve plots are depicted for Norway and Spain on Figures~\ref{fig:leverage_nic_NO2} and \ref{fig:leverage_nic_ES} in Appendix~\ref{app:plots}, which reveal different patterns, but with the same ultimate conclusion that, after the joint control, any asymmetry is eliminated. These findings are very interesting, since it seems that asymmetric relationships between volatility and price changes in electricity markets can largely be explained by the price and volatility levels themselves at the aggregate level. Hence, when modeling price averages and their volatility, it seems that state-dependence in price and volatility is more relevant than correlation of innovations. 

\begin{remark}
    One may wonder if negative prices exhibit a different kind of leverage effect than positive prices. Although negative prices occur somewhat frequently at the hourly level in some generation zones, the average daily price is never negative in any of the zones throughout the dataset. Hence we can get no such information, since there are no negative values to study. 
\end{remark}

\begin{figure}[t]
    \centering
    \begin{subfigure}[t]{0.4\textwidth}
        \includegraphics[width=\linewidth]{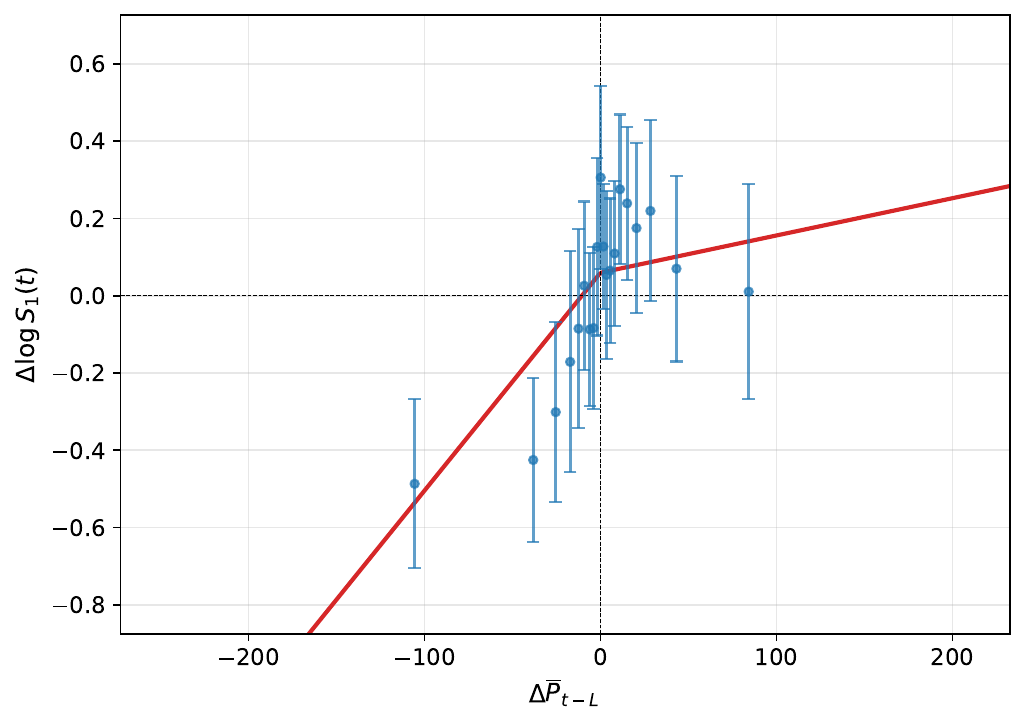}
        \caption{Unconditional regression, Wald test p-value for $\beta_1=\beta_2$ of $p=0.0431$.}
    \end{subfigure}
    \hspace{0.04\linewidth}
    \begin{subfigure}[t]{0.4\textwidth}
        \includegraphics[width=\linewidth]{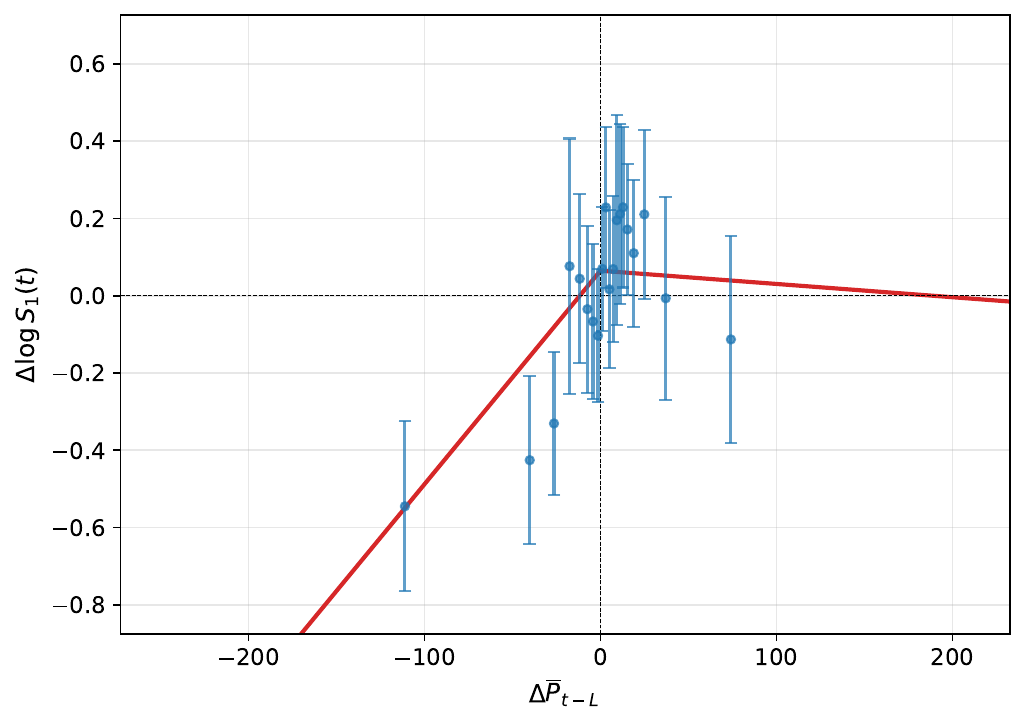}
        \caption{Price/level-scaling control regression, Wald test p-value for $\beta_1=\beta_2$ of $p=0.0128$.}
    \end{subfigure}

    \vspace{0.3cm}

    \begin{subfigure}[t]{0.4\textwidth}
        \includegraphics[width=\linewidth]{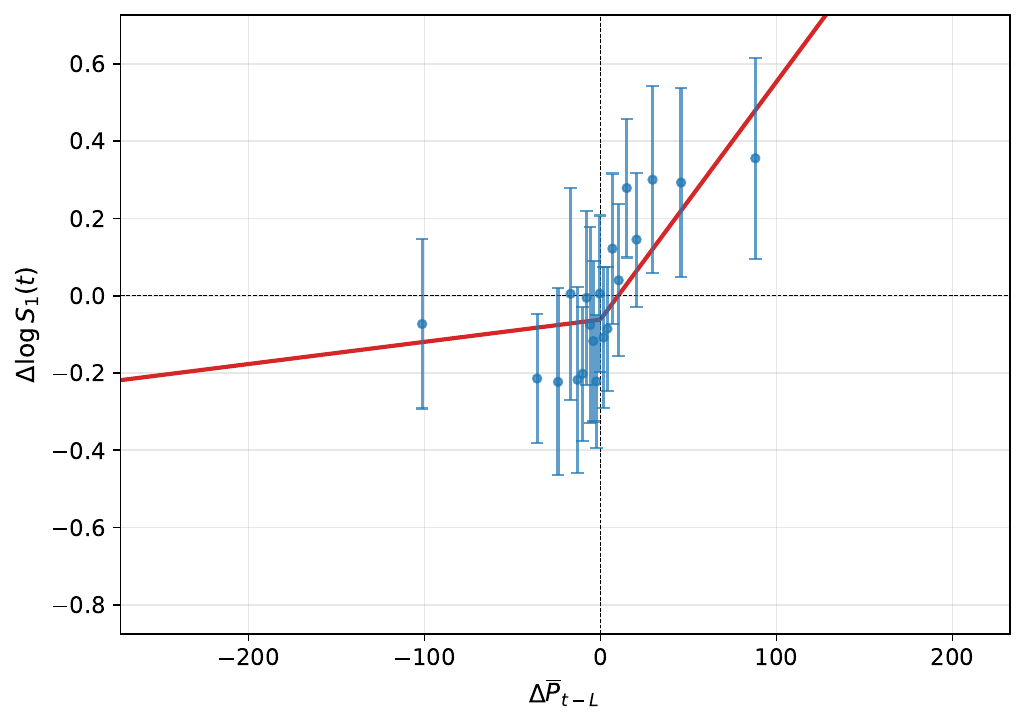}
        \caption{Mean-reversion control regression, Wald test p-value for $\beta_1=\beta_2$ of $p=0.0062$.}
    \end{subfigure}
    \hspace{0.04\linewidth}
    \begin{subfigure}[t]{0.4\textwidth}
        \includegraphics[width=\linewidth]{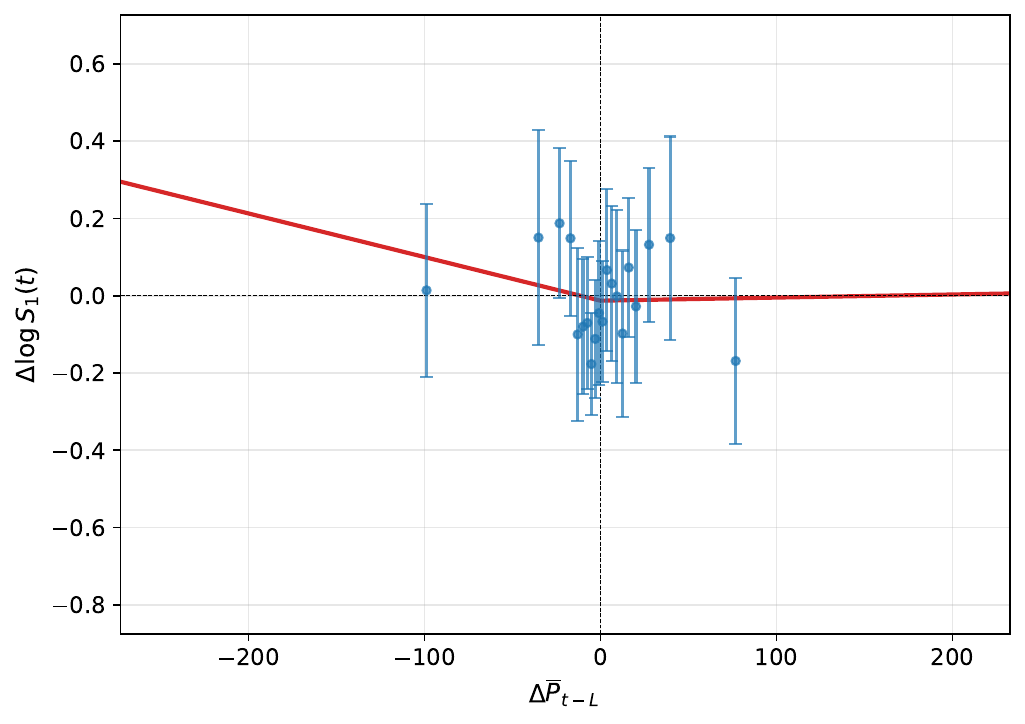}
        \caption{Joint control regression, Wald test p-value for $\beta_1=\beta_2$ of $p=0.5898$.}
    \end{subfigure}
    \caption{Estimated impact of price changes on $\Delta \log(S_1(t))$ in the German generation zone. Blue dots represent the $\mu_b$ coefficients with $95\%$ confidence bands constructed with HAC standard errors. The red line is the regression line.}\label{fig:leverage_nic_DE}
\end{figure}

\subsection{Functional leverage curves}
So far, we have only viewed the leverage effect through the lens of the factor scores. While informative, this approach is unable to detect more subtle effects, such as leverage effects in individual hours, which may be washed away in the aggregation. In this Section, we therefore take a closer look at the finer structure of leverage on the level of individual hours. To do so, we let $IV_{j}^{(h)}=\widehat{\Sigma}_{\mathrm{pl}}^w(h,h)$ and run the regression
\[ 
IV^{(h)}_j = \beta_0^{(h)} + \beta_+^{(h)} \max\lbrace\Delta\overline{P}_{j-w},0\rbrace + \beta_{-}^{(h)} \min \lbrace \Delta\overline{P}_{j-w},0\rbrace + \varepsilon_h(j),
\]
where $\overline{P}$ again denotes the weekly average price. From this we construct a \emph{functional leverage curve} $h \mapsto \left(\beta_+^{(h)}, \beta_-^{(h)}\right)$ showing the diurnal pattern by which volatility across the day responds to price shocks. In the absence of leverage, we expect these curves to be overlaying, and if there is no difference across the hours of the day, we expect the curves to be flat. Keeping in line with the previous section, we conduct an unconditional weak form regression and a conditional strong form regression, where the hourly price level and mean-reversion effects have been partialled out. We depict the functional leverage curves on Figure~\ref{fig:functional_leverage_curve}, where a similar pattern as observed in the aggregated case emerges. Only few hours exhibit weak form leverage, and any asymmetry disappears after controlling for price/level-scaling and volatility mean-reversion across all zones. Furthermore, most estimates of $\beta_+^{(h)}$ and $\beta_-^{(h)}$ are not statistically different from zero. This supports the findings of Section~\ref{sec:leverage_level_factor} and suggests that leverage effects generally can be explained by other mechanical effects and that there is no meaningful difference across delivery hours.

\begin{figure}[h]
    \centering

    \makebox[0.32\textwidth]{\textbf{Germany}}\hfill
    \makebox[0.32\textwidth]{\textbf{Norway}}\hfill
    \makebox[0.32\textwidth]{\textbf{Spain}}\\[4pt]

    \begin{subfigure}[t]{0.32\textwidth}
        \includegraphics[width=\textwidth]{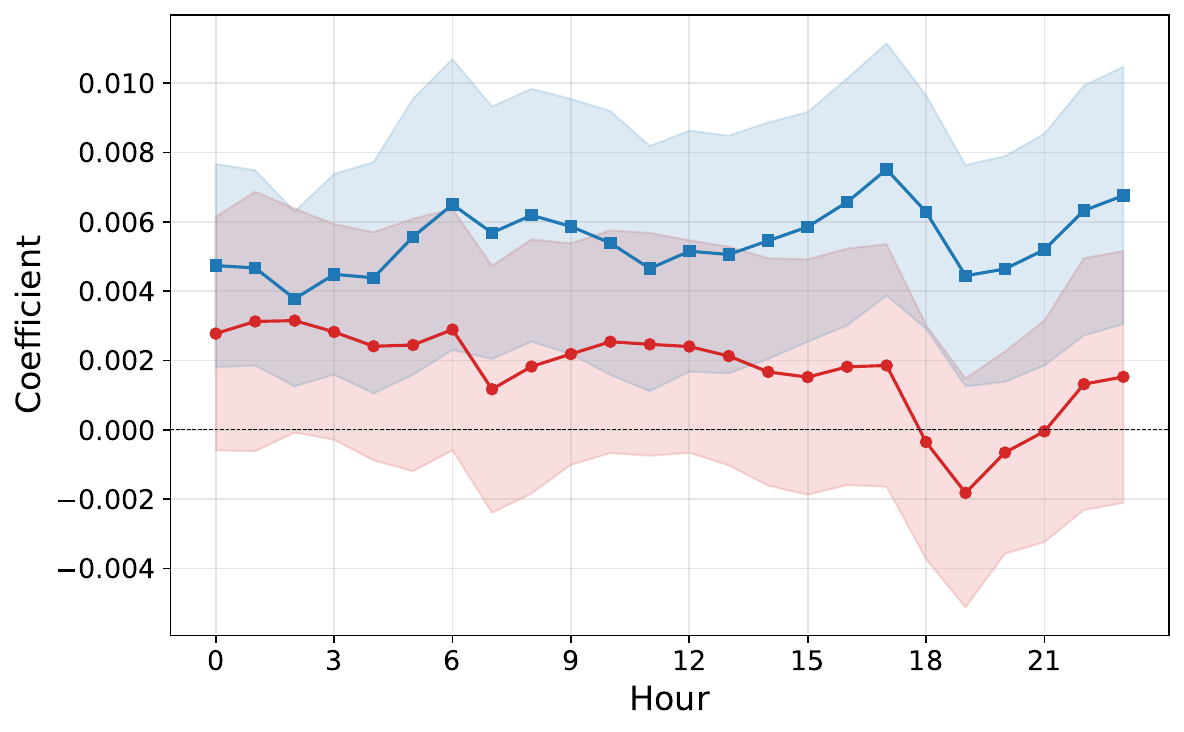}
    \end{subfigure}\hfill
    \begin{subfigure}[t]{0.32\textwidth}
        \includegraphics[width=\textwidth]{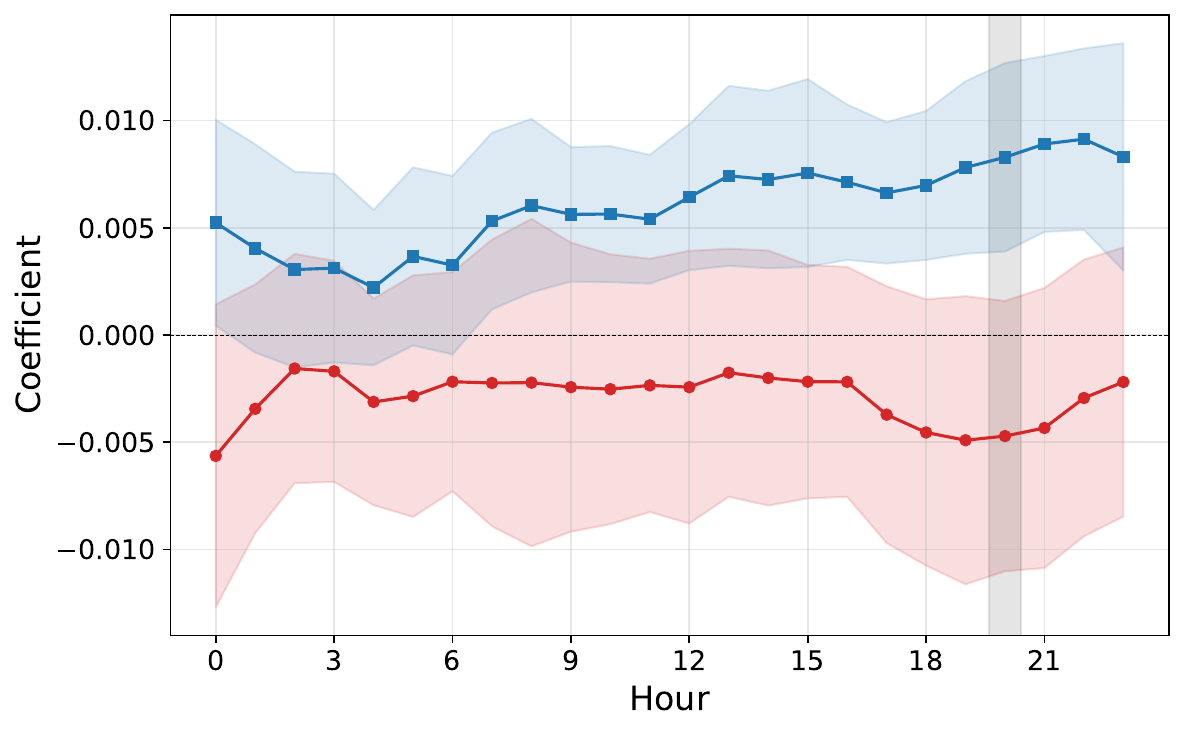}
    \end{subfigure}\hfill
    \begin{subfigure}[t]{0.32\textwidth}
        \includegraphics[width=\textwidth]{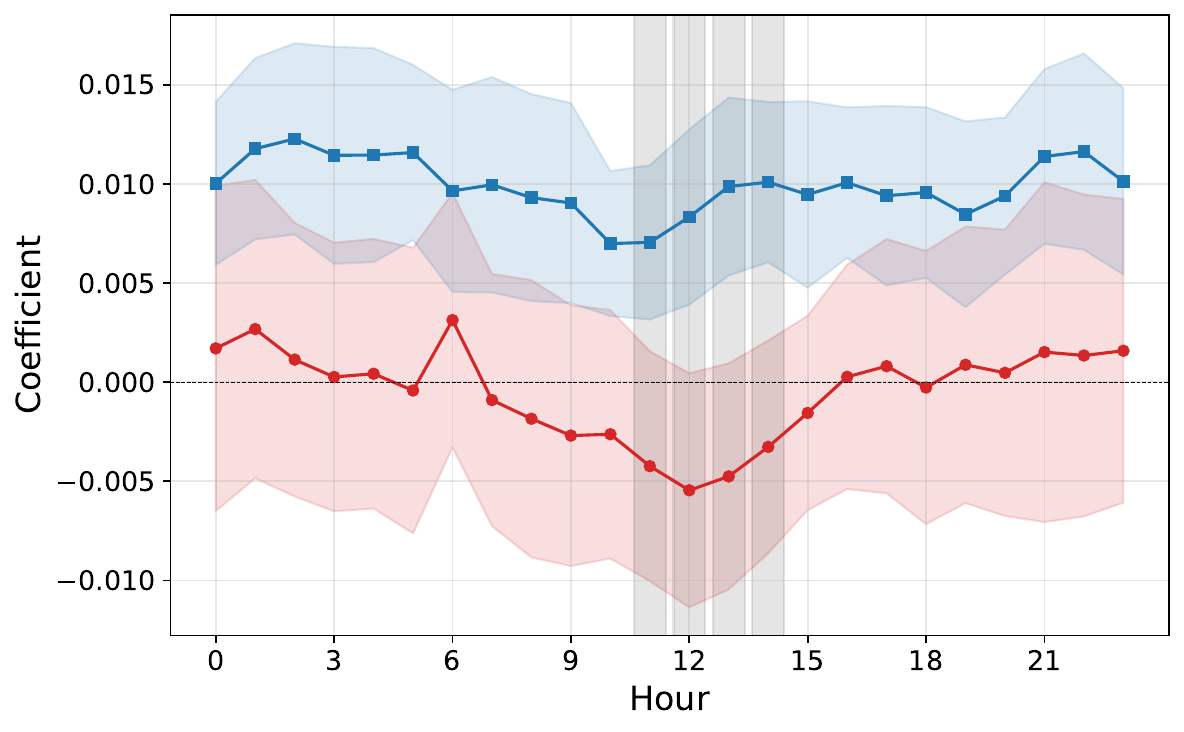}
    \end{subfigure}

    \vspace{6pt}

    \begin{subfigure}[t]{0.32\textwidth}
        \includegraphics[width=\textwidth]{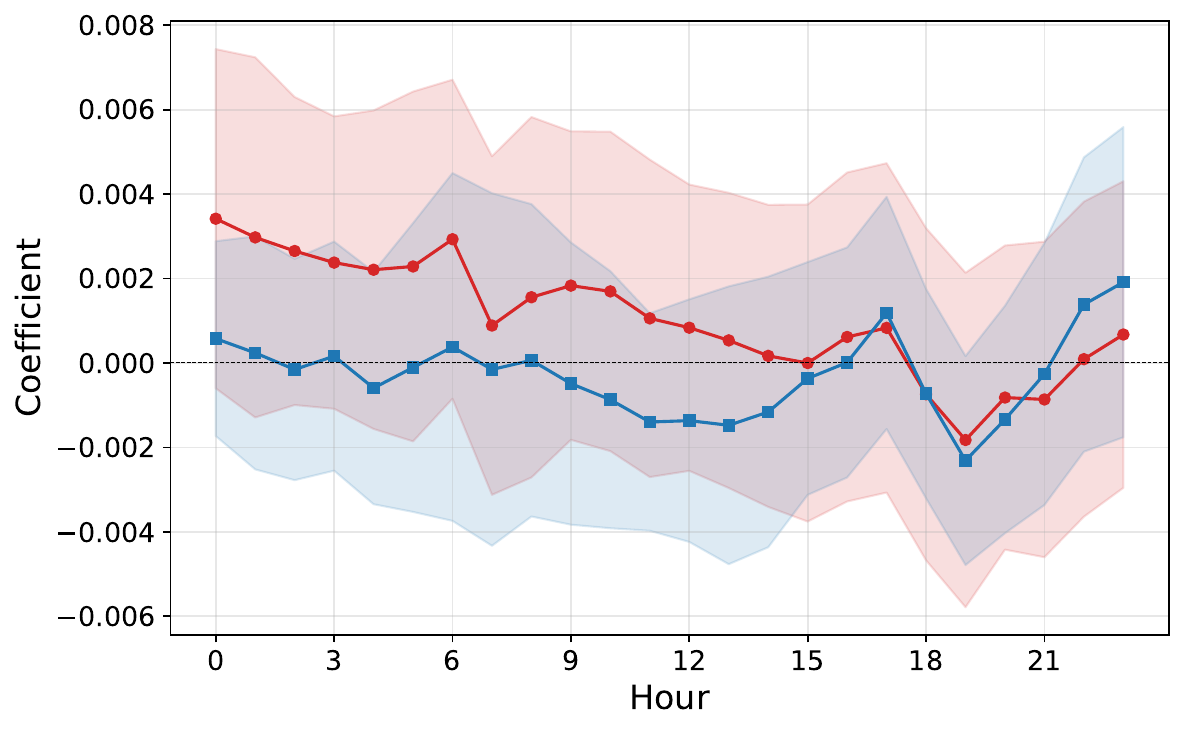}
    \end{subfigure}\hfill
    \begin{subfigure}[t]{0.32\textwidth}
        \includegraphics[width=\textwidth]{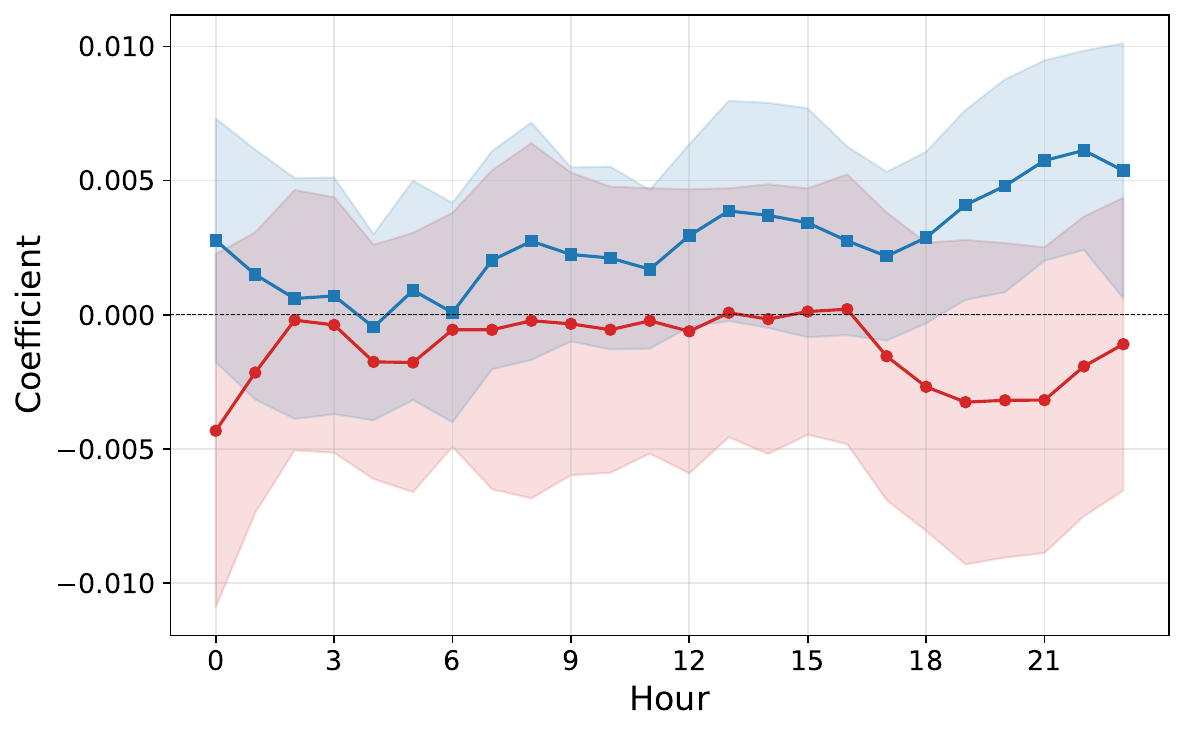}
    \end{subfigure}\hfill
    \begin{subfigure}[t]{0.32\textwidth}
        \includegraphics[width=\textwidth]{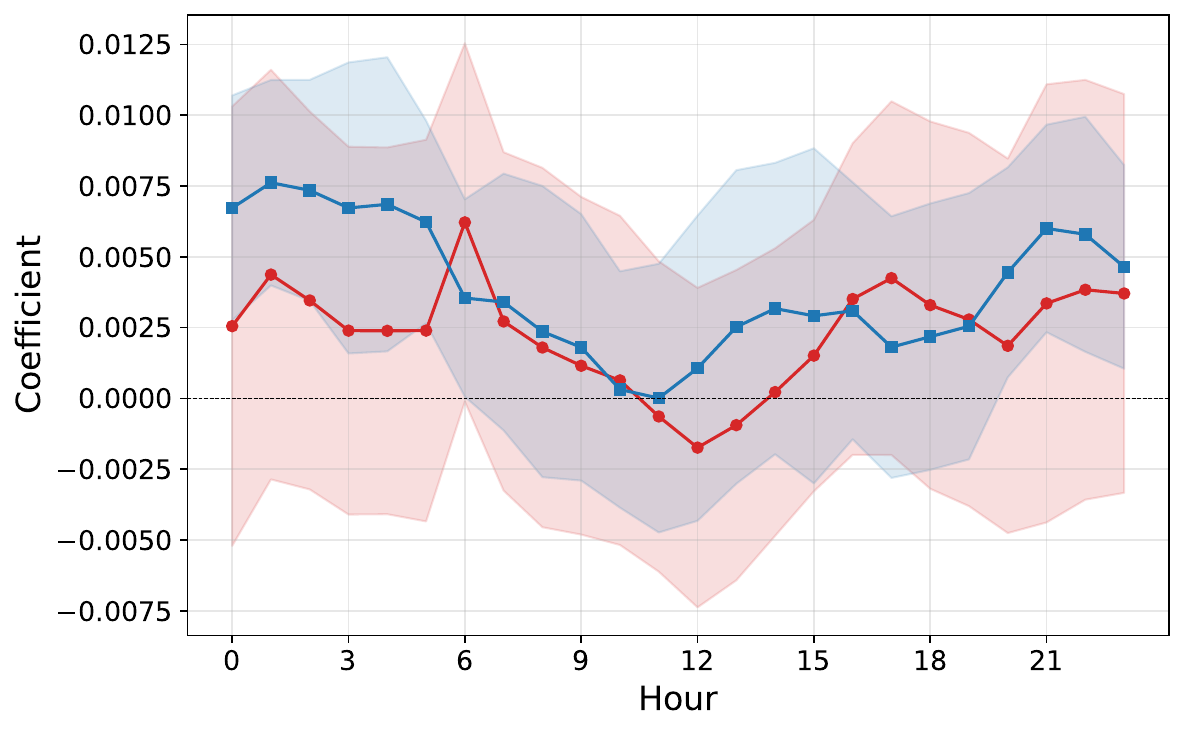}
    \end{subfigure}

    \caption{\textbf{Top:} Unconditional functional leverage curves with 95\% confidence bands superimposed. \textbf{Bottom:} Conditional functional leverage curves with 95\% confidence bands superimposed. The curve corresponding to $\beta_+^{(h)}$ is in red, while the curve corresponding to $\beta_-^{(h)}$ is in blue. Hours where $\beta_+^{(h)}$ is significantly different from $\beta_-^{(h)}$ are marked with a grey bar.}
    \label{fig:functional_leverage_curve}
\end{figure}

\section{Conclusion}\label{sec:conclusion}
This paper tackles the challenging problem of volatility estimation in the panel of electricity spot prices. By interpreting the panel as local averages of a latent function-valued price process, we construct an easy to implement estimator of the integrated conditional covariance operator that admits a simple decomposition into drift, propagation, and innovation terms. The estimator reveals a complex structure in the weekly realized covariation of electricity spot prices throughout several generation zones. Mean-reversion effects play a significant role in the weekly price variation and must be accounted for; in doing so, we find that mean-reversion is seemingly stronger in the early hours of the day. Peak-load and evening periods are associated with higher levels of overall volatility and we have documented that this correlates highly with the price level itself. Additionally, the realized production mix is important in explaining volatility dynamics, but the relative importance of renewables, fossil fuels, etc., varies greatly between generation zones. Consequently, each generation zone requires individual modeling tools and there is no ``one size fits all'' model across markets. We have zoomed in on the so-called inverse leverage effect, which is the notion that shocks to spot prices are positively correlated with shocks to volatility. We find evidence that such inverse leverage effects can be spurious and naturally arise from the strong state-dependence between volatility and price levels as well as mean-reversion in volatility itself. 

Our findings open up several venues of future research into the volatility of electricity markets. In particular, it will be interesting to look deeper into the interplay of volatility and renewable generation, seasonal and diurnal patterns, and forecasting of the RCV estimator.

\section{Disclosure statement}
The authors declare no conflicts of interest.

\section{Funding}
T.~K.~Kloster gratefully acknowledges financial support from the Center of Research in Energy: Economics and Markets and The Danish Council of Independent Research under DFF grant 10.46540/5247-00005B. F.~E.~Benth gratefully acknowledges financial support from the SURE-AI Centre grant 357482, Research Council of Norway. 

\printbibliography

\appendix

\section{Proofs}\label{app:proofs}
\paragraph{Proof of Proposition~\ref{prop:increment_decomp}}
\begin{proof}
Let $M_n=\mathbb E\left[\int_{t_{n-1}}^{t_{n}} \mathcal{S}(t_n-s)\Sigma_s \mathcal{S}(t_n-s)^\ast ds\mid\mathcal{F}_{t_{n-1}}\right]$ and note that $\Delta_n\widetilde X=A(X_{t_n}-X_{t_{n-1}})=A(B_n+D_n+M_n)$. It therefore holds that
\begin{equation}\label{eq:DeltaX_decomp}
(\Delta_n\widetilde X)(\Delta_n\widetilde X)^\top
=
A(B_n{+}D_n)(B_n{+}D_n)^\ast A^\ast
+A M_n M_n^\ast A^\ast
+\text{cross terms}.
\end{equation}
The random variables $B_n$ and $D_n$ are $\mathcal F_{t_{n-1}}$-measurable
(or predictable over $[t_{n-1},t_n]$), whereas $M_n$ is a stochastic integral over $(t_{n-1},t_n]$ and hence
satisfies $\mathbb E[M_n\mid \mathcal F_{t_{n-1}}]=0$. Therefore the cross terms vanish in conditional expectation:
\[
\mathbb E\!\left[(B_n{+}D_n)M_n^\ast\mid \mathcal F_{t_{n-1}}\right]=0,
\qquad
\mathbb E\!\left[M_n(B_n{+}D_n)^\ast\mid \mathcal F_{t_{n-1}}\right]=0.
\]
It remains to identify the conditional second moment of $M_n$. By the It\^o isometry for cylindrical noise \citep[see, e.g., ]{DaPratoZabczyk2014}, it follows that
\[
\mathbb E\left[M_n M_n^\ast\mid \mathcal F_{t_{n-1}}\right]
=
\mathbb E\left[\int_{t_{n-1}}^{t_n}\mathcal S(t_n-s)\Sigma_s\mathcal S(t_n-s)^\ast\,ds\ \Big|\ \mathcal F_{t_{n-1}}\right].
\]
Premultiplying and postmultiplying by $A$ and $A^\ast$ yields \eqref{eq:final-matrix}. The coordinate-wise identity follows by taking $(i,j)$ entries and using
$\langle x,g_i\rangle = (A x)_i$ together with $\langle \mathcal S(t)u,g\rangle=\langle u,\mathcal S(t)^\ast g\rangle$.
\end{proof}

\paragraph{Proof of Proposition~\ref{prop:long_span_limit}}
\begin{proof}
For $n\ge 1$ set
\[
Z_n:=\frac{1}{\delta}(\Delta_n\widetilde X)(\Delta_n\widetilde X)^\top\in\mathbb{R}^{d\times d}.
\]
By construction of the window estimator (with disjoint windows of length $w$) it holds that
\[
\widehat{\Sigma}_j^w=\frac{1}{\delta w}\sum_{\ell=1}^w
(\Delta_{(j-1)w+\ell}\widetilde X)(\Delta_{(j-1)w+\ell}\widetilde X)^\top
=\frac{1}{w}\sum_{\ell=1}^w Z_{(j-1)w+\ell}.
\]
Hence,
\[
\begin{aligned}
\frac{1}{N_w}\sum_{j=1}^{N_w}\widehat{\Sigma}_j^w
&=\frac{1}{N_w}\sum_{j=1}^{N_w}\frac{1}{w}\sum_{\ell=1}^w Z_{(j-1)w+\ell}
=\frac{1}{N_w w}\sum_{n=1}^{N_w w} Z_n.
\end{aligned}
\]
Since the discrete-time sequence $(\widetilde X_{t_n})_{n\ge 0}$ is strictly stationary, so is
$(\Delta_n\widetilde X)_{n\ge 1}$ and $(Z_n)_{n\ge 1}$. Moreover, since $Z_n$ is nonnegative definite,
each entry $(Z_n)_{ij}$ is integrable whenever $\mathbb{E}\|\Delta_1\widetilde X\|^2<\infty$, because
\[
|(Z_n)_{ij}|=\frac{1}{\delta}\,|\Delta_n\widetilde X^{(i)}\,\Delta_n\widetilde X^{(j)}|
\le \frac{1}{2\delta}\Big((\Delta_n\widetilde X^{(i)})^2+(\Delta_n\widetilde X^{(j)})^2\Big),
\]
and stationarity gives $\mathbb{E}(\Delta_n\widetilde X^{(i)})^2=\mathbb{E}(\Delta_1\widetilde X^{(i)})^2<\infty$. Applying Birkhoff's ergodic theorem (per Assumption~\ref{ass:main_stat_assumption}) entrywise to the stationary and integrable process $(Z_n)_{n\ge 1}$ yields, for each $i,j$,
\[
\frac{1}{N_w w}\sum_{n=1}^{N_w w} (Z_n)_{ij}
\xrightarrow[N_w\to\infty]{a.s.} \mathbb{E}[(Z_1)_{ij}].
\]
Since $d$ is fixed, convergence holds for the full matrix, i.e.
\[
\frac{1}{N_w w}\sum_{n=1}^{N_w w} Z_n
\xrightarrow[N_w\to\infty]{a.s.} \mathbb{E}[Z_1]
=\frac{1}{\delta}\mathbb{E}\!\left[(\Delta_1\widetilde X)(\Delta_1\widetilde X)^\top\right].
\]
Combining with the identity
$\frac{1}{N_w}\sum_{j=1}^{N_w}\widehat{\Sigma}_j^w=\frac{1}{N_w w}\sum_{n=1}^{N_w w} Z_n$ shows that
\[
\frac{1}{N_w}\sum_{j=1}^{N_w}\widehat{\Sigma}_j^w\xrightarrow[N_w\to\infty]{a.s.}\frac{1}{\delta}\mathbb{E}\left[ (\Delta_1\widetilde{X})(\Delta_1\widetilde{X})^\top \right].
\]
Applying Proposition~\ref{prop:increment_decomp} proves \eqref{eq:long_span_limit}.
\end{proof}

\paragraph{Proof of Proposition~\ref{prop:drift_bound}}
\begin{proof}
By Proposition~\ref{prop:long_span_limit} (applied entrywise), it suffices to analyze the deterministic limit on the right-hand side of \eqref{eq:long_span_limit}. By expanding the second moment via \eqref{eq:DeltaX_decomp} and using orthogonality of the martingale term to predictable terms yields
\[
\begin{aligned}
\mathbb E\left[(\Delta_1\widetilde X)(\Delta_1\widetilde X)^\top\right]
&=
\mathbb E\!\left[A(B_1+D_1)(B_1+D_1)^\ast A^\ast\right]
\\
&\quad+
A\mathbb E\!\left[\int_{t_0}^{t_0+\delta}
\mathcal S(t_0+\delta-s)\Sigma_s \mathcal S(t_0+\delta-s)^\ast\,ds\right]A^\ast.
\end{aligned}
\]
Subtracting the (propagation) correction term and combining with the convergence of \eqref{eq:long_span_limit} gives
\[
\frac{\delta}{N_w}\sum_{j=1}^{N_w}\widehat{\Sigma}_j^w
- A(\mathcal S(\delta)-I)A^\ast
\xrightarrow[N_w\to\infty]{a.s.}
A\mathbb E\!\left[\int_{t_0}^{t_0+\delta}
\mathcal S(t_0+\delta-s)\Sigma_s \mathcal S(t_0+\delta-s)^\ast\,ds\right]A^\ast
+ R_\delta,
\]
where
\[
R_\delta
:=
\mathbb E\!\left[A(B_1+D_1)(B_1+D_1)^\ast A^\ast\right]
- A(\mathcal S(\delta)-I)A^\ast.
\]
It remains to show that $R_\delta=O(\delta^2)$ under the stated drift moment condition. First, by boundedness of $A$ and $\mathcal S$ on $[0,\delta]$, and by the Cauchy--Schwarz inequality, it holds that
\[
\begin{aligned}
\mathbb{E}\left[ \lVert A D_1\rVert^2\right]
&\leq \lVert A\rVert^2\mathbb{E}\left[\left\lVert\int_{t_0}^{t_0+\delta}\mathcal S(t_0+\delta-s)\mu_sds\right\rVert_{L^2(\mathbb{S}^1)}^2\right] \\
&\leq \lVert A\rVert^2\delta\int_{t_0}^{t_0+\delta}\mathbb{E}\left[\lVert\mathcal S(t_0+\delta-s)\mu_s\rVert_{L^2(\mathbb{S}^1)}^2\right] ds \\
&\leq \lVert A\rVert^2\delta\Big(\sup_{u\in[0,\delta]}\lVert\mathcal S(u)\rVert_{\mathrm{op}}^2\Big)\int_{t_0}^{t_0+\delta}\mathbb{E}\left[\lVert\mu_s\rVert_{L^2(\mathbb{S}^1)}^2\right]ds.
\end{aligned}
\]
The operator norm $\sup_{u\in[0,\delta]}\lVert\mathcal S(u)\rVert_{\mathrm{op}}^2$ can be bounded due to the Hille-Yosida theorem, and we therefore conclude that there exists a finite constant $C>0$ depending only on $\lVert A\rVert$, $\sup_{u\in[0,\delta]}\lVert\mathcal S(u)\rVert_{\mathrm{op}}$, and the constant $C_\mu$ which bounds $\sup_{s\geq 0}\mathbb{E}\left[ \lVert \mu_s\rVert^2_{L^2(\mathbb{S}^1)} \right]$, such that
\[
\mathbb{E}\left[ \lVert A D_1\rVert^2\right]\leq C\delta^2.
\]
Consequently,
\[
\mathbb{E}\left[A D_1D_1^\ast A^\ast\right]=O(\delta^2).
\]
Moreover, expanding $(B_1+D_1)(B_1+D_1)^\ast$ gives
\[
(B_1+D_1)(B_1+D_1)^\ast = B_1B_1^\ast + B_1D_1^\ast + D_1B_1^\ast + D_1D_1^\ast.
\]
Thus $R_\delta$ can be written as
\[
R_\delta =\Big(\mathbb{E}[AB_1B_1^\ast A^\ast]-A(\mathcal{S}(\delta)-I)A^\ast\Big)
+\mathbb{E}[A(B_1D_1^\ast + D_1B_1^\ast + D_1D_1^\ast)A^\ast].
\]
The last term is $O(\delta^2)$ because $\mathbb{E}\left[\lVert AD_1\rVert^2\right]=O(\delta^2)$ and by the Cauchy--Schwarz inequality, we find that
\[
\big\lVert \mathbb{E}[AB_1D_1^\ast A^\ast]\big\rVert
\leq \left(\mathbb{E}\lVert AB_1\rVert^2\right)^{1/2}\left(\mathbb{E}\lVert AD_1\rVert^2\right)^{1/2}
= O(\delta),
\]
and similarly for the transpose term, while $\mathbb{E}[AD_1D_1^\ast A^\ast]=O(\delta^2)$. Therefore $R_\delta=O(\delta^2)$, which proves the claim.
\end{proof}
\paragraph{Proof of Proposition~\ref{prop:plugin_estimator}}
\begin{proof}
For notational convenience, define the empirical lag-$0$ and lag-$1$ covariation matrices as
\[
\widehat{\Gamma}_N:=\frac{1}{N}\sum_{n=1}^N \widetilde{X}_{n-1}\widetilde{X}_{n-1}^\top, 
\qquad \widehat{C}_N :=\frac{1}{N}\sum_{n=1}^N \widetilde X_n\widetilde X_{n-1}^\top.
\]
Then $\widehat{\mathcal S}_{\delta,N}=\widehat{C}_N\widehat{\Gamma}_N^{-1}$. Since $(\widetilde{X}_n)_{n\in\mathbb{N}}$ is strictly stationary and ergodic with
$\mathbb E[\lVert \widetilde{X}_0\rVert^2]<\infty$, the matrix-valued sequences
$(\widetilde{X}_n\widetilde{X}_{n-1}^\top)_{n\in\mathbb{N}}$ and
$(\widetilde{X}_{n-1}\widetilde{X}_{n-1}^\top)_{n\in\mathbb{N}}$ are stationary,
ergodic, and integrable entrywise. Hence, by Birkhoff's ergodic theorem applied
entrywise it holds that 
\begin{equation}\label{eq:plugin_proof2}
\widehat{C}_N \xrightarrow[N\to\infty]{a.s.} \mathbb E[\widetilde{X}_1\widetilde{X}_0^\top],
\qquad \widehat{\Gamma}_N \xrightarrow[N\to\infty]{a.s.} \Gamma.
\end{equation}
Because $\Gamma$ is positive definite, $\widehat{\Gamma}_N$ is invertible for all
sufficiently large $N$, almost surely, and continuity of matrix inversion therefore yields that
\[
\widehat{\mathcal{S}}_{\delta,N}
=\widehat{C}_N\widehat{\Gamma}_N^{-1}
\xrightarrow[N\to\infty]{a.s.}
\mathbb E[\widetilde{X}_1\widetilde{X}_0^\top]\Gamma^{-1}
=\mathcal S_\delta.
\]
Next, each block estimator $\widehat{\Sigma}_{j,\mathrm{pl}}^{w}$ is positive
semidefinite, since it is a finite sum of outer products,
\[
\widehat{\Sigma}_{j,\mathrm{pl}}^{w}=\frac{1}{\delta w}
\sum_{\ell=1}^w\widehat{\varepsilon}_{(j-1)w+\ell}^{\mathrm{pl}}
\big(\widehat{\varepsilon}_{(j-1)w+\ell}^{\mathrm{pl}}\big)^\top .
\]
Let $N=wN_w$ and note that
\[
\frac{1}{N_w}\sum_{j=1}^{N_w}\widehat{\Sigma}_{j,\mathrm{pl}}^{w}=\frac{1}{\delta N}\sum_{n=1}^N\widehat{\varepsilon}_n^{\mathrm{pl}}\big(\widehat{\varepsilon}_n^{\mathrm{pl}}\big)^\top.
\]
Using that $\widehat\varepsilon_n^{\,\mathrm{pl}}=\widetilde X_n-\widehat{\mathcal S}_{\delta,N}\widetilde X_{n-1}$, we obtain
\begin{equation}\label{eq:plugin_proof1}
\begin{aligned}
\frac{1}{N}\sum_{n=1}^N
\widehat{\varepsilon}_n^{\mathrm{pl}}
\big(\widehat{\varepsilon}_n^{\mathrm{pl}}\big)^\top
&= \frac{1}{N}\sum_{n=1}^N \widetilde{X}_n\widetilde{X}_n^\top -
\widehat{\mathcal{S}}_{\delta,N} \Big(\frac{1}{N}\sum_{n=1}^N \widetilde{X}_{n-1}\widetilde{X}_n^\top\Big) \\
&\quad -\Big(\frac{1}{N}\sum_{n=1}^N \widetilde{X}_n\widetilde{X}_{n-1}^\top\Big)
\widehat{\mathcal{S}}_{\delta,N}^\top+\widehat{\mathcal{S}}_{\delta,N}
\Big(\frac{1}{N}\sum_{n=1}^N \widetilde{X}_{n-1}\widetilde{X}_{n-1}^\top\Big)
\widehat{\mathcal{S}}_{\delta,N}^\top.
\end{aligned}
\end{equation}
Since $\widehat{\mathcal S}_{\delta,N}\widehat\Gamma_N=\widehat C_N$,
the third and fourth terms of \eqref{eq:plugin_proof1} cancel, and therefore
\[
\frac{1}{N}\sum_{n=1}^N
\widehat{\varepsilon}_n^{\mathrm{pl}}
\big(\widehat{\varepsilon}_n^{\mathrm{pl}}\big)^\top
=\frac{1}{N}\sum_{n=1}^N \widetilde{X}_n\widetilde{X}_n^\top
-\widehat{C}_N\widehat{\Gamma}_N^{-1}\widehat{C}_N^\top.
\]
By stationarity and Birkhoff's ergodic theorem applied entrywise, we find that
\[
\frac{1}{N}\sum_{n=1}^N \widetilde{X}_n\widetilde{X}_n^\top
\xrightarrow[N\to\infty]{a.s.}\mathbb{E}[\widetilde{X}_0\widetilde{X}_0^\top]
=\Gamma.
\]
Combining this with the already established convergences of
$\widehat{C}_N$ and $\widehat{\Gamma}_N$ in \eqref{eq:plugin_proof2}, we get
\[
\frac{1}{N_w}\sum_{j=1}^{N_w}\widehat{\Sigma}_{j,\mathrm{pl}}^{w}
\xrightarrow[N_w\to\infty]{a.s.}\frac{1}{\delta}\Big(\Gamma - \mathbb{E}[\widetilde{X}_1\widetilde{X}_0^\top]\Gamma^{-1} \mathbb{E}[\widetilde{X}_0\widetilde{X}_1^\top]
\Big).
\]
Now, since $\mathcal{S}_\delta = \mathbb{E}[\widetilde{X}_1\widetilde{X}_0^\top]\Gamma^{-1}$,
this limit can be written as $\frac{1}{\delta}\big(\Gamma-\mathcal{S}_\delta \Gamma \mathcal{S}_\delta^\top\big)$. It remains to identify this with the stated oracle target. Expanding, we have
\[
\begin{aligned}
\mathbb{E}\Big[(\widetilde{X}_1-\mathcal{S}_\delta\widetilde{X}_0)(\widetilde{X}_1-\mathcal{S}_\delta\widetilde{X}_0)^\top\Big]
&= \mathbb{E}[\widetilde{X}_1\widetilde{X}_1^\top]-\mathcal{S}_\delta\mathbb{E}[\widetilde{X}_0\widetilde{X}_1^\top]-\mathbb{E}[\widetilde{X}_1\widetilde{X}_0^\top]\mathcal{S}_\delta^\top
+\mathcal{S}_\delta\Gamma\mathcal{S}_\delta^\top.
\end{aligned}
\]
By stationarity, $\mathbb{E}[\widetilde X_1\widetilde X_1^\top]=\Gamma$, while
\[
\mathbb{E}[\widetilde{X}_1\widetilde{X}_0^\top]=\mathcal{S}_\delta\Gamma,
\qquad\mathbb{E}[\widetilde{X}_0\widetilde{X}_1^\top]
=\Gamma\mathcal{S}_\delta^\top.
\]
Hence
\[
\mathbb{E}\Big[(\widetilde{X}_1-\mathcal{S}_\delta\widetilde{X}_0)
(\widetilde{X}_1-\mathcal{S}_\delta\widetilde{X}_0)^\top\Big]
=\Gamma-\mathcal{S}_\delta\Gamma\mathcal{S}_\delta^\top,
\]
which proves that
\[
\frac{1}{N_w}\sum_{j=1}^{N_w}\widehat\Sigma_{j,\mathrm{pl}}^{w}\xrightarrow[N_w\to\infty]{a.s.}
\frac{1}{\delta}\mathbb{E}\Big[(\widetilde{X}_1-\mathcal{S}_\delta \widetilde{X}_0)
(\widetilde{X}_1-\mathcal{S}_\delta \widetilde{X}_0)^\top\Big].
\]
Finally, using $\Delta_1\widetilde{X}=\widetilde{X}_1-\widetilde{X}_0$, we have
\[
\widetilde{X}_1-\mathcal{S}_\delta\widetilde{X}_0
=\Delta_1\widetilde{X}-(\mathcal{S}_\delta-I_d)\widetilde{X}_0.
\]
Therefore,
\begin{align*}
&\mathbb{E}\Big[(\widetilde{X}_1-\mathcal{S}_\delta\widetilde{X}_0)
(\widetilde{X}_1-\mathcal{S}_\delta\widetilde{X}_0)^\top\Big]
\\
&\qquad= \mathbb{E}\left[(\Delta_1\widetilde{X})(\Delta_1\widetilde{X})^\top\right]
-\mathbb{E}\Big[(\mathcal{S}_\delta-I_d)\widetilde{X}_0\widetilde{X}_0^\top(\mathcal{S}_\delta-I_d)^\top\Big].
\end{align*}
Note that $\mathbb E[\widetilde X_1\widetilde X_0^\top]=\mathcal S_\delta\Gamma$, which proves the equivalent representation \eqref{eq:plugin_equiv_representation} of the limit and completes the proof.
\end{proof}

\paragraph{Proof of Corollary~\ref{cor:analytic_semigroup}}
\begin{proof}
Analyticity of $\mathcal{S}$ implies the increment bound
$\lVert (\mathcal{S}(\delta)-I)x\rVert_{L^2(\mathbb{S}^1)}\lesssim \delta^\alpha\lVert(-\mathcal{A})^\alpha x\rVert_{L^2(\mathbb{S}^1)}$ for $x\in\mathcal{D}((-\mathcal{A})^\alpha)$.
Therefore $\mathbb{E}\left[\lVert A B_1\rVert^2 \right]=O(\delta^{2\alpha})$. The drift term satisfies $\mathbb{E}\left[ \lVert A D_1\rVert^2\right]=O(\delta^2)$ and is thus of
no larger order. Consequently,
\[
\mathbb{E}[A(B_1+D_1)(B_1+D_1)^\ast A^\ast]=O(\delta^{2\alpha}).
\]
\end{proof}

\section{Data description}\label{app:data}
The spot price data is gathered via the ENTSO-E transparency platform \href{https://transparency.entsoe.eu/}{https://transparency.entsoe.eu/}. In addition to the price data, there is also data on forecasted load (demand), wind production and solar production. These forecasts are available and used by market participants in the day-ahead auction process. There is also data on the actual, realized, generation variables. More precisely, it is possible to see the hourly electricity production coming from each fuel type. The fuel types used in production varies substantially across generation zones; for example, in Norway there is essentially no solar production, but a large share of hydropower. In Germany, the fuel mix is very diverse, and in Spain the solar production is a much more significant contributor.

In the following, we shall analyze the data in more detail, focusing on the German generation zone. Similar results can be obtained for both the Norwegian and Spanish zone under consideration, up to differences in data availability. Figure~\ref{fig:genex_share_DE_LU} depicts the relative share of various fuel sources in the German generation zone, where it is apparent that solar and wind are large contributing sources.

\begin{figure}
    \centering
    \includegraphics[width=\linewidth]{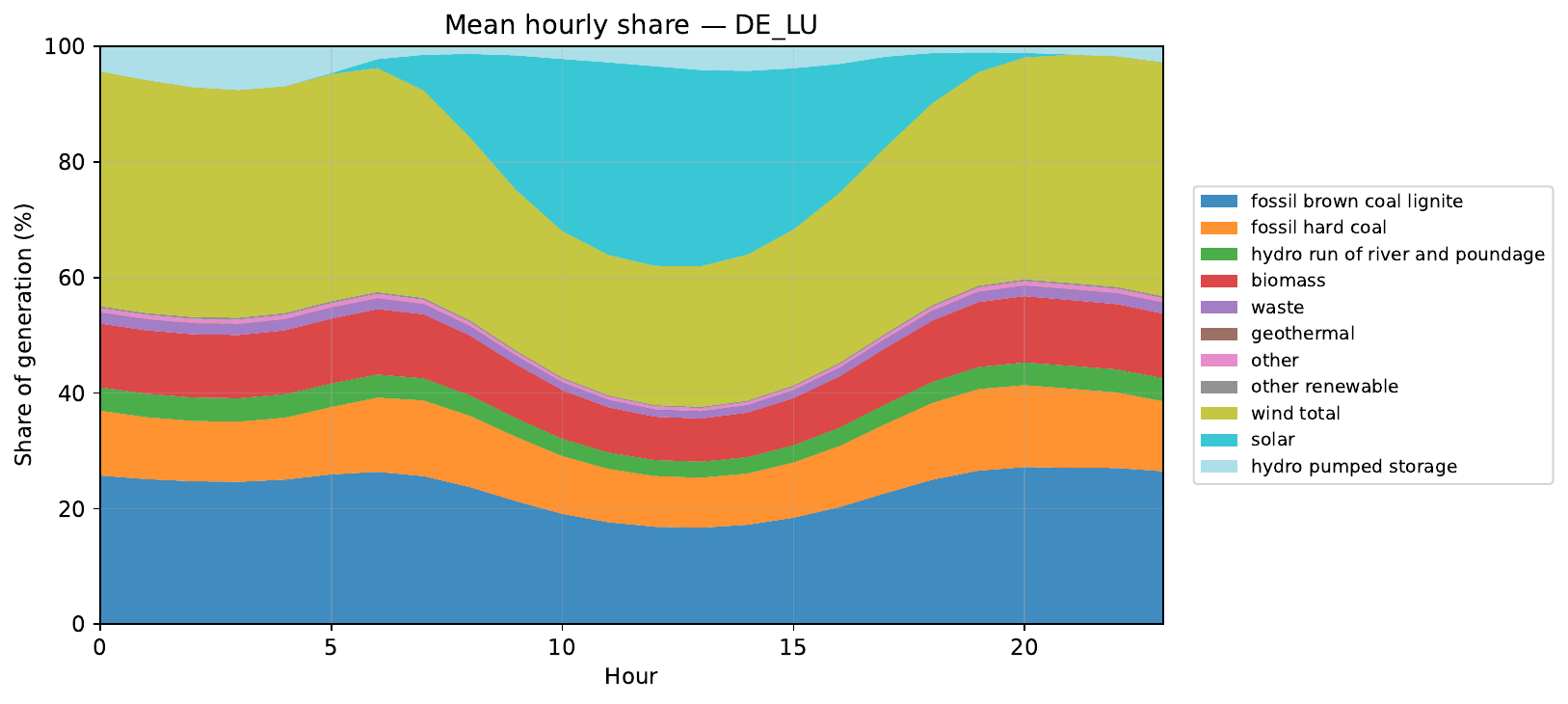}
    \caption{The relative share of various fuel sources in the German generation zone.}
    \label{fig:genex_share_DE_LU}
\end{figure} 

\begin{figure}
    \centering
    \begin{subfigure}[b]{0.32\textwidth}
        \centering
        \includegraphics[width=\linewidth]{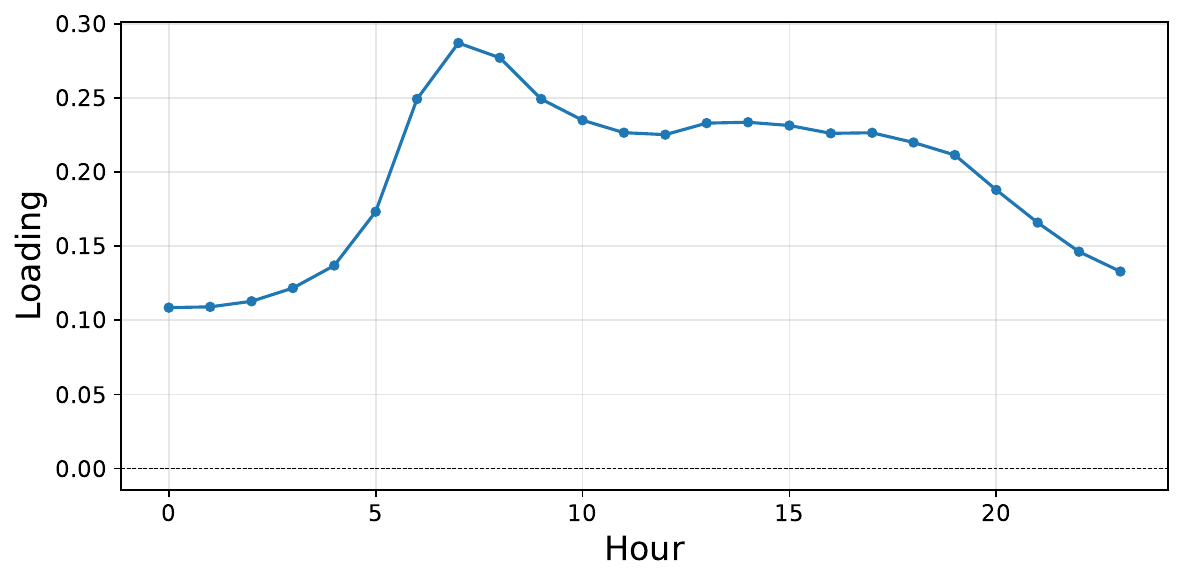}
        \caption{PC1, 93.5\%}
    \end{subfigure}
    \hfill
    \begin{subfigure}[b]{0.32\textwidth}
        \centering
        \includegraphics[width=\linewidth]{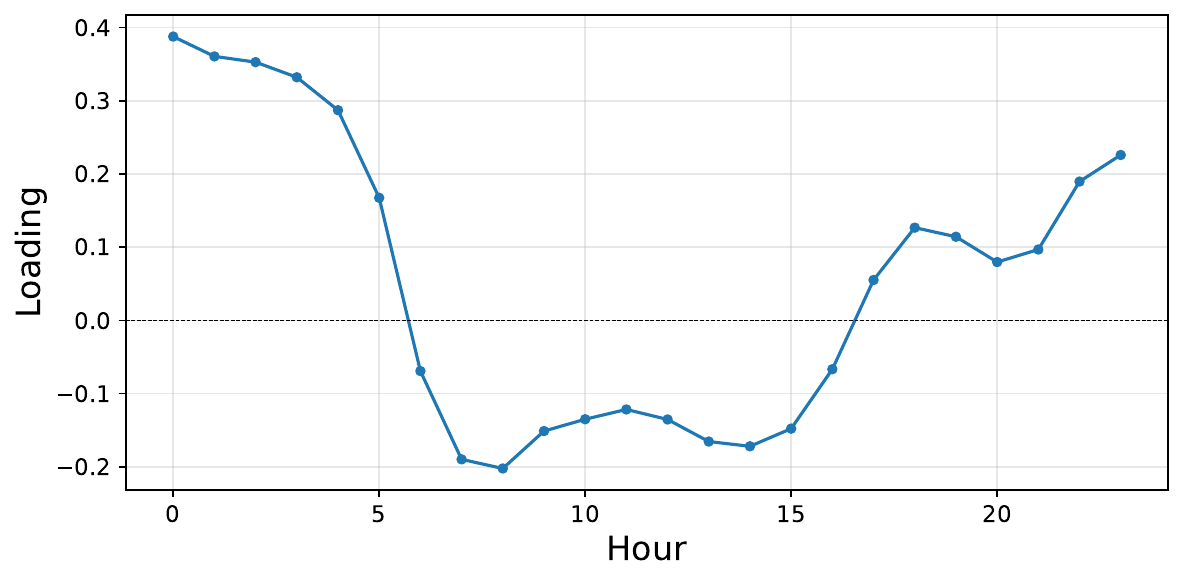}
        \caption{PC2, 4.2\%}
    \end{subfigure}
    \hfill
    \begin{subfigure}[b]{0.32\textwidth}
        \centering
        \includegraphics[width=\linewidth]{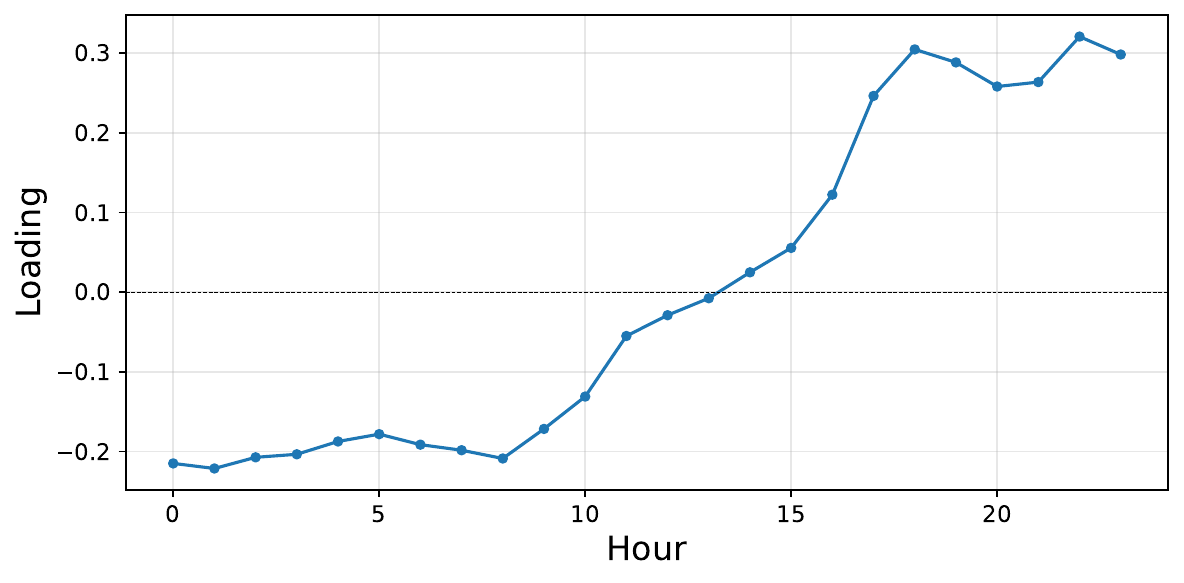}
        \caption{PC3, 1.2\%}
    \end{subfigure}
    \caption{Factor loadings for the first three principal components of the daily load forecasts. The subcaptions show the percentage of variance explained by that factor.}
    \label{fig:pca_loadings_load}
\end{figure}

\subsection{Load/demand profiles}
The forecasted load is very stable over time, indicating that the demand for electricity generally varies in the same diurnal pattern each day. On Figure~\ref{fig:pca_loadings_load} we show the factor loadings of the first three principal components. The first factor drives 93.5\% of the variation in forecasted load and thus represents the average shape of the forecast. The second factor seems to capture deviations between peak and off-peak hours, while the third factor looks like a slope factor.

\subsection{Wind and solar generation profiles}
The wind and solar generation profiles are also very stable over time. Wind patterns are more variable and unpredictable, whereas solar patterns are extremely stable as they are deterministically linked to the hours of the day where the sun is up. On Figures~\ref{fig:pca_loadings_wind} and \ref{fig:pca_loadings_solar}, we show the factor loadings of the first three principal components for wind and solar, respectively. The wind factors are nicely interpretable as level, slope, and curvature, whereas the first solar factor has a bell-curve shape, indicating the mid-day hours with most sunlight. The second and third factor seem to control deviations in evening and morning hours, respectively.

\begin{figure}
    \centering
    \begin{subfigure}[b]{0.32\textwidth}
        \centering
        \includegraphics[width=\linewidth]{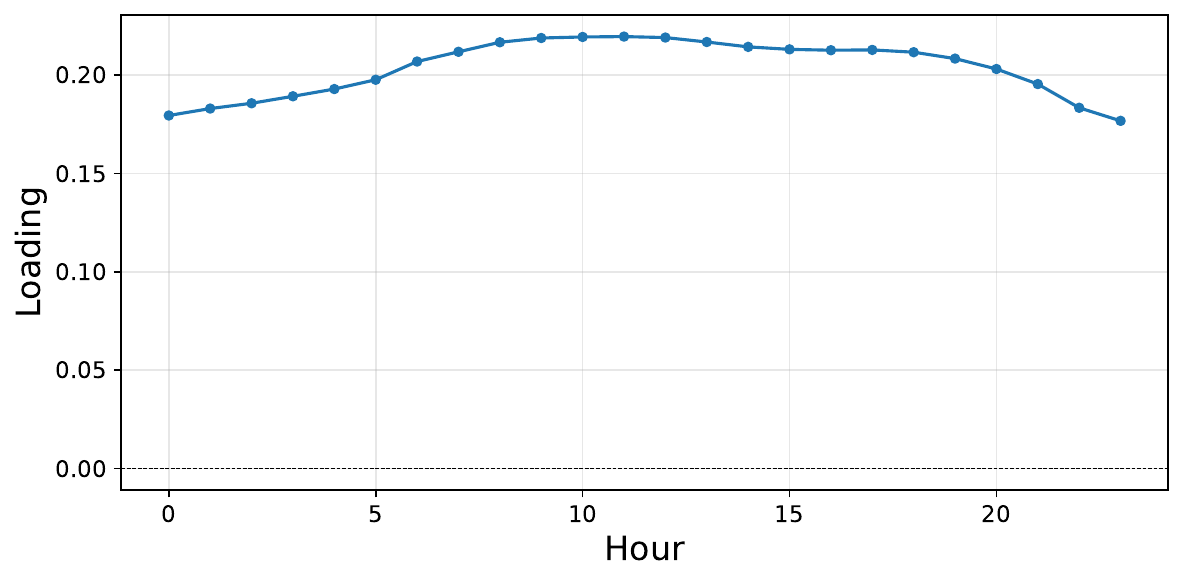}
        \caption{PC1, 85.8\%}
    \end{subfigure}
    \hfill
    \begin{subfigure}[b]{0.32\textwidth}
        \centering
        \includegraphics[width=\linewidth]{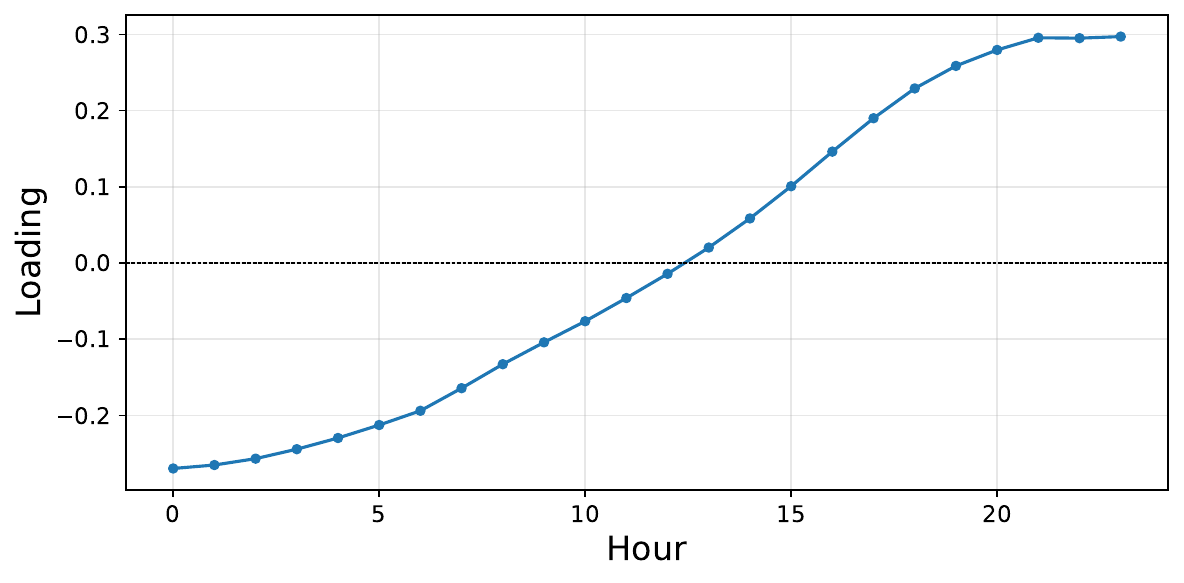}
        \caption{PC2, 9.3\%}
    \end{subfigure}
    \hfill
    \begin{subfigure}[b]{0.32\textwidth}
        \centering
        \includegraphics[width=\linewidth]{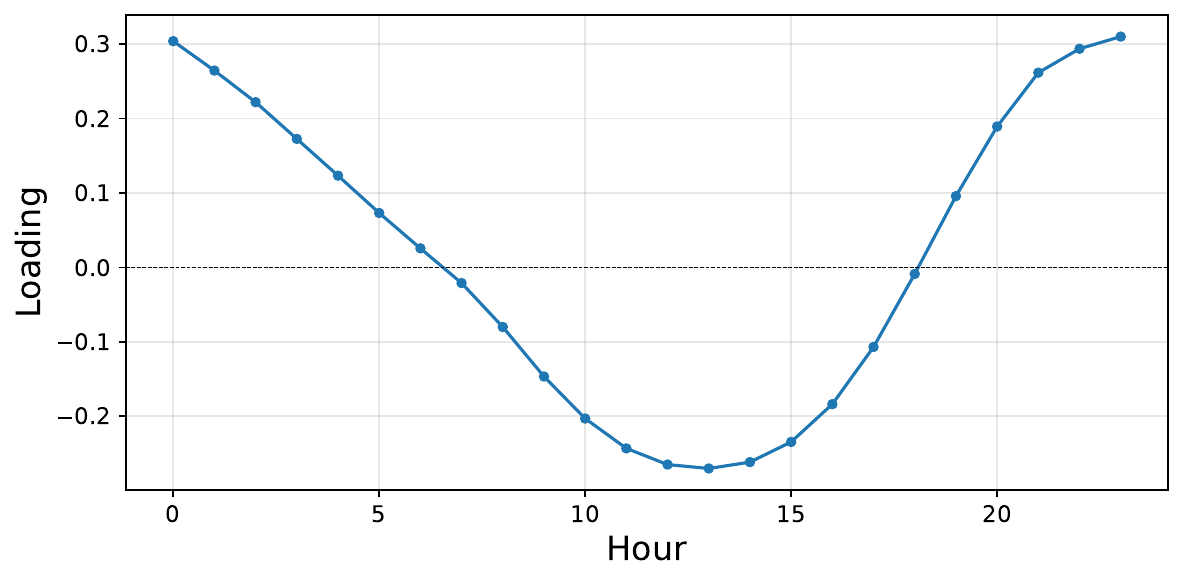}
        \caption{PC3, 3.2\%}
    \end{subfigure}
    \caption{Factor loadings for the first three principal components of the daily wind generation. The subcaptions show the percentage of variance explained by that factor.}
    \label{fig:pca_loadings_wind}
\end{figure}

\begin{figure}
    \centering
    \begin{subfigure}[b]{0.32\textwidth}
        \centering
        \includegraphics[width=\linewidth]{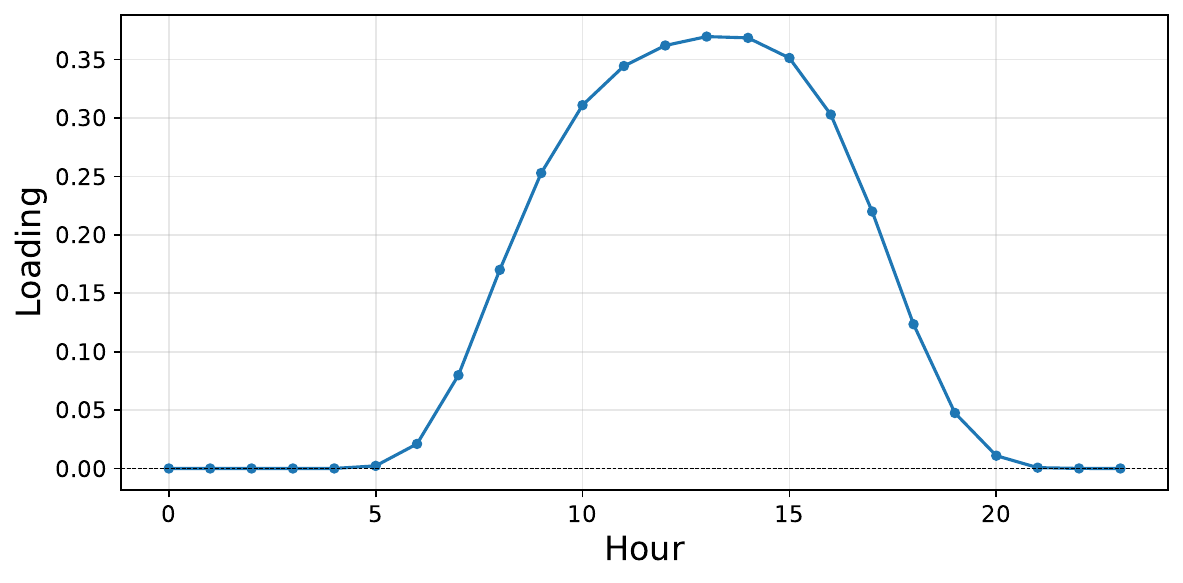}
        \caption{PC1, 94.9\%}
    \end{subfigure}
    \hfill
    \begin{subfigure}[b]{0.32\textwidth}
        \centering
        \includegraphics[width=\linewidth]{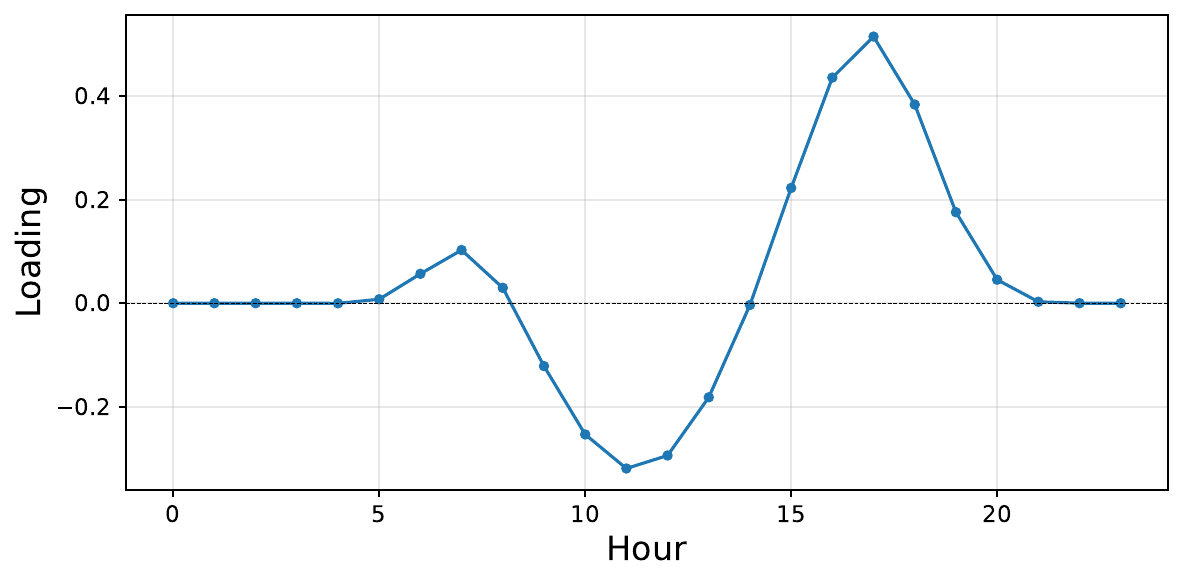}
        \caption{PC2, 3.2\%}
    \end{subfigure}
    \hfill
    \begin{subfigure}[b]{0.32\textwidth}
        \centering
        \includegraphics[width=\linewidth]{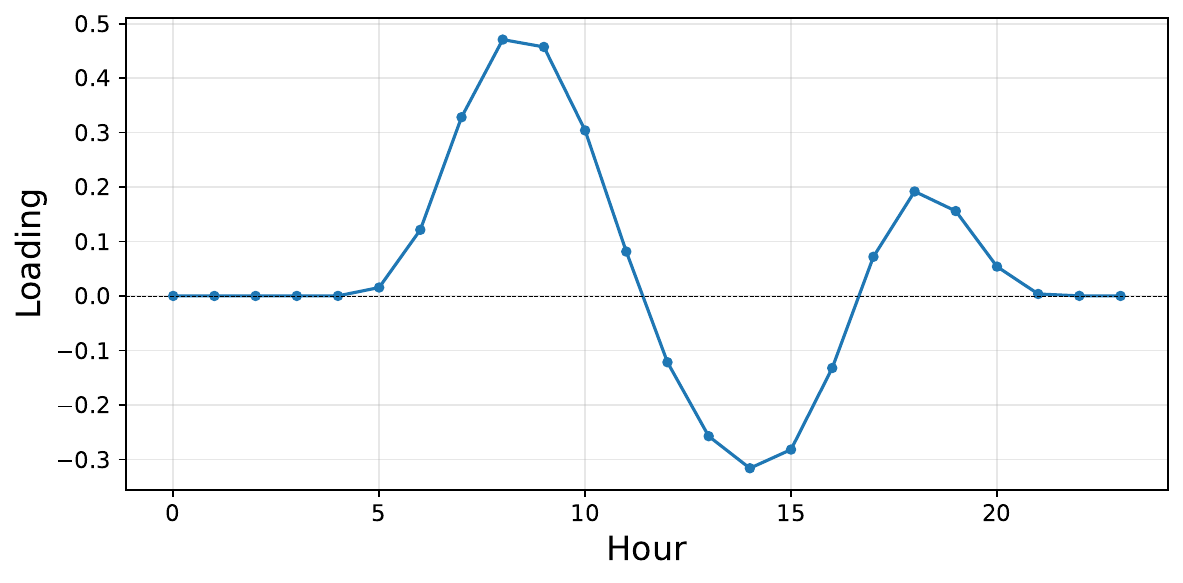}
        \caption{PC3, 1.5\%}
    \end{subfigure}
    \caption{Factor loadings for the first three principal components of the daily solar generation. The subcaptions show the percentage of variance explained by that factor.}
    \label{fig:pca_loadings_solar}
\end{figure}

\clearpage

\section{Stationarity test results}\label{app:stationarity}

\begin{table}[h]
\centering
\caption{Stationarity test results for electricity spot price levels.
Panel~A reports the tests for the German generation zone, Panel~B reports the tests for the NO2 Generation zone, and Panel~C reports the tests for the Spanish generation zone.}
\label{tab:stationarity_levels}
\smallskip
\begin{tabular}{l c c c}
\toprule
& Statistic & $p$-value & Decision \\
\midrule
\multicolumn{4}{l}{\textit{Panel A: Multivariate tests on daily price levels in Germany}} \\[3pt]
\quad Multivariate KPSS & 21.2358 & $<0.0001$ & Reject \\
\quad ADF (component-wise) & & & 19/24 reject \\
\qquad Median & $-3.305$ & 0.0147 & \\
\qquad Range & $[-3.650,\, -2.677]$ & $[0.005,\, 0.078]$ & \\
\quad KPSS (component-wise) & & & 24/24 reject \\
\qquad Median & 1.740 & $<0.01$ & \\
\qquad Range & $[1.061,\, 2.012]$ & & \\
[6pt]
\multicolumn{4}{l}{\textit{Panel B: Multivariate tests on daily price levels in Norway (NO2)}} \\[3pt]
\quad Multivariate KPSS & 15.4141 & $<0.0001$ & Reject \\
\quad ADF (component-wise) & & & 23/24 reject \\
\qquad Median & $-3.231$ & 0.0184 & \\
\qquad Range & $[-3.640,\, -2.804]$ & $[0.005,\, 0.058]$ & \\
\quad KPSS (component-wise) & & & 24/24 reject \\
\qquad Median & 2.747 & $<0.01$ & \\
\qquad Range & $[2.562,\, 2.872]$ & & \\
[6pt]
\multicolumn{4}{l}{\textit{Panel C: Multivariate tests on daily price levels in Spain}} \\[3pt]
\quad Multivariate KPSS& 29.0122 & $<0.0001$ & Reject \\
\quad ADF (component-wise) & & & 0/24 reject \\
\qquad Median & $-2.307$ & 0.1698 & \\
\qquad Range & $[-2.624,\, -2.125]$ & $[0.088,\, 0.234]$ & \\
\quad KPSS (component-wise) & & & 24/24 reject \\
\qquad Median & 1.390 & $<0.01$ & \\
\qquad Range & $[0.958,\, 1.929]$ & & \\
\bottomrule
\end{tabular}
\end{table}

\begin{table}
\centering
\caption{Stationarity test results for electricity spot price differences/increments.
Panel~A reports the tests for the German generation zone, Panel~B reports the tests for the NO2 Generation zone, and Panel~C reports the tests for the Spanish generation zone.}
\label{tab:stationarity_differences}
\smallskip
\begin{tabular}{l c c c}
\toprule
& Statistic & $p$-value & Decision \\
\midrule
\multicolumn{4}{l}{\textit{Panel A: Multivariate tests on daily price levels in Germany}} \\[3pt]
\quad Multivariate KPSS & 0.1544 & 1.0000 & Fail to reject \\
\quad ADF (component-wise) & & & 24/24 reject \\
\qquad Median & $-13.611$ & $<0.0001$ & \\
\qquad Range & $[-14.526,\, -13.044]$ & & \\
\quad KPSS (component-wise) & & & 0/24 reject \\
\qquad Median & 0.034 & $>0.10$ & \\
\qquad Range & $[0.025,\, 0.066]$ & & \\
[6pt]
\multicolumn{4}{l}{\textit{Panel B: Multivariate tests on daily price levels in Norway (NO2)}} \\[3pt]
\quad Multivariate KPSS & 0.0951 & 1.0000 & Fail to reject \\
\quad ADF (component-wise) & & & 24/24 reject \\
\qquad Median & $-9.480$ & $<0.0001$ & \\
\qquad Range & $[-11.441,\, -8.582]$ & & \\
\quad KPSS (component-wise) & & & 0/24 reject \\
\qquad Median & 0.046 & $>0.10$ & \\
\qquad Range & $[0.025,\, 0.070]$ & & \\
[6pt]
\multicolumn{4}{l}{\textit{Panel C: Multivariate tests on daily price levels in Spain}} \\[3pt]
\quad Multivariate KPSS & 0.1333 & 1.0000 & Fail to reject \\
\quad ADF (component-wise) & & & 24/24 reject \\
\qquad Median & $-14.887$ & $<0.0001$ & \\
\qquad Range & $[-15.738,\, -13.988]$ & & \\
\quad KPSS (component-wise) & & & 0/24 reject \\
\qquad Median & 0.058 & $>0.10$ & \\
\qquad Range & $[0.039,\, 0.219]$ & & \\
\bottomrule
\end{tabular}
\end{table}

\clearpage

\section{Plots and tables for Norway and Spain}\label{app:plots}

\begin{table}[b]
\footnotesize
\centering
\caption{Regressions of $\log(S_1(t))$ on same-week fundamentals in the Norwegian zone (NO-2), using the propagation-adjusted RCV. Entries report coefficients with HAC $t$-statistics in parentheses.}
\label{tab:contemp_logIV_norway}
\setlength{\tabcolsep}{5pt}
\begin{tabular}{@{}lcc@{}}
\toprule
 & RES + Price & All sources \\
\midrule
\grp{\textbf{Price level}} \\
$\overline{P}$     
    & \ct{0.0440$^{***}$}{11.27}
    & \ct{0.0270$^{***}$}{10.57} \\
$\overline{P}^2$   
    & \ct{$-$6.6e-5$^{***}$}{$-$6.18}
    & \ct{$-$3.3e-5$^{***}$}{$-$6.36} \\

\grp{\textbf{Forecast PC1 scores}} \\
Load PC1
    & \na
    & \ct{1.3e-4$^{***}$}{6.33} \\
Wind PC1
    & \ct{6.4e-4$^{***}$}{7.33}
    & \ct{5.1e-3}{1.08} \\

\grp{\textbf{Wind actuals and forecast error}} \\
Wind actual
    & \na
    & \ct{$-$2.5e-2}{$-$1.07} \\
Wind error
    & \ct{$-$3.1e-3$^{***}$}{$-$4.20}
    & \ct{$-$2.5e-2}{$-$1.06} \\

\grp{\textbf{Hydro generation}} \\
Hydro pump.\ stor.
    & \na
    & \ct{1.7e-3$^{**}$}{2.25} \\
Hydro run-of-river
    & \na
    & \ct{3.9e-3$^{***}$}{6.66} \\
Hydro reservoir
    & \na
    & \ct{$-$5.4e-4$^{***}$}{$-$5.50} \\

\grp{\textbf{Thermal and other}} \\
Fossil gas
    & \na
    & \ct{$-$7.3e-3}{$-$0.51} \\
Waste
    & \na
    & \ct{$-$0.100$^{***}$}{$-$2.71} \\
\midrule
$R^2$ & 0.592 & 0.692 \\
$N$   & 2{,}100 & 1{,}835 \\
\bottomrule
\end{tabular}
\end{table}

\begin{table}
\footnotesize
\centering
\setlength{\tabcolsep}{3pt}
\renewcommand{\arraystretch}{0.9}
\caption{Regressions of $\log(S_1(t))$ on same-week fundamentals in Spain, using the propagation-adjusted RCV. Entries report coefficients with HAC $t$-statistics in parentheses.}
\label{tab:contemp_logIV_ES}
\begin{tabular}{@{}lcc@{}}
\toprule
 & RES + Price & All sources \\
\midrule
\grp{\textbf{Price level}} \\
$\overline{P}$     & \ct{0.0228$^{***}$}{9.21} & \ct{0.0301$^{***}$}{13.50} \\
$\overline{P}^2$   & \ct{$-$2.9e-5$^{***}$}{$-$3.26} & \ct{$-$3.6e-5$^{***}$}{$-$4.78} \\

\grp{\textbf{RES forecasts, PC1 scores}} \\
Load PC1           & \na & \ct{$-$3.0e-5$^{***}$}{$-$4.58} \\
Wind PC1           & \ct{6.1e-5$^{***}$}{15.06} & \ct{4.0e-4$^{*}$}{1.94} \\
Solar PC1          & \ct{5.0e-5$^{***}$}{13.07} & \ct{$-$1.0e-4$^{**}$}{$-$2.44} \\

\grp{\textbf{RES actuals and forecast errors}} \\
Wind actual        & \na & \ct{$-$1.8e-3$^{*}$}{$-$1.78} \\
Wind error         & \ct{8.8e-5}{0.63} & \ct{$-$1.9e-3$^{*}$}{$-$1.90} \\
Solar actual       & \na & \ct{7.9e-4$^{***}$}{2.87} \\
Solar error        & \ct{$-$6.8e-5}{$-$0.31} & \ct{5.5e-4$^{*}$}{1.77} \\

\grp{\textbf{Thermal and nuclear generation}} \\
Lignite            & \na & \ct{$\sim$0$^{***}$}{2.59} \\
Hard coal          & \na & \ct{1.9e-4$^{*}$}{1.65} \\
Fossil gas         & \na & \ct{$-$2.1e-5}{$-$0.67} \\
Fossil oil         & \na & \ct{$-$3.3e-3$^{***}$}{$-$3.88} \\
Nuclear            & \na & \ct{2.3e-5}{0.51} \\

\grp{\textbf{Hydro and other}} \\
Hydro pump.\ stor. & \na & \ct{1.1e-3$^{***}$}{7.44} \\
Hydro run-of-river & \na & \ct{$-$6.5e-4$^{***}$}{$-$3.00} \\
Hydro reservoir    & \na & \ct{3.8e-4$^{***}$}{5.72} \\
Biomass            & \na & \ct{$-$8.0e-4}{$-$1.63} \\
Waste              & \na & \ct{$-$2.0e-3$^{**}$}{$-$2.30} \\
Other              & \na & \ct{1.7e-3}{0.98} \\
Other renewable    & \na & \ct{$-$0.0167$^{***}$}{$-$3.61} \\
\midrule
$R^2$              & 0.637 & 0.869 \\
$N$                & 2{,}118 & 2{,}118 \\
\bottomrule
\end{tabular}
\end{table}

\begin{figure}[htbp]
    \centering
    \begin{subfigure}[b]{0.45\textwidth}
        \centering
        \includegraphics[width=\textwidth]{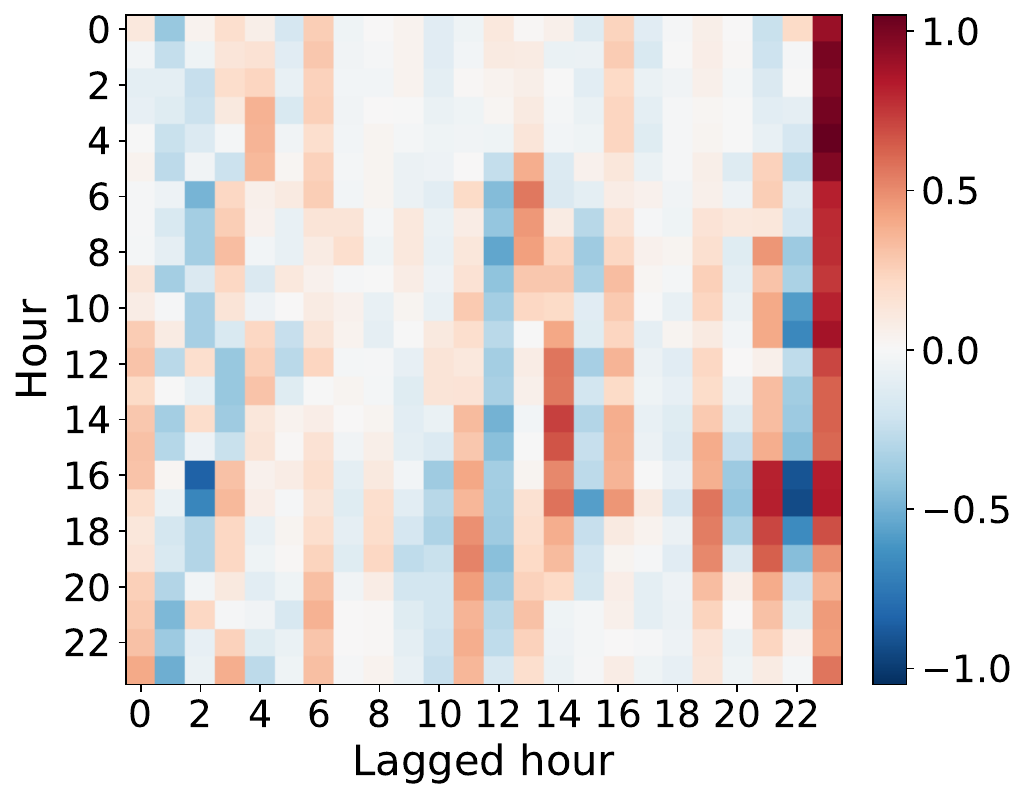}
        \subcaption{Semigroup matrix $\widehat{\mathcal{S}}_\delta$}
    \end{subfigure}
    \hfill
    \begin{subfigure}[b]{0.45\textwidth}
        \centering
        \includegraphics[width=\textwidth]{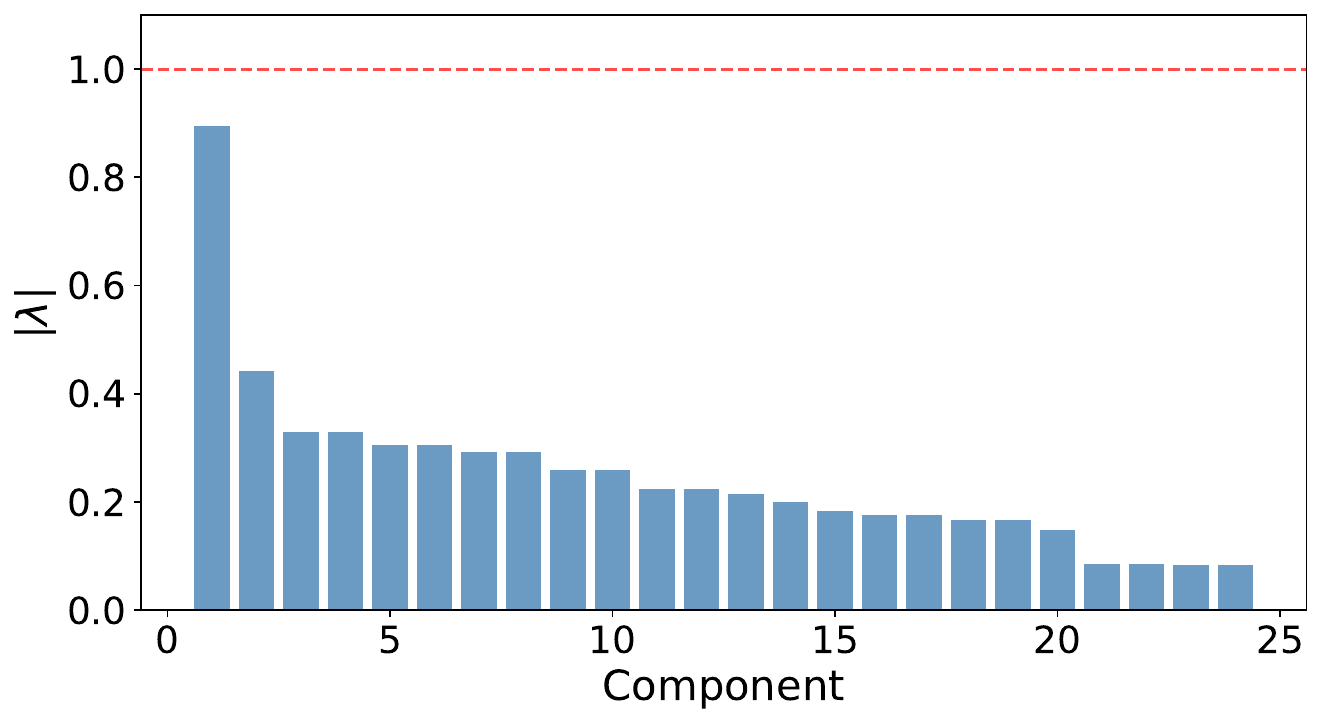}
        \subcaption{Eigenvalue spectrum}
    \end{subfigure}
    \caption{Estimated semigroup matrix $\widehat{\mathcal{S}}_\delta$ in the Norwegian (NO2) generation zone (left) and the corresponding eigenvalue spectrum (right).}
    \label{fig:estimated_semigroup_NO_2}
\end{figure}

\begin{figure}[htbp]
    \centering
    \begin{subfigure}[b]{0.45\textwidth}
        \centering
        \includegraphics[width=\textwidth]{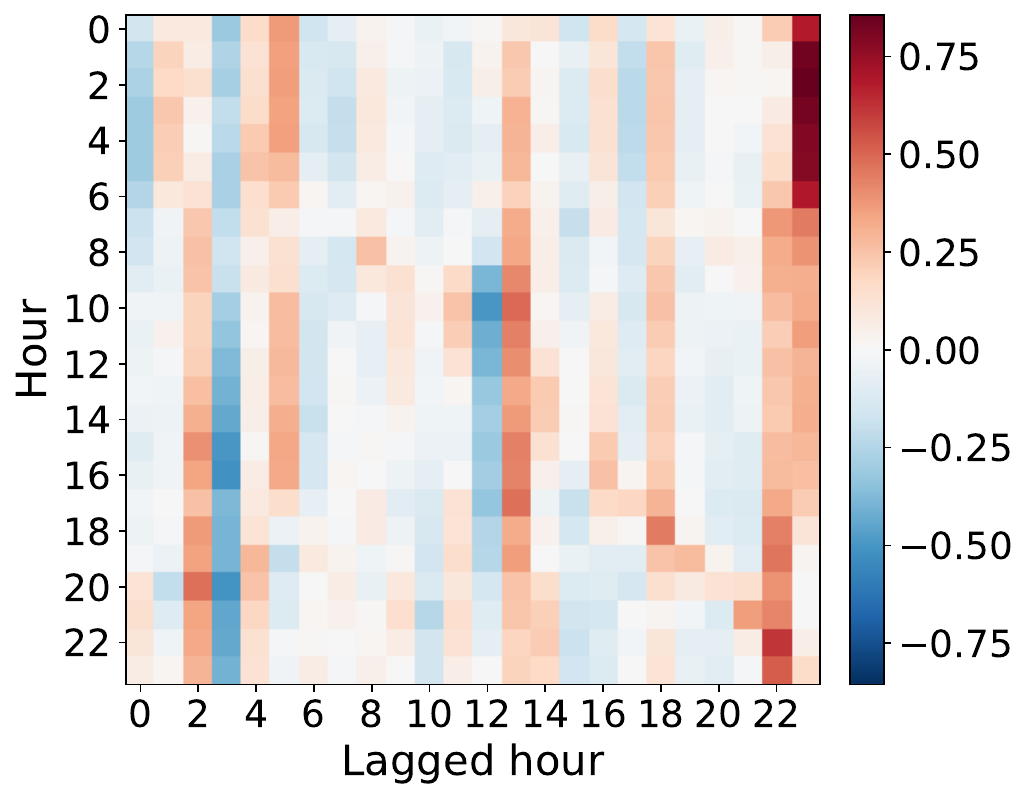}
        \subcaption{Semigroup matrix $\widehat{\mathcal{S}}_\delta$}
    \end{subfigure}
    \hfill
    \begin{subfigure}[b]{0.45\textwidth}
        \centering
        \includegraphics[width=\textwidth]{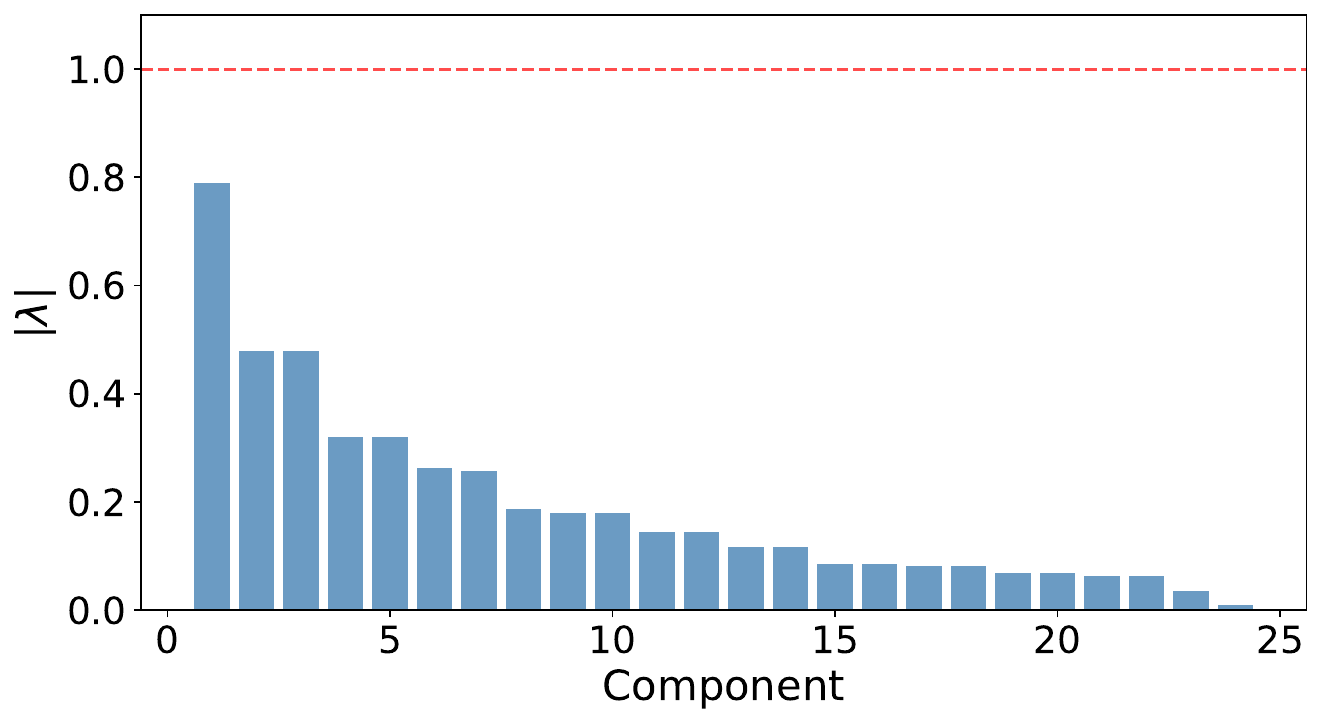}
        \subcaption{Eigenvalue spectrum}
    \end{subfigure}
    \caption{Estimated semigroup matrix $\widehat{\mathcal{S}}_\delta$ in the Spanish generation zone (left) and the corresponding eigenvalue spectrum (right).}
    \label{fig:estimated_semigroup_ES}
\end{figure}

\begin{figure}
    \centering
    \includegraphics[width=\linewidth]{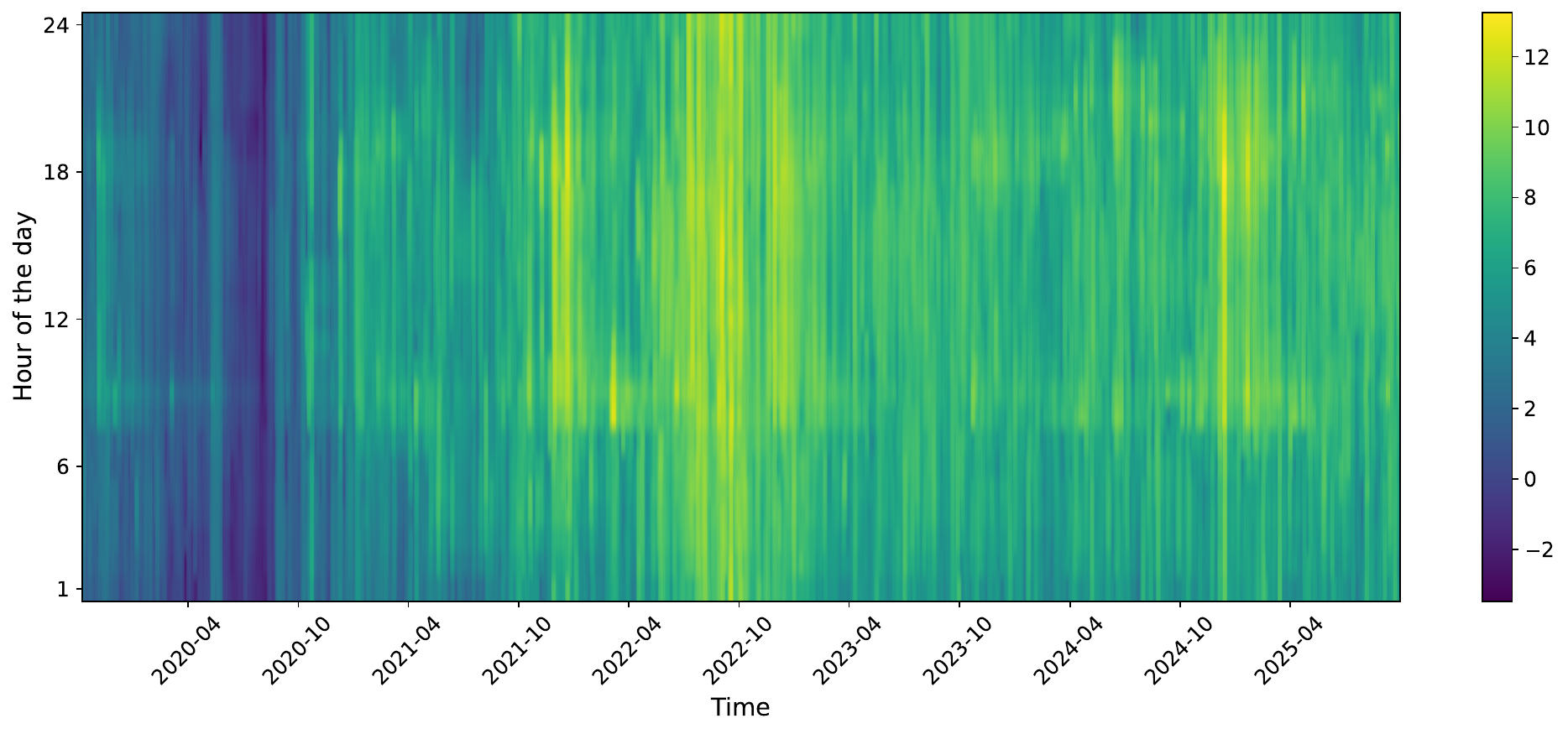}
    \caption{Heatmap of logarithm of estimated weekly RCV, $\mathrm{diag}(\widehat{\Sigma}^7_t)$, for the Norwegian (NO2) generation zone.}\label{fig:log_vol_heatmap_NO_2}
\end{figure}

\begin{figure}
    \centering
    \includegraphics[width=\linewidth]{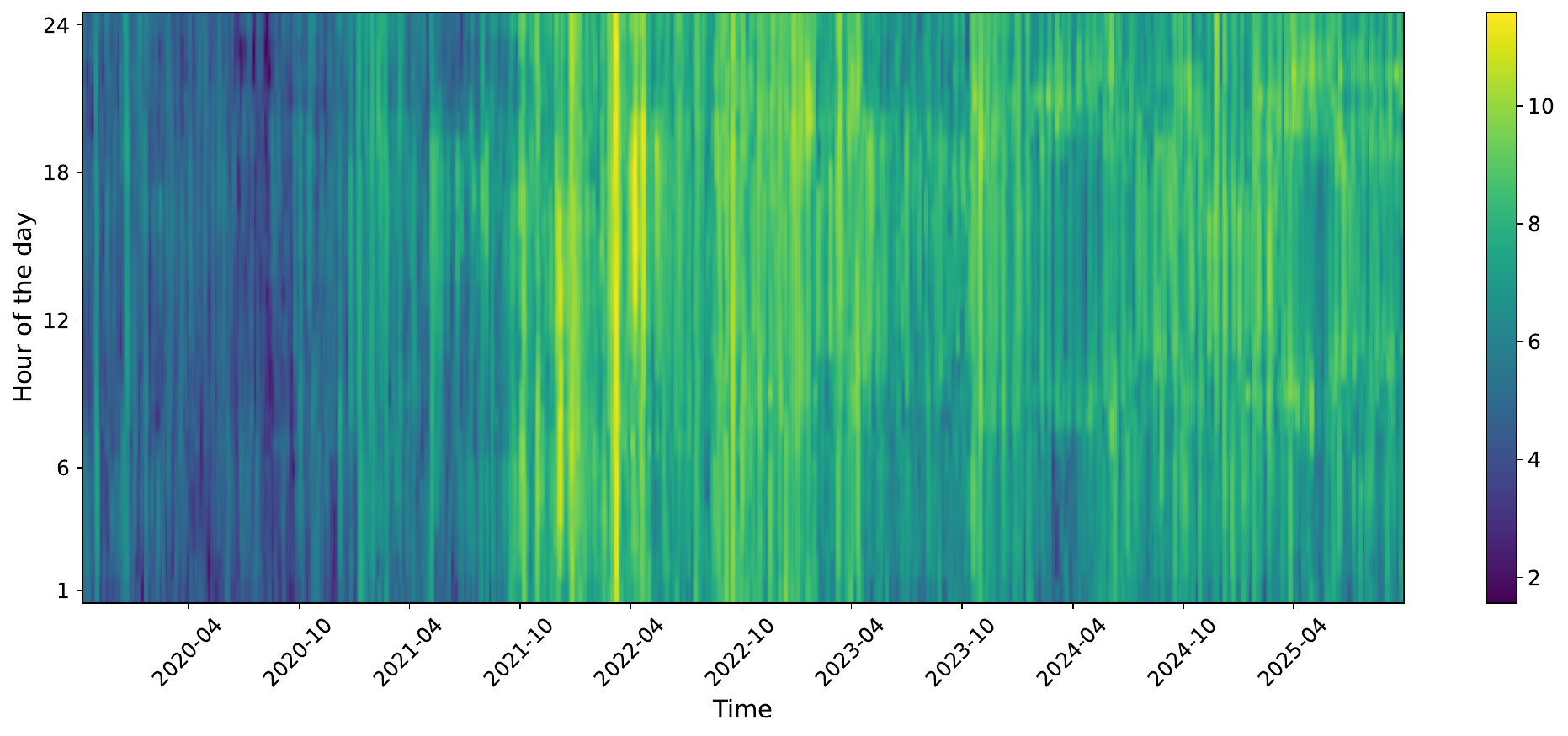}
    \caption{Heatmap of logarithm of estimated weekly RCV, $\mathrm{diag}(\widehat{\Sigma}^7_t)$, for the Spanish generation zone.}\label{fig:log_vol_heatmap_ES}
\end{figure}

\begin{figure}
    \centering
    \includegraphics[width=0.9\linewidth]{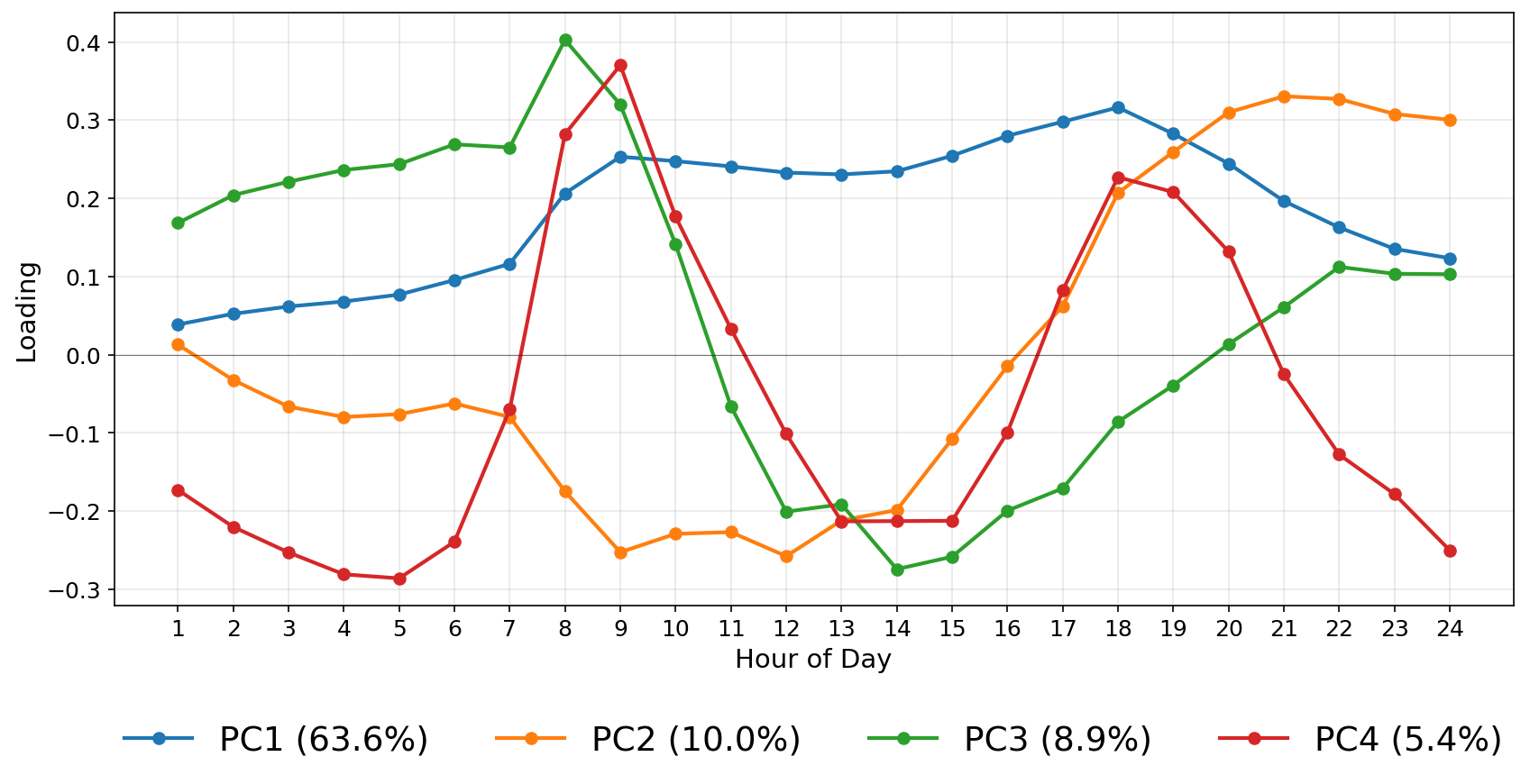}
    \caption{First four factor loadings in the Norwegian (NO2) generation zone. Note that due to rotation invariance, the loadings are only identified up to a sign. The legends also denote the percentage of variance explained by each component.}
    \label{fig:factor_loadings_NO_2}
\end{figure}
\begin{figure}
    \centering
    \includegraphics[width=0.9\linewidth]{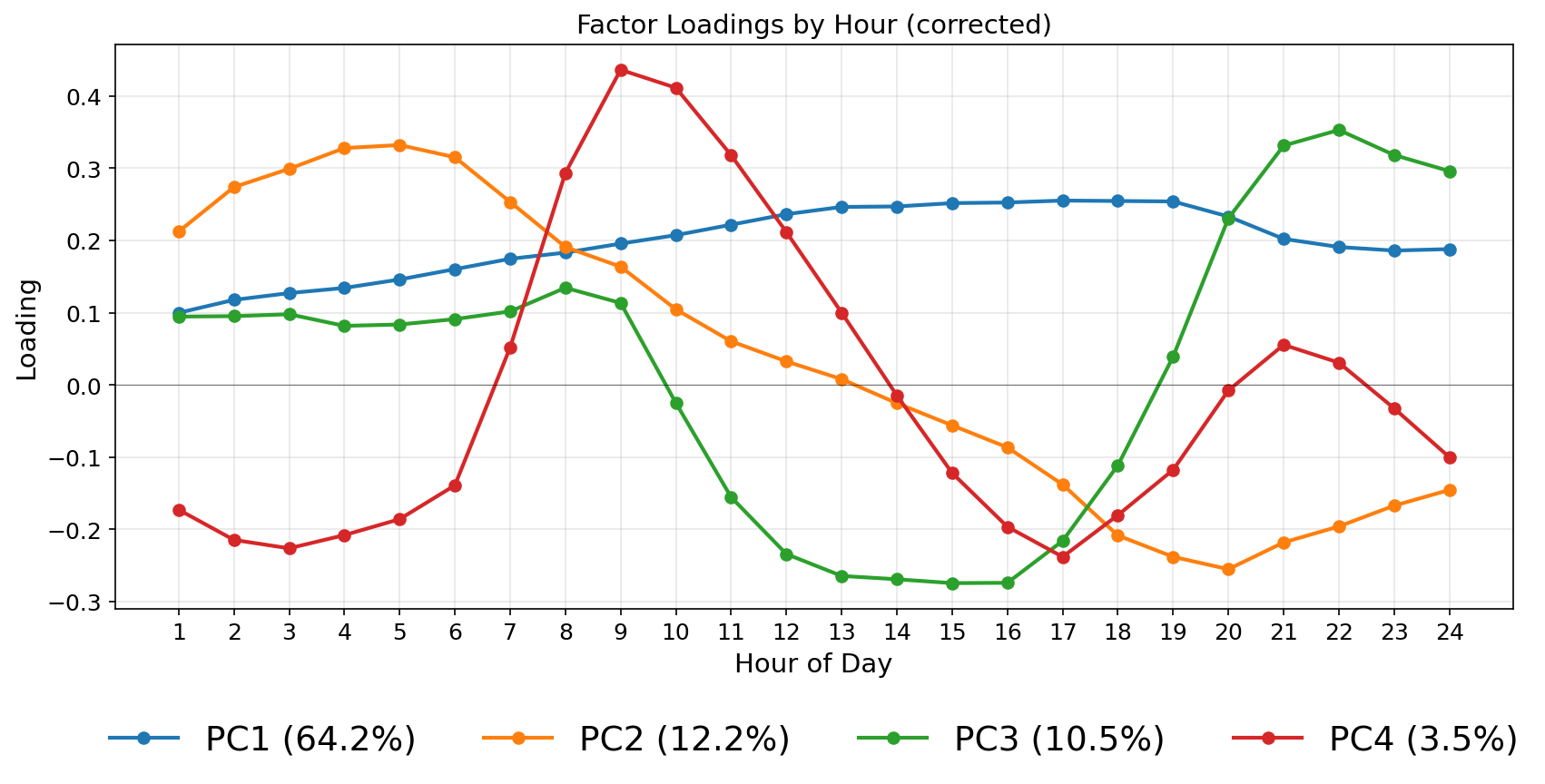}
    \caption{First four factor loadings in the Spanish generation zone. Note that due to rotation invariance, the loadings are only identified up to a sign. The legends also denote the percentage of variance explained by each component.}
    \label{fig:factor_loadings_ES}
\end{figure}

\begin{figure}
    \centering
    \begin{subfigure}[b]{0.4\textwidth}
        \includegraphics[width=\linewidth]{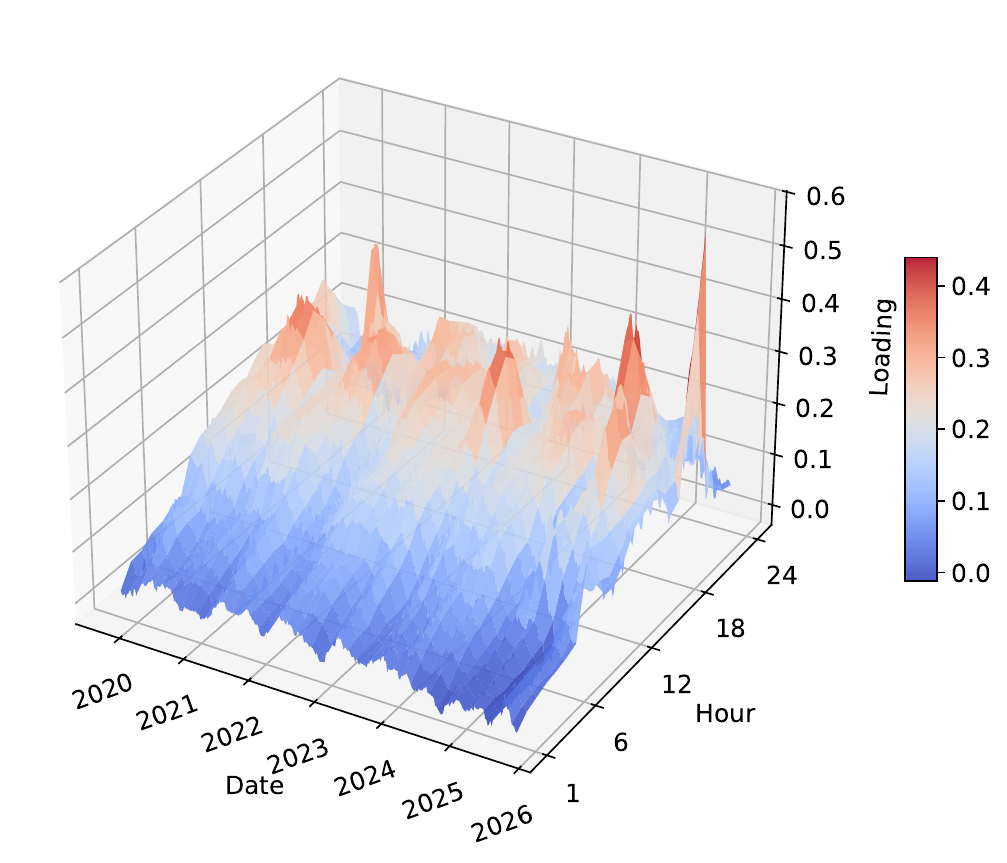}
        \caption{PC1}
    \end{subfigure}
    \hspace{0.04\linewidth}
    \begin{subfigure}[b]{0.4\textwidth}
        \includegraphics[width=\linewidth]{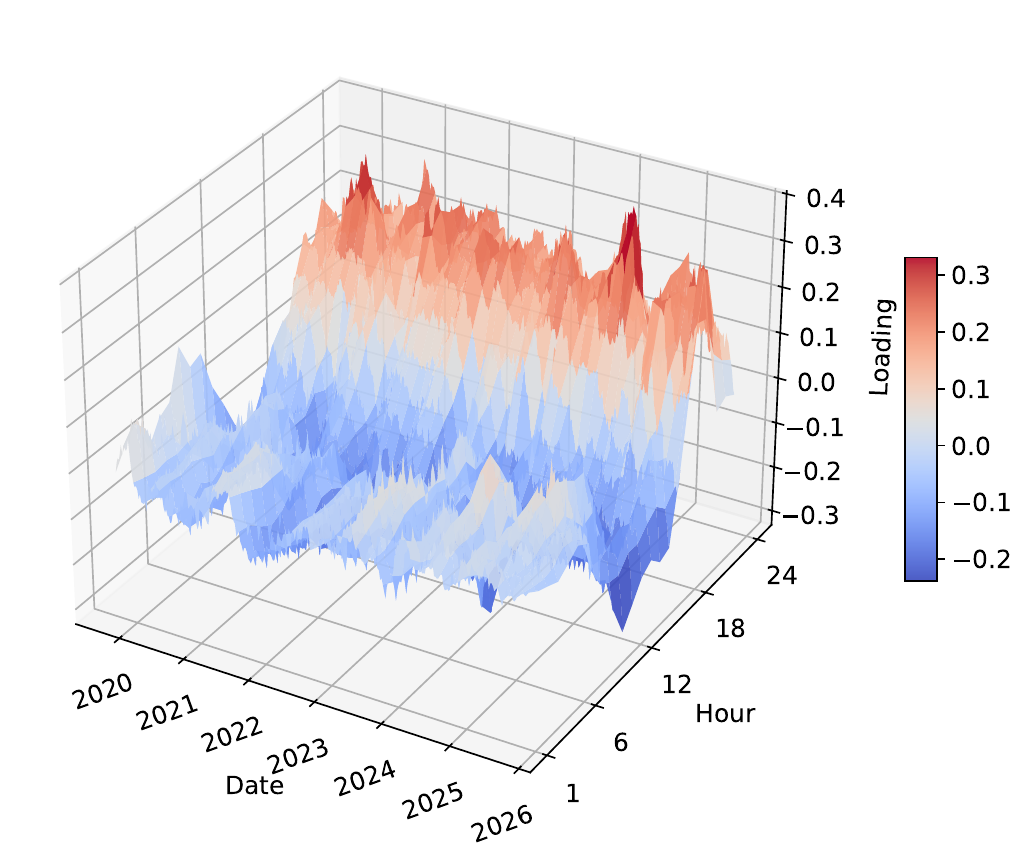}
        \caption{PC2}
    \end{subfigure}

    \vspace{0.3cm}

    \begin{subfigure}[b]{0.4\textwidth}
        \includegraphics[width=\linewidth]{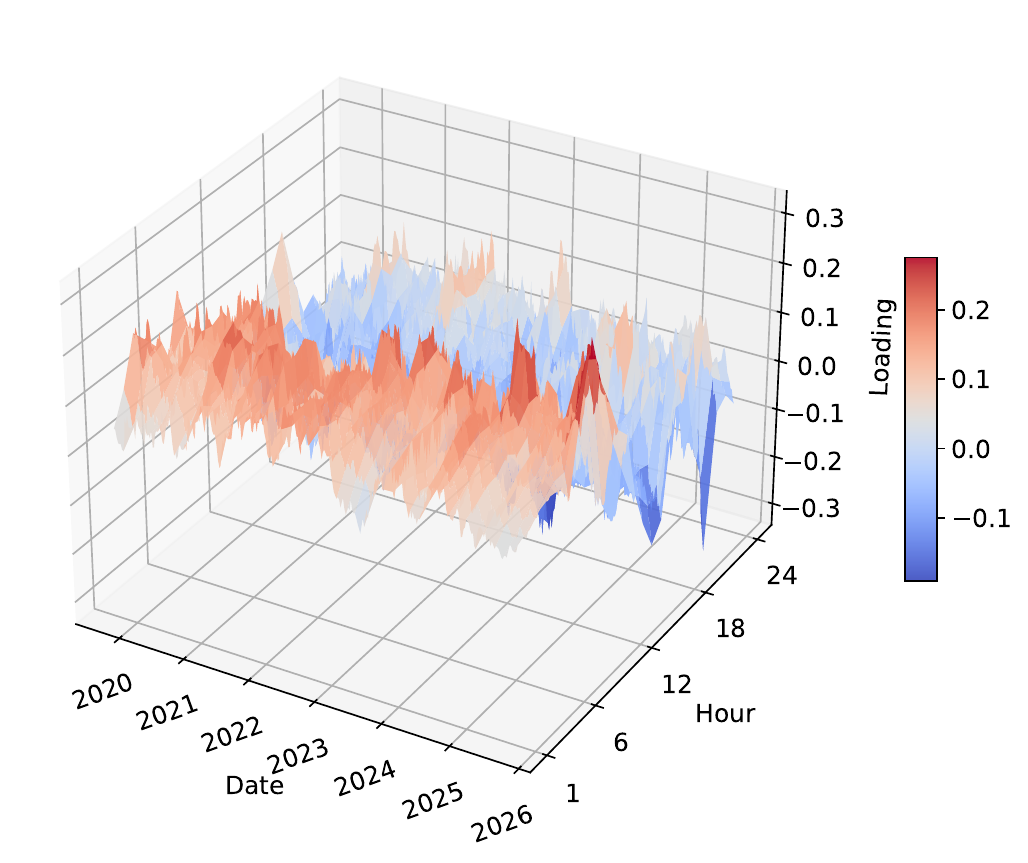}
        \caption{PC3}
    \end{subfigure}
    \hspace{0.04\linewidth}
    \begin{subfigure}[b]{0.4\textwidth}
        \includegraphics[width=\linewidth]{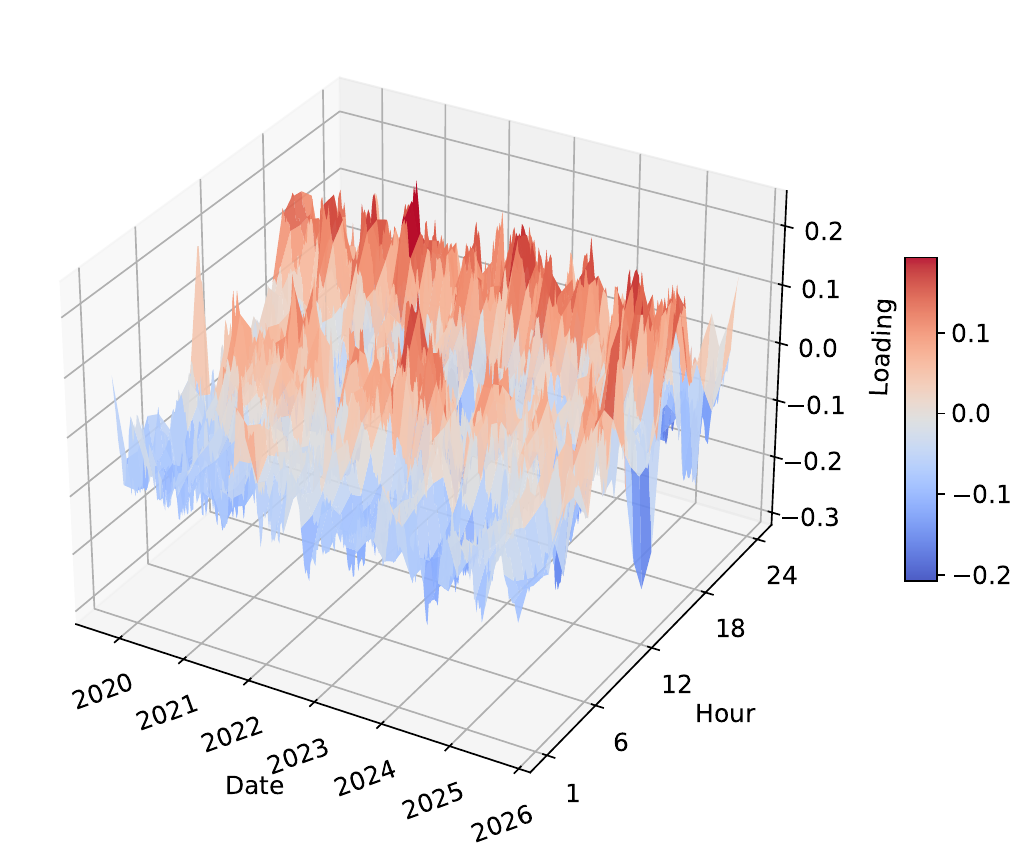}
        \caption{PC4}
    \end{subfigure}
    \caption{Factor loadings over time for the first four principal components in the Norwegian (NO2) generation zone.}\label{fig:loading_stability_NO_2}
\end{figure}

\begin{figure}
    \centering
    \begin{subfigure}[b]{0.4\textwidth}
        \centering
        \includegraphics[width=\textwidth]{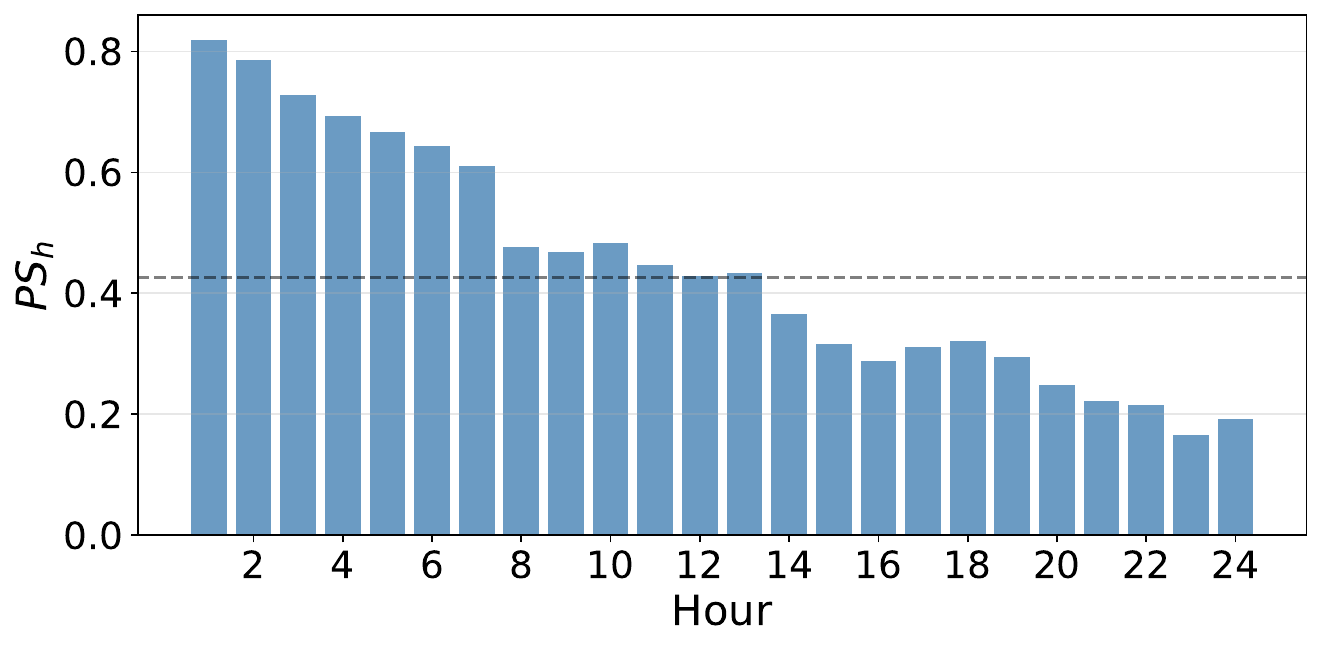}
        \subcaption{Propagation share per hour.}
    \end{subfigure}
    \hspace{0.04\linewidth}
    \begin{subfigure}[b]{0.4\textwidth}
        \centering
        \includegraphics[width=\textwidth]{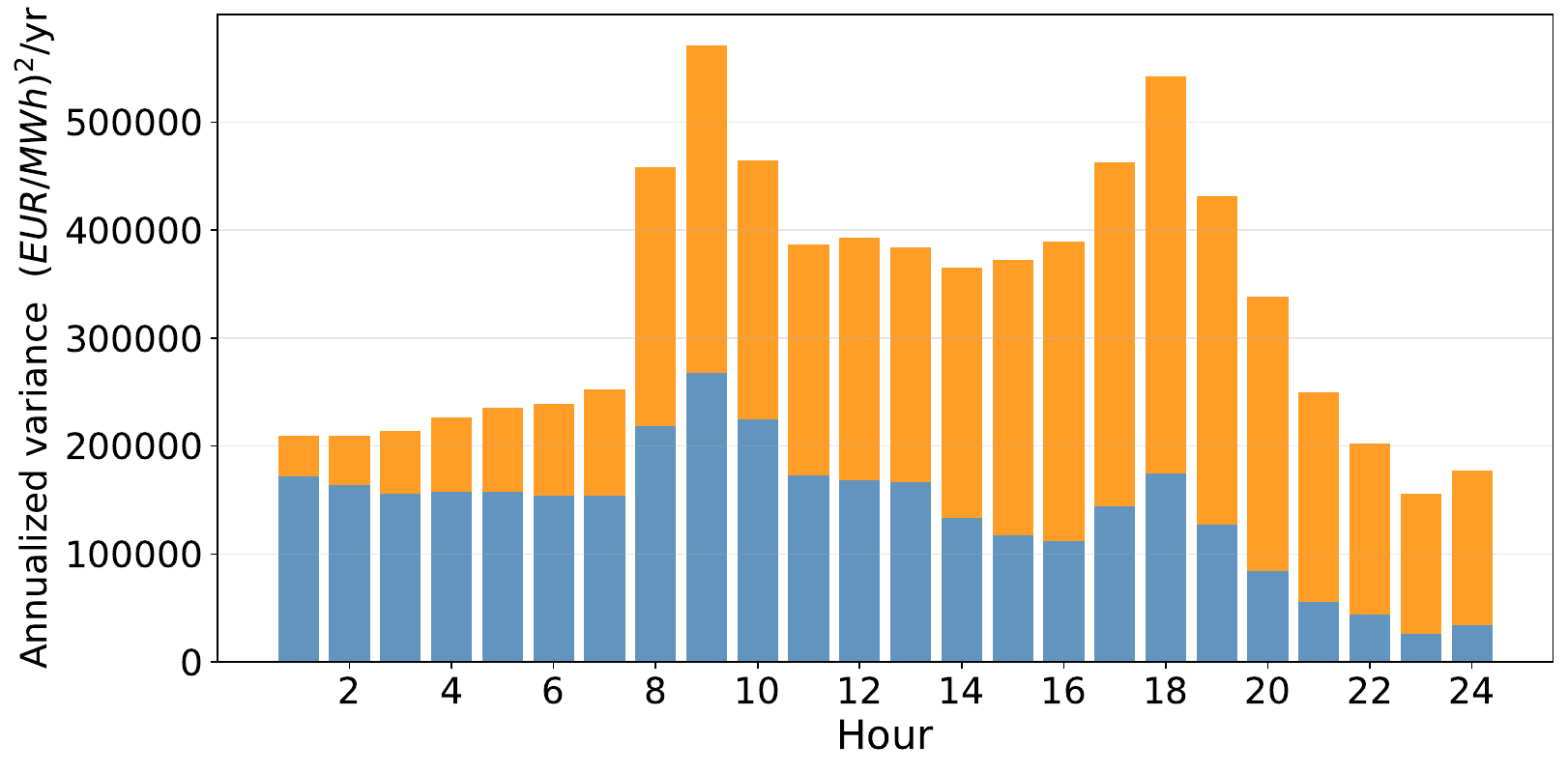}
        \subcaption{Propagation versus innovation.}
    \end{subfigure}
    \caption{Propagation share per hour in the Norwegian (NO2) generation zone, with the total average overlayed as a black line (left). Annualized total RCV in the Norwegian (NO2) generation zone with the blue bottom bar indicating the propagation effect and the orange top bar indicating the innovation effect (right).}
    \label{fig:propagation_NO_2}
\end{figure}

\begin{figure}
    \centering
    \begin{subfigure}[b]{0.4\textwidth}
        \includegraphics[width=\linewidth]{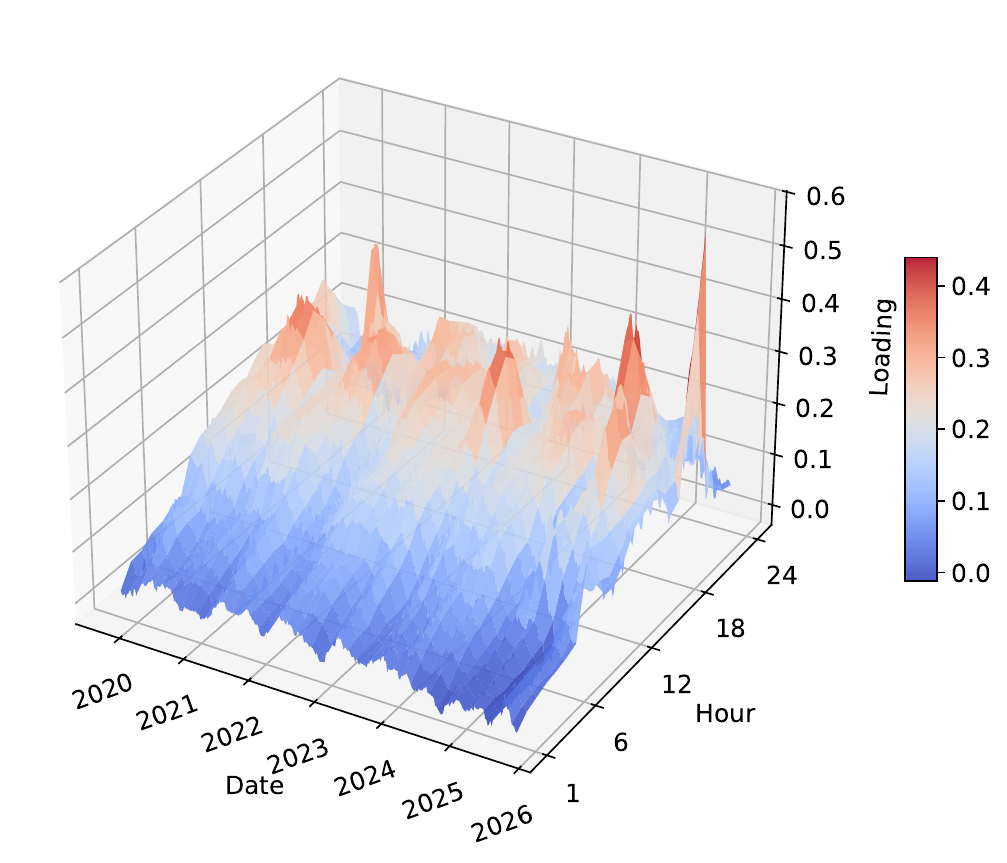}
        \caption{PC1}
    \end{subfigure}
    \hspace{0.04\linewidth}
    \begin{subfigure}[b]{0.4\textwidth}
        \includegraphics[width=\linewidth]{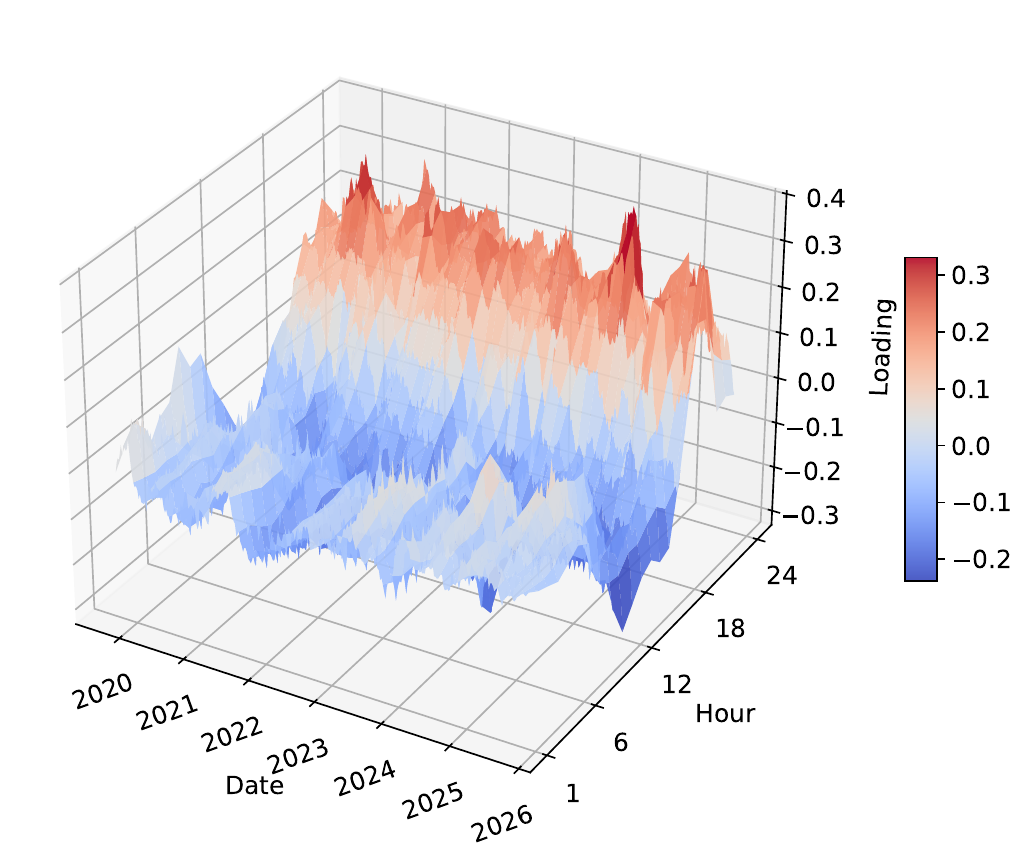}
        \caption{PC2}
    \end{subfigure}

    \vspace{0.3cm}

    \begin{subfigure}[b]{0.4\textwidth}
        \includegraphics[width=\linewidth]{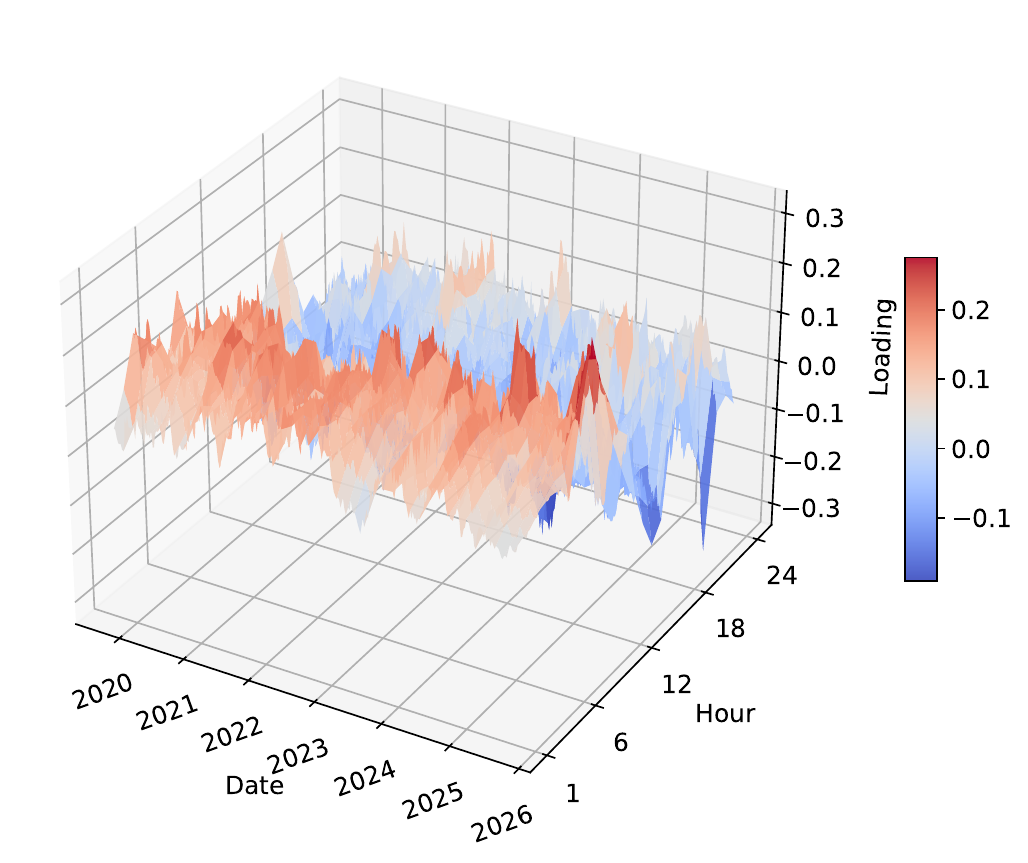}
        \caption{PC3}
    \end{subfigure}
    \hspace{0.04\linewidth}
    \begin{subfigure}[b]{0.4\textwidth}
        \includegraphics[width=\linewidth]{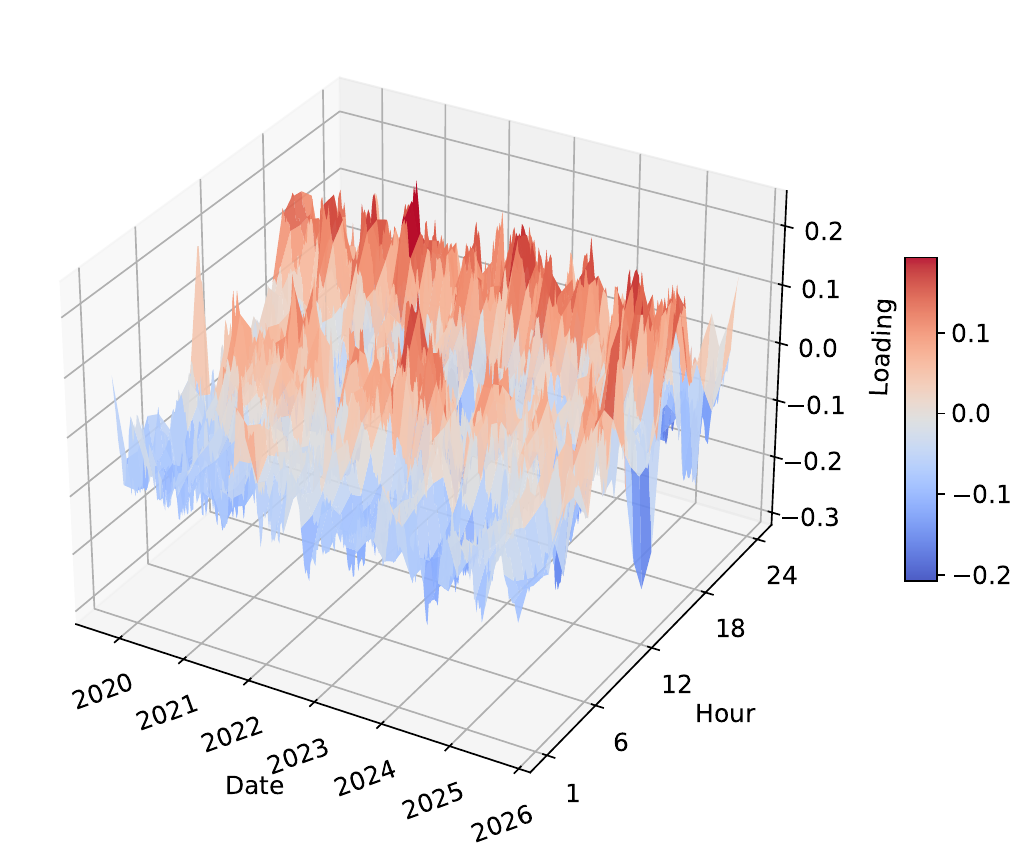}
        \caption{PC4}
    \end{subfigure}
    \caption{Factor loadings over time for the first four principal components in the Spanish generation zone.}\label{fig:loading_stability_ES}
\end{figure}

\begin{figure}
    \centering
    \begin{subfigure}[b]{0.4\textwidth}
        \centering
        \includegraphics[width=\textwidth]{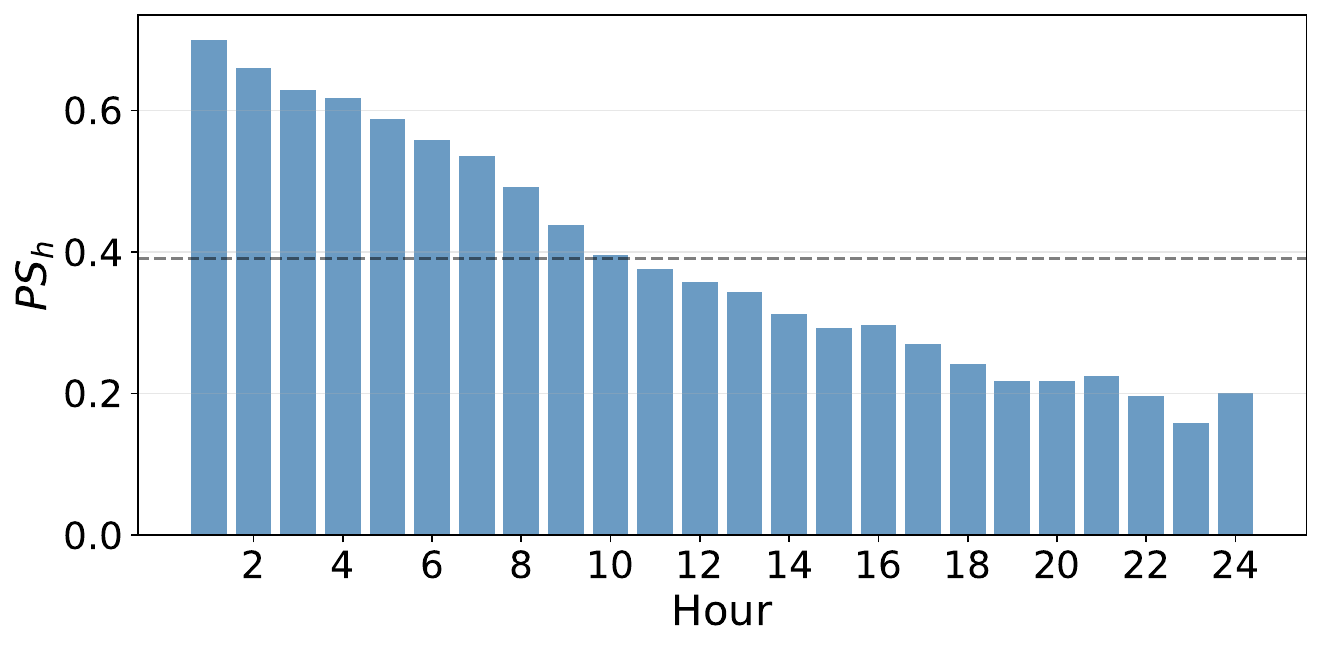}
        \subcaption{Propagation share per hour.}
    \end{subfigure}
    \hspace{0.04\linewidth}
    \begin{subfigure}[b]{0.4\textwidth}
        \centering
        \includegraphics[width=\textwidth]{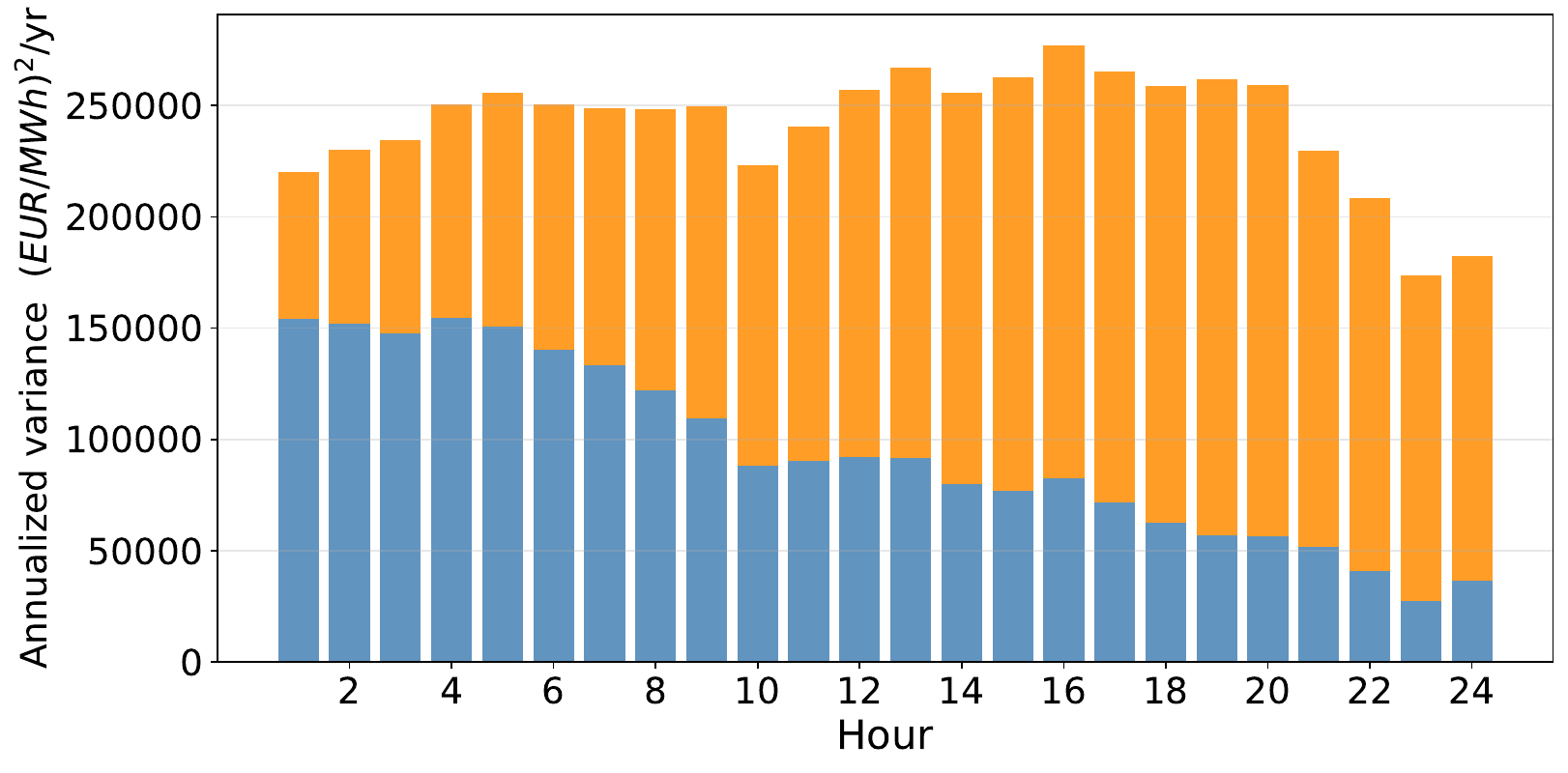}
        \subcaption{Propagation versus innovation.}
    \end{subfigure}
    \caption{Propagation share per hour in the Spanish generation zone, with the total average overlayed as a black line (left). Annualized total RCV in the Spanish generation zone with the blue bottom bar indicating the propagation effect and the orange top bar indicating the innovation effect (right).}
    \label{fig:propagation_ES}
\end{figure}

\begin{figure}
    \centering
    \begin{subfigure}[t]{0.4\textwidth}
        \includegraphics[width=\linewidth]{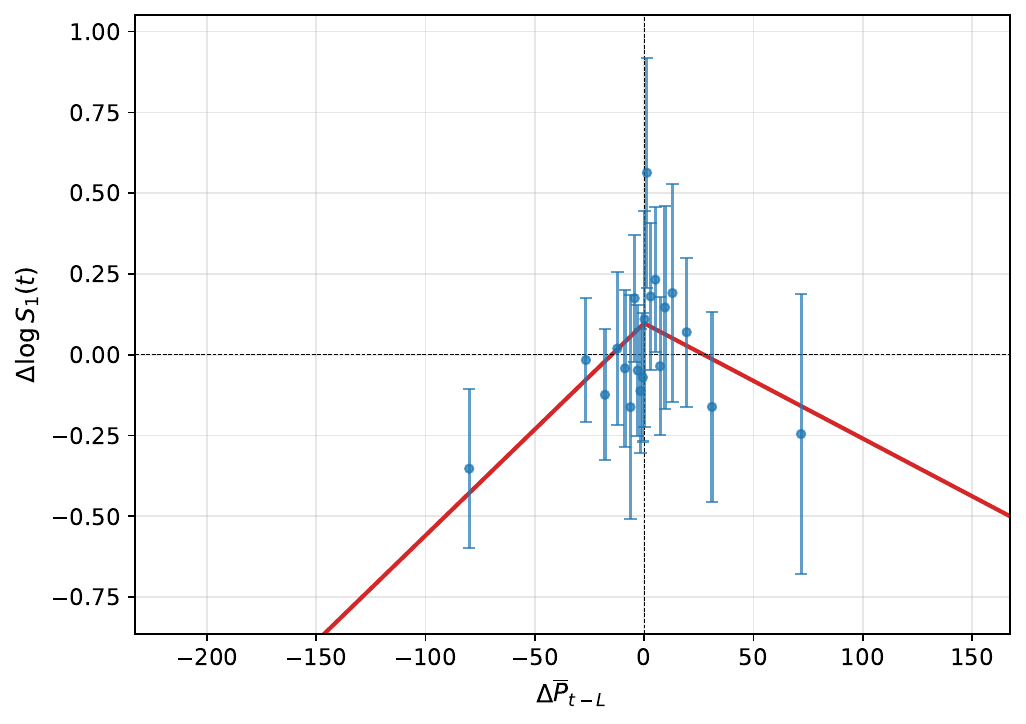}
        \caption{Unconditional regression, Wald test p-value for $\beta_1=\beta_2$ of $p=0.0056$.}
    \end{subfigure}
    \hspace{0.04\linewidth}
    \begin{subfigure}[t]{0.4\textwidth}
        \includegraphics[width=\linewidth]{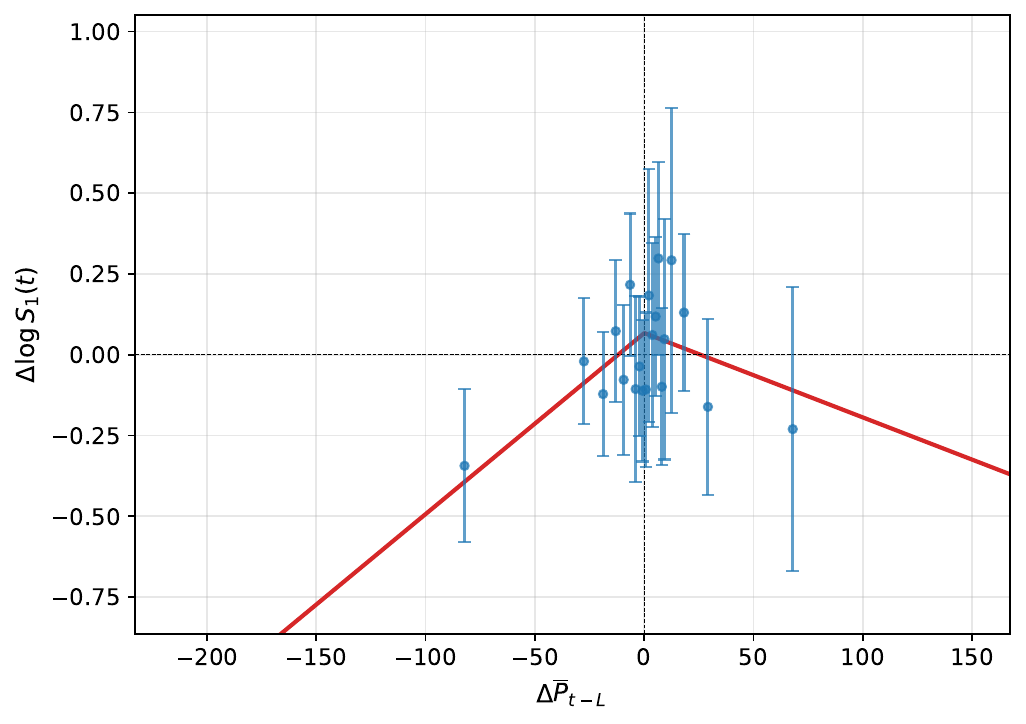}
        \caption{Price/level-scaling control regression, Wald test p-value for $\beta_1=\beta_2$ of $p=0.0310$.}
    \end{subfigure}

    \vspace{0.3cm}

    \begin{subfigure}[t]{0.4\textwidth}
        \includegraphics[width=\linewidth]{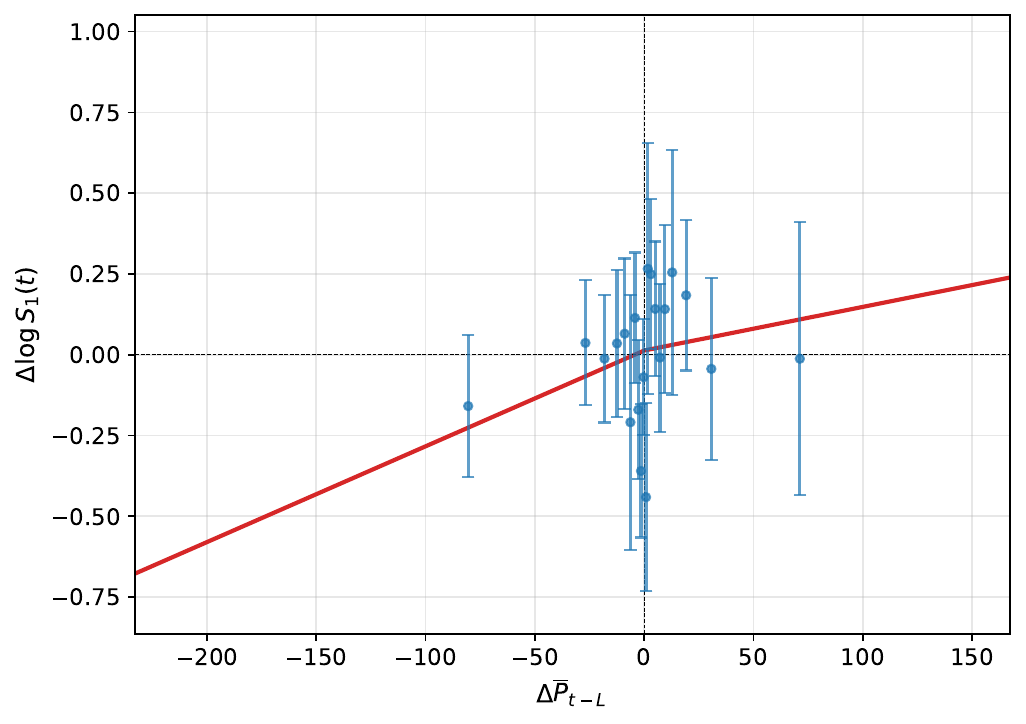}
        \caption{Mean-reversion control regression, Wald test p-value for $\beta_1=\beta_2$ of $p=0.6331$.}
    \end{subfigure}
    \hspace{0.04\linewidth}
    \begin{subfigure}[t]{0.4\textwidth}
        \includegraphics[width=\linewidth]{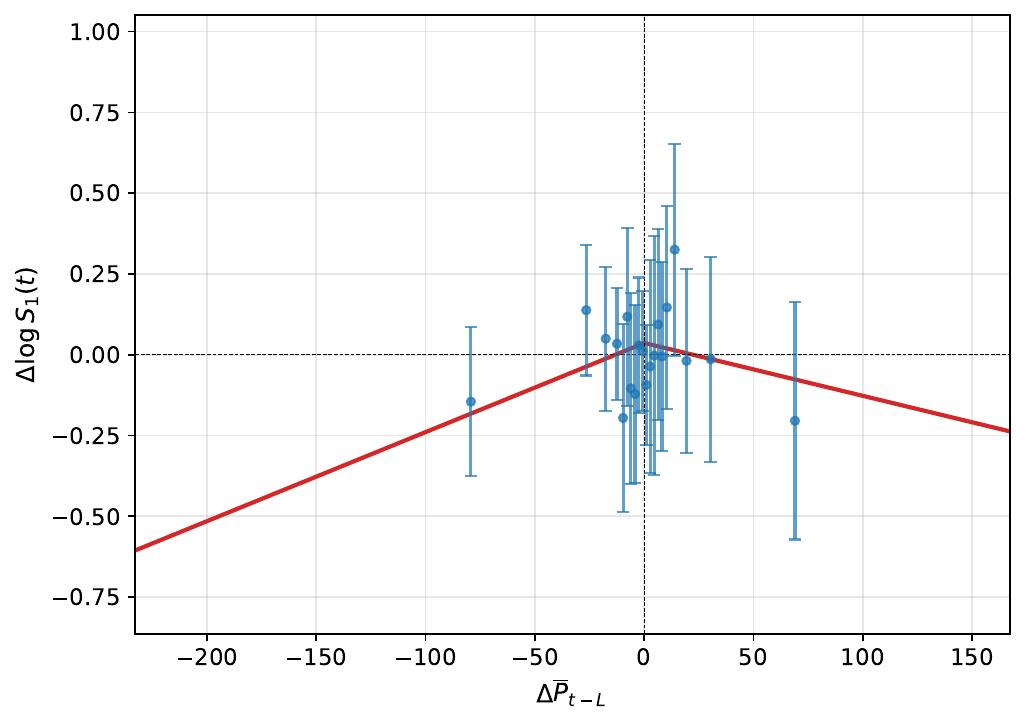}
        \caption{Joint control regression, Wald test p-value for $\beta_1=\beta_2$ of $p=0.1584$.}
    \end{subfigure}
    \caption{Estimated impact of price changes on $\Delta \log(S_1(t))$ in the Norwegian (NO2) generation zone. Blue dots represent the $\mu_b$ coefficients with $95\%$ confidence bands constructed with HAC standard errors. The red line is the regression line.}\label{fig:leverage_nic_NO2}
\end{figure}

\begin{figure}
    \centering
    \begin{subfigure}[t]{0.4\textwidth}
        \includegraphics[width=\linewidth]{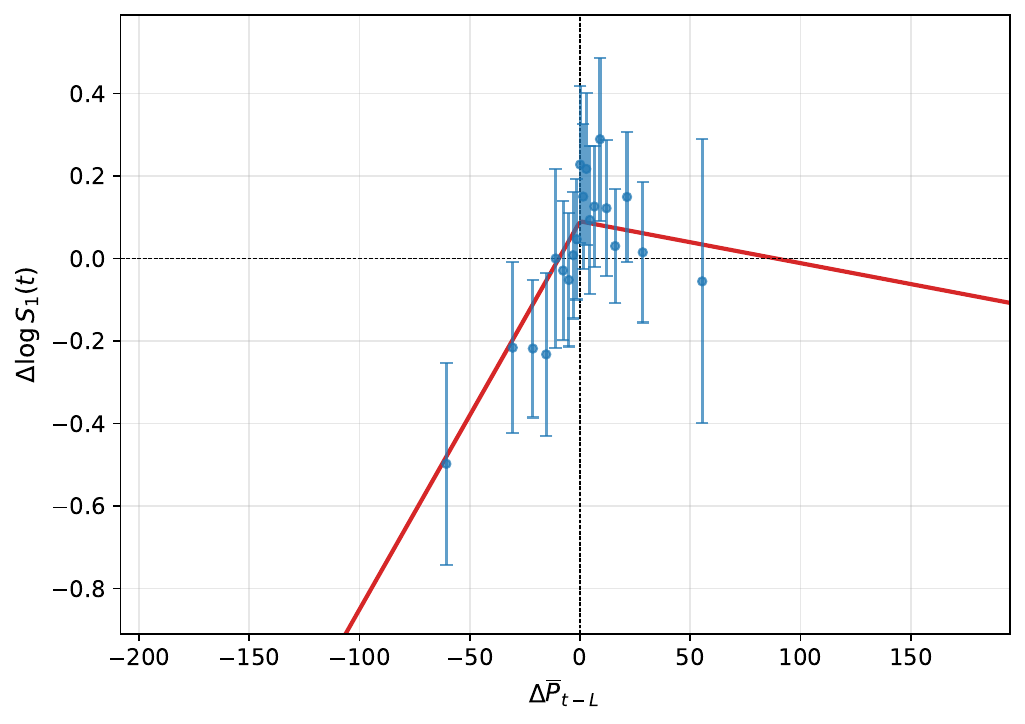}
        \caption{Unconditional regression, Wald test p-value for $\beta_1=\beta_2$ of $p=0.0049$.}
    \end{subfigure}
    \hspace{0.04\linewidth}
    \begin{subfigure}[t]{0.4\textwidth}
        \includegraphics[width=\linewidth]{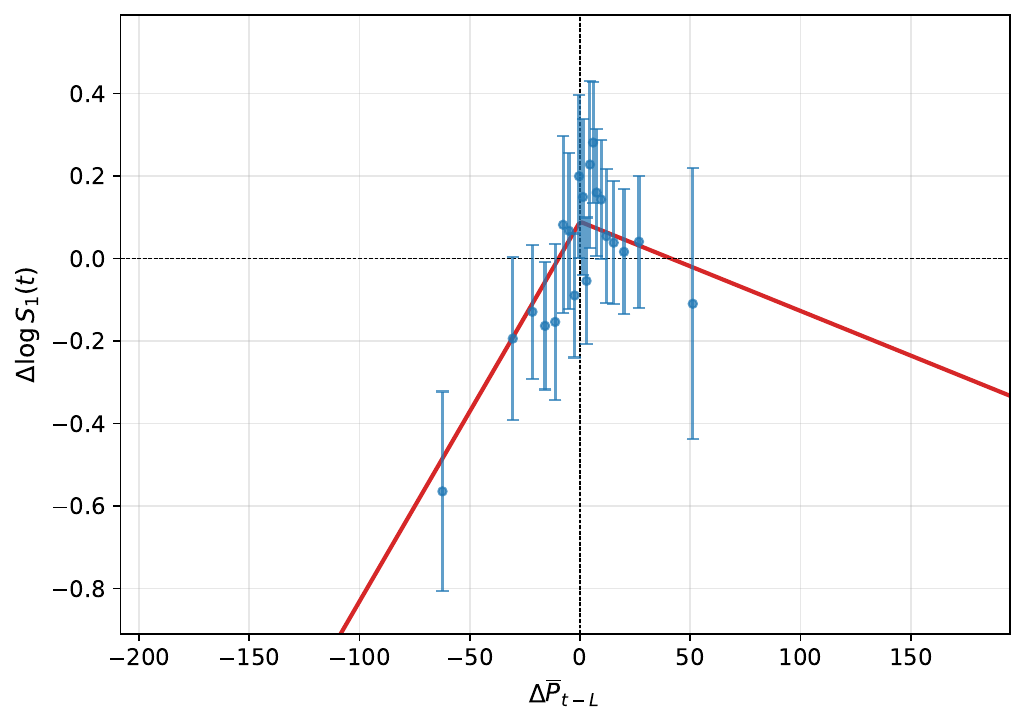}
        \caption{Price/level-scaling control regression, Wald test p-value for $\beta_1=\beta_2$ of $p=0.0030$.}
    \end{subfigure}

    \vspace{0.3cm}

    \begin{subfigure}[t]{0.4\textwidth}
        \includegraphics[width=\linewidth]{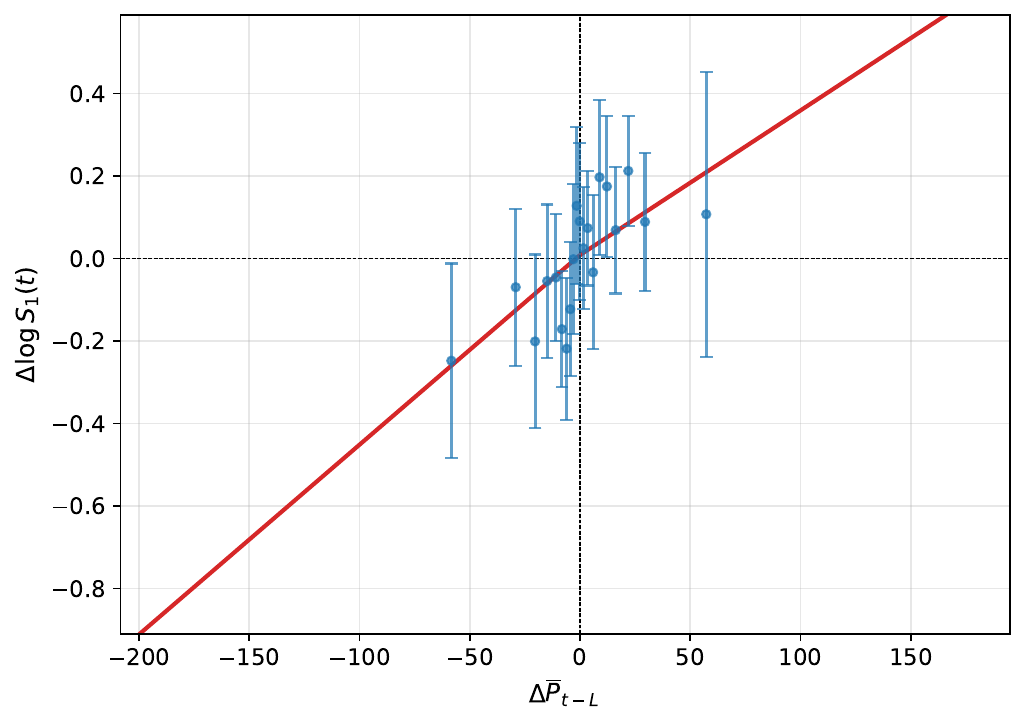}
        \caption{Mean-reversion control regression, Wald test p-value for $\beta_1=\beta_2$ of $p=0.7591$.}
    \end{subfigure}
    \hspace{0.04\linewidth}
    \begin{subfigure}[t]{0.4\textwidth}
        \includegraphics[width=\linewidth]{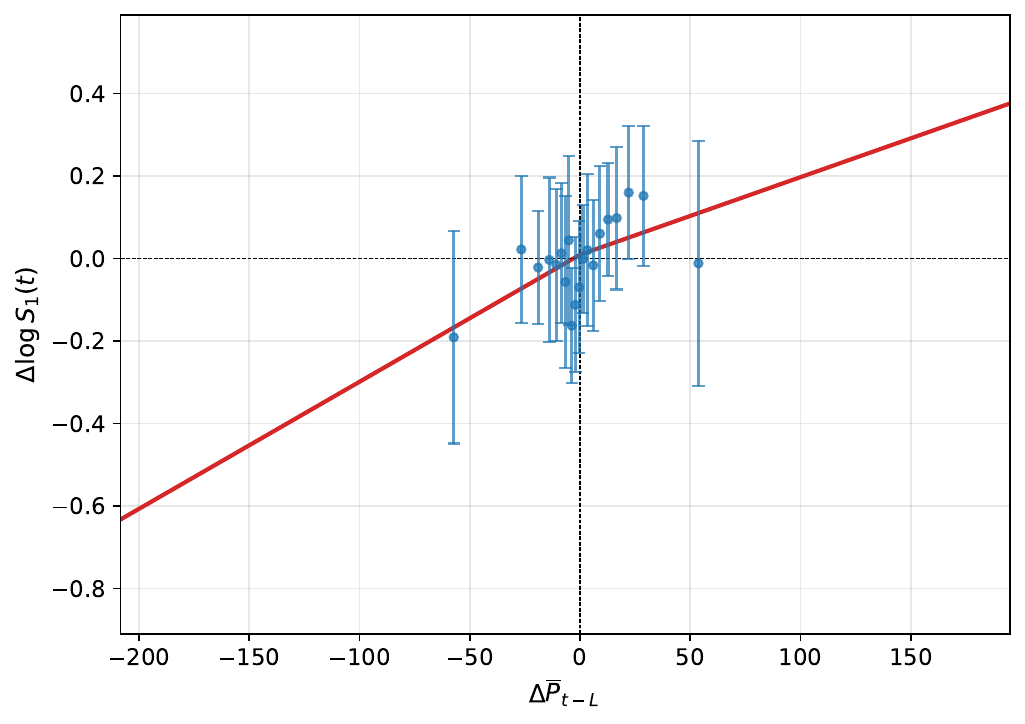}
        \caption{Joint control regression, Wald test p-value for $\beta_1=\beta_2$ of $p=0.7488$.}
    \end{subfigure}
    \caption{Estimated impact of price changes on $\Delta \log(S_1(t))$ in the Spanish generation zone. Blue dots represent the $\mu_b$ coefficients with $95\%$ confidence bands constructed with HAC standard errors. The red line is the regression line.}\label{fig:leverage_nic_ES}
\end{figure}

\end{document}